%% file: main.tex
\documentclass[conference]{IEEEtran}
\ifCLASSOPTIONcompsoc
  \usepackage[nocompress]{cite}
\else
  \usepackage{cite}
\fi
\ifCLASSINFOpdf
\else
\fi
\input{command}

\begin{document}
%

\iffullpaper 
\title{HGMatch: A Match-by-Hyperedge Approach for Subgraph Matching on Hypergraphs \\ (Complete Version)}
\else 
\title{HGMatch: A Match-by-Hyperedge Approach for Subgraph Matching on Hypergraphs}
\fi

\author{Zhengyi Yang$^{\ddagger}$, Wenjie Zhang$^{\ddagger}$, Xuemin Lin$^{\S}$*, Ying Zhang$^{\natural}$, Shunyang Li$^{\ddagger}$%
\vspace{1.6mm}\\
\fontsize{10}{10}
\selectfont\itshape
$^\ddagger$University of New South Wales, $^\S$Shanghai Jiao Tong University, 
$^{\natural}$University of Technology Sydney \\
\footnotesize{*Corresponding Author} \\
\fontsize{9}{9} \selectfont\ttfamily\upshape
$^{\ddagger}$\{zyang,zhangw,sli\}@cse.unsw.edu.au, $^{\S}$xuemin.lin@sjtu.edu.cn, $^{\natural}$ying.zhang@uts.edu.au\\
}


%


\maketitle

\input{chapter/0.abstract}


%
\IEEEpeerreviewmaketitle

\input{chapter/1.introduction}

\input{chapter/2.related_work}
\input{chapter/3.background}

\input{chapter/4.overview}

\input{chapter/5.matching_by_edges}

\input{chapter/6.parallel_execution}
\input{chapter/7.experiment}
\input{chapter/8.conclusion}

\iffullpaper 
\else

\section*{ACKNOWLEDGMENT}
Wenjie Zhang’s research is supported by the Australian Research Council Future Fellowship (FT210100303).

\fi

\newpage
\balance





%


\bibliographystyle{abbrv}
\bibliography{bib/references,bib/huge_sigmod}

\end{document}

%% file: command.tex
\usepackage{xspace}
\usepackage{xcolor}
\usepackage{graphicx}
\usepackage{amsmath}
\usepackage{amsthm}
\usepackage{algorithmic}
\usepackage[aboveskip=-1pt,belowskip=-2pt]{subcaption}
\usepackage{enumitem,kantlipsum}
\usepackage{multirow}
\usepackage{url}
\usepackage[linesnumbered,ruled,vlined]{algorithm2e}
\usepackage{setspace}
\usepackage{amsfonts}
\usepackage{balance}
\usepackage{float}
\usepackage{soul}

\newif\iffullpaper

\fullpapertrue

\iffullpaper 
\else 

\usepackage{titlesec}
\titlespacing*{\section}{0pt}{1pt}{1pt}
\fi

\DeclareMathOperator{\argmin}{argmin}

\newcommand{\overbar}[1]{\mkern 1.5mu\overline{\mkern-1.5mu#1\mkern-1.5mu}\mkern 1.5mu}

\newcommand{\red}[1]{{#1}}

\newcommand{\hm}{\kw{HGMatch}}
\newcommand{\sm}{subhypergraph matching\xspace}
\newcommand{\Sm}{Subhypergraph matching\xspace}

\newcommand{\cfl}{\kw{CFL}}
\newcommand{\ceci}{\kw{CECI}}
\newcommand{\daf}{\kw{DAF}}
\newcommand{\rapid}{\kw{RapidMatch}}

\newcommand{\cflh}{\kw{CFL}-\kw{H}}
\newcommand{\cecih}{\kw{CECI}-\kw{H}}
\newcommand{\dafh}{\kw{DAF}-\kw{H}}

\newcommand{\scan}{\kw{SCAN}}
\newcommand{\expand}{\kw{EXPAND}}
\newcommand{\sink}{\kw{SINK}}

\newcommand{\tscan}{$T_{\scan}$\xspace}
\newcommand{\texpand}{$T_{\expand}$\xspace}
\newcommand{\tsink}{$T_{\sink}$\xspace}

\SetKwInOut{KwIn}{Input}
\SetKwInOut{KwOut}{Output}

\newcommand{\HC}{\kw{HC}}
\newcommand{\MA}{\kw{MA}}
\newcommand{\CH}{\kw{CH}}
\newcommand{\CP}{\kw{CP}}
\newcommand{\SB}{\kw{SB}}
\newcommand{\HB}{\kw{HB}}
\newcommand{\WT}{\kw{WT}}
\newcommand{\TC}{\kw{TC}}
\newcommand{\SA}{\kw{SA}}
\newcommand{\AR}{\kw{AR}}
\newcommand{\JF}{\kw{JF17K}}

\newtheorem{theorem}{Theorem}[section]
\newtheorem{lemma}[theorem]{Lemma}

\theoremstyle{definition}
\newtheorem{definition}{Definition}[section]

\theoremstyle{remark}
\newtheorem{example}{Example}[section]

\theoremstyle{remark}
\newtheorem{observation}{Observation}[section]

\theoremstyle{remark}

\theoremstyle{definition}
\newtheorem*{remark}{Remark}

\newcommand{\stitle}[1]{\vspace{0.5ex} \noindent{{\bf #1}}}

\newcommand{\sstitle}[1]{\vspace{0.5ex} \noindent{\textit{ #1}}}

\newcommand{\kw}[1]{{\ensuremath {\mathsf{#1}}}\xspace}

\long\def\comment#1{}

\newcommand{\reffig}[1]{Fig.~\ref{fig:#1}}
\newcommand{\refsec}[1]{Section~\ref{sec:#1}}
\newcommand{\reftable}[1]{Table~\ref{tab:#1}}
\newcommand{\refalg}[1]{Algorithm~\ref{alg:#1}}

\newcommand{\refdef}[1]{Definiton~\ref{def:#1}}
\newcommand{\refthm}[1]{Theorem~\ref{thm:#1}}
\newcommand{\reflem}[1]{Lemma~\ref{lem:#1}}

%% file: chapter/0.abstract.tex
\begin{abstract}

Hypergraphs are a generalisation of graphs in which a hyperedge can connect any number of vertices.
It can describe n-ary relationships and high-order information among entities compared to conventional graphs.
In this paper, we study the fundamental problem of subgraph matching on hypergraphs (i.e. \sm).
Existing methods directly extend subgraph matching algorithms to the case of hypergraphs.
However, this approach delays hyperedge verification and underutilises the high-order information in hypergraphs, which leads to large search space and high enumeration costs.
%
Furthermore, with the growing size of hypergraphs, it is becoming hard to compute \sm sequentially. 
Thus, we propose an efficient and parallel \sm system, \hm, to handle \sm in massive hypergraphs.
We propose a novel match-by-hyperedge framework to utilise high-order information in hypergraphs and use set operations for efficient candidate generation.
%
Moreover, we develop an optimised parallel execution engine in \hm based on the dataflow model, which features a task-based scheduler and fine-grained dynamic work stealing to achieve bounded memory execution and better load balancing.
Experimental evaluation on $10$ real-world datasets shows that \hm outperforms the extended version of the state-of-the-art subgraph matching algorithms ({\footnotesize \cfl, \daf, \ceci, and \rapid}) by orders of magnitude when using a single thread, and achieves almost linear scalability when the number of threads increases.
\end{abstract}

%% file: chapter/1.introduction.tex
\section{Introduction}\label{sec:introduction}

Hypergraphs are a generalisation of graphs in which a hyperedge can connect any number of vertices. In contrast, an edge in conventional graphs (i.e., pairwise graphs) connects exactly two vertices. Compared to conventional graphs, hypergraphs can describe $n$-ary relationships among entities. Thus, hypergraphs are able to capture high-order interactions among multiple entities that are not directly expressible in conventional graphs.
%
%

One of the most fundamental problems in hypergraphs is \textit{\sm}. Specifically, given a query hypergraph and a data hypergraph, subhypergraph matching aims to find all subhypergraphs of the data hypergraph that are isomorphic to the query hypergraph. For example, given a sample query hypergraph $q$ and a data hypergraph $H$ in \reffig{hmatch-example}. \Sm finds two embeddings, which contains hyperedges $\{e_1,e_3,e_5\}$ and $\{e_2,e_4,e_6\}$ in $H$. 

\begin{figure}[t]
    \centering
    \begin{subfigure}[b]{.36\columnwidth}
        \centering
        \includegraphics[width=\textwidth]{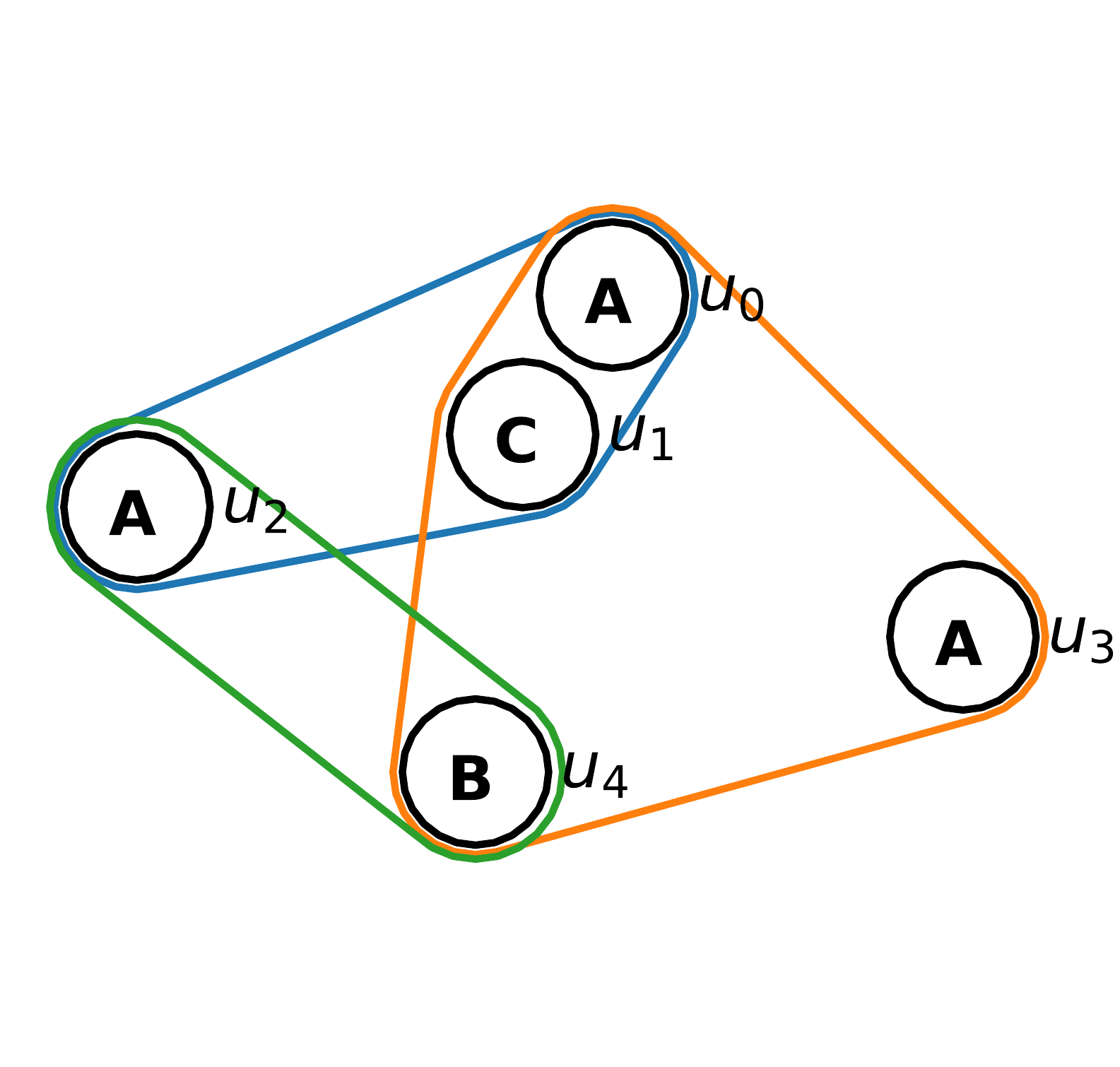}
        \caption{Query hypergraph ($q$)}
        \label{fig:query}
    \end{subfigure} %
    ~
    \begin{subfigure}[b]{.45\columnwidth}
        \centering
        \includegraphics[width=\textwidth]{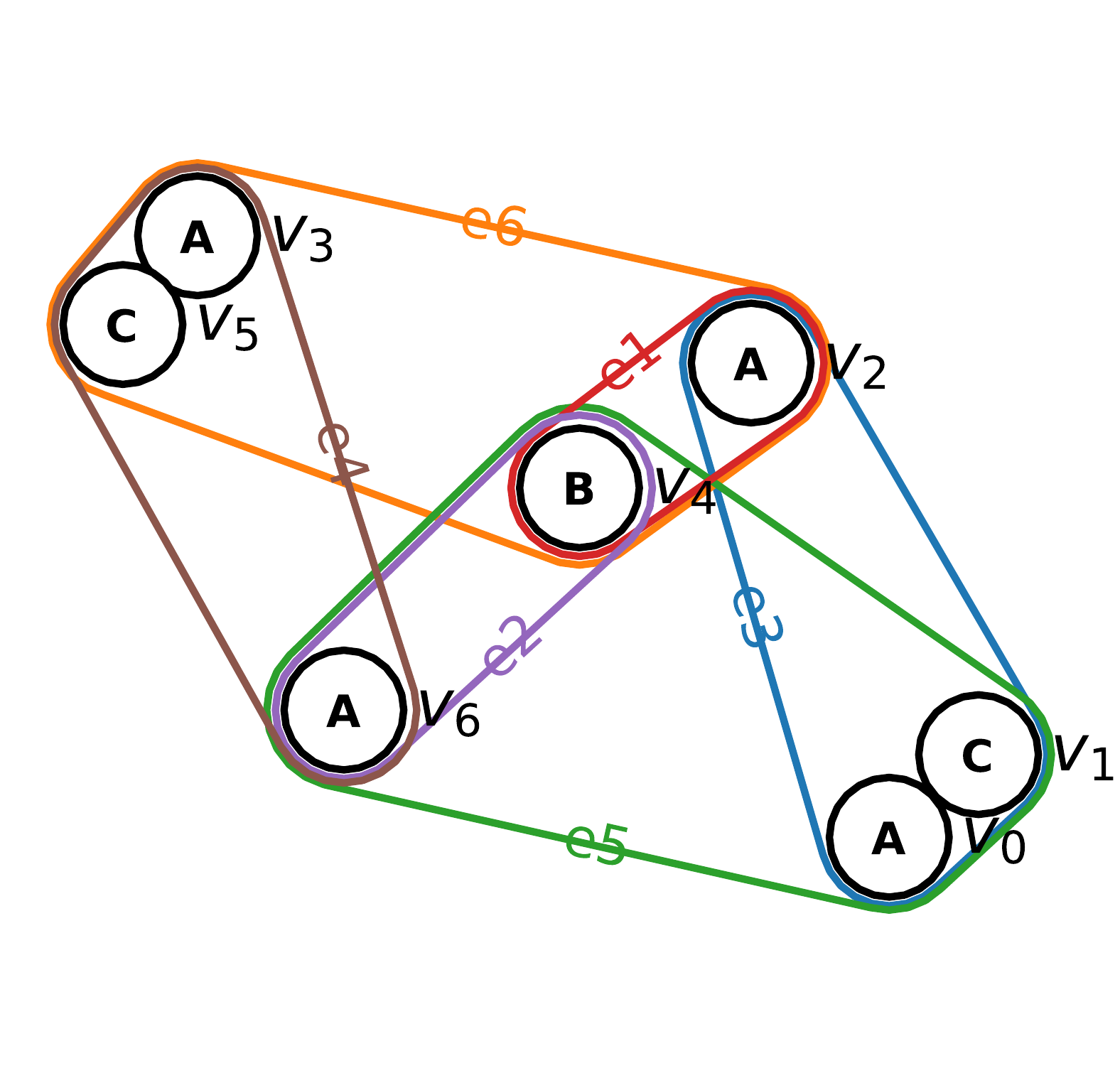}
        \caption{Data hypergraph ($H$)}
        \label{fig:data}
    \end{subfigure}
    \caption{ Example of Subhypergraph Matching. We draw a “rubber band”-like hypergraph, where convex hulls represent hyperedges and are drawn around the vertices they contain. `A', `B' and `C' represent vertex labels. $v_i$ and $e_i$ represent vertex and hyperedge IDs, respectively.}
    \label{fig:hmatch-example}
\end{figure}

\stitle{Applications} \Sm is associated with many applications in the real world. We list four representative examples in different domains. 
Other applications include object detection in computer vision \cite{img-hyper1,img-hyper2,img-hyper3}, complex pattern search in collaborative network \cite{coauthor-hyper,Yu2016}, finding user alignment in social networks \cite{hyper-app-social}, and image retrieval \cite{img-re}.

\sstitle{\underline{Mining Biological Networks.}}
Representing biological networks as graphs is a powerful approach to reveal underlying patterns and critical components from biomolecular data. However, conventional graphs do not natively capture the multi‐way relationships present among genes and proteins in biological systems \cite{bio-hyper-gene2}.
Therefore, hypergraphs are usually used to represent complex high-order relationships in bioinformatics.
For example, protein interactions can be represented as hypergraphs where proteins are vertices and complexes are hyperedges \cite{bio-hyper-ppi1,bio-hyper-ppi2}. Similarly, gene transcriptomic expression data can also be represented as hypergraphs, with each gene modelled as a hyperedge and each condition as a vertex \cite{bio-hyper-gene1,bio-hyper-gene2,bio-hyper-gene3}.
By adopting the hypergraph model, biologists can represent protein/gene complexes of their interests as query hypergraphs and find them in large biological networks to understand their interactions and roles.

\sstitle{\underline{Querying Hypergraph Databases.}}
In recent years, there emerges a number of hypergraph databases, including AtomSpace \cite{AtomSpace}, HyperGraphDB \cite{HyperGraphDB}, and TypeDB \cite{TypeDB}, which rely on \sm as a basic building block.
A hypergraph database can effectively model the multi-way relations among many real-world entities. 
Many hypergraph databases are developed and deployed, especially in the area of artificial general intelligence (AGI) \cite{AGI} (e.g., the OpenCog project \cite{opencog}).
%
%
Specifically, AGI systems require a knowledge representation database to support critical reasoning.
In a typical AGI application, knowledge is stored in the form of hypergraphs \cite{opencog-graph}, where each vertex and hyperedge represents an \textit{Atom} \cite{hodges1997shorter} with a certain type.
A pattern matcher, which performs \sm, is used to search for specific patterns in the hypergraphs.
After specifying some arrangements of atoms (i.e., a query hypergraph), the pattern matcher will find all instances of that hypergraph in the atom space (i.e. a data hypergraph).
The matched results can then be sent to a rule engine \cite{watkin2017introduction,baader1999term} for further reasoning.

\sstitle{\underline{Pattern Learning in NLP.}}
Hypergraphs are also increasingly popular in machine leaning \cite{hyper-app-learn,hyper-app-learn3,hyper-app-learn4,hyper-app-ml} and natural language processing (NLP) \cite{parsing_hyper,hyper-app-text,hyper-app-nlp}.
The authors of \cite{hyper-app-nlp} propose the concept of \textit{semantic hypergraphs} where each word is a vertex, and each valid sentence is a hyperedge. Semantic hypergraphs can be constructed by parsing large corpus using modern machine learning techniques in NLP.
During the process of pattern learning, some sentences are first selected from a given training corpus. It can be drawn at random or by any other criterion adapted to the pattern-learning task at hand. The selected sentences are inferred and transformed into a hypergraph query. \Sm is then performed in the semantic hypergraph to find matched embeddings. Finally, the embeddings are presented to humans for validation of the corresponding learning tasks. The process repeats with a human-refined query hypergraph if no valid embeddings are found.

\sstitle{\underline{\red{Q/A over Hypergraph Knowledge Base.}}}
\red{
It is observed in \cite{Wen2016OnTR} that more than 33\% of the entities participate in non-binary relations in the knowledge base Freebase \cite{Freebase}, and further observed in \cite{ijcai2020p303} that 61\% of the entities participate in non-binary relations. 
Question answering (Q/A) allows users to query real-world questions over the knowledge base. By representing the knowledge base in hypergraphs, it allows us to better express and explore the massive non-binary relations in the knowledge base, where the evaluation of queries can be performed using \sm \cite{qa-rdf}.
We present a case study of in \refsec{case_study}.
}


%

\stitle{Motivations.} 
Subgraph matching in conventional graphs has been extensively studied in the literature.
Existing algorithms of subgraph matching \cite{turbo-iso,cfl,daf,ceci,rapidmatch,jin2021fast,match-survey,graph-ql,vf2,quicksi} primarily works on better matching orders, pruning rules, index structures, and enumeration methods, to improve the efficiency.
However, \sm in hypergraphs has attracted little attention despite its emerging applications, as mentioned above.

\begin{figure}[t]
    \centering
    \includegraphics[width=0.3\textwidth]{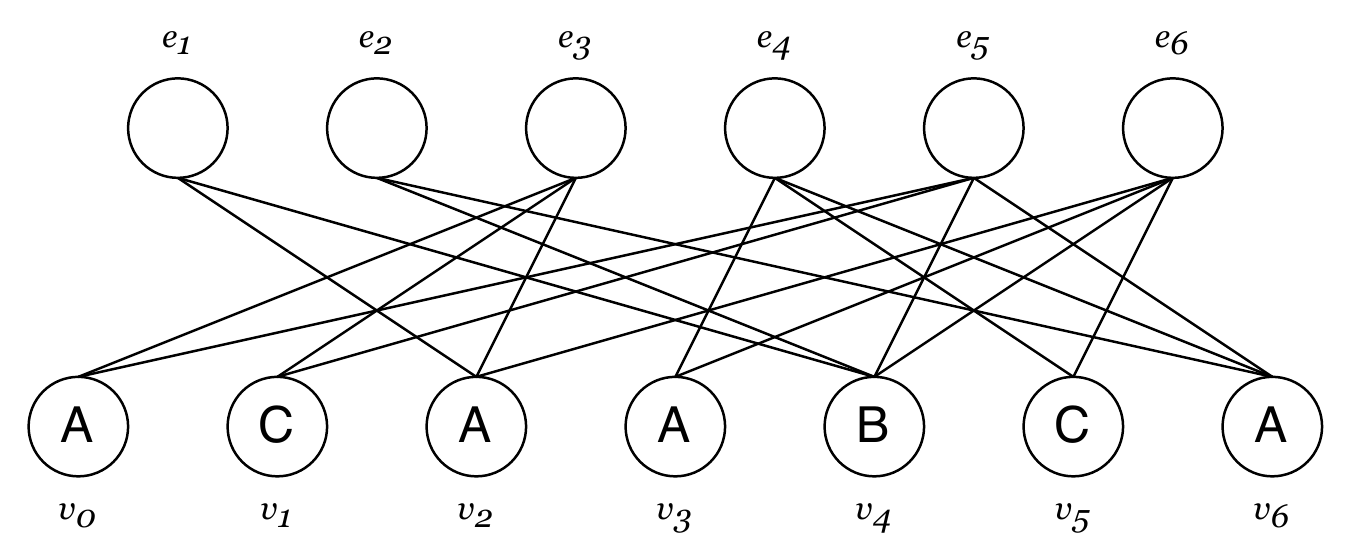}
    \caption{Converted Bipartite Graph (Data Hypergraph in \reffig{data}).}
    \label{fig:bi-graph}
\end{figure}

One straightforward approach of \sm is to convert the hypergraph to a \emph{bipartite graph} by treating the hyperedges as the vertices \cite{bio-hyper-ppi2}. 
An example bipartite graph of the data hypergraph in \reffig{data} is given in \reffig{bi-graph}, where the upper vertices refer to hyperedges in the original hypergraph, the lower vertices refer to vertices in the original hypergraph, and the edges refer to the connectivity of hyperedges.
After converting both the query and data hypergraphs into bipartite graphs, conventional subgraph matching algorithms can be applied to find embeddings of the original hypergraph.
However, this strawman approach will significantly inflate the size of the graphs. For example, a hypergraph with $2$ million vertices and $15$ million hyperedges will result in a bipartite graph with $17$ million vertices and $1$ billion edges \cite{communities-graph}.
Due to the NP-hardness \cite{np-complete} nature of subgraph matching, it is hard to compute embeddings on such inflated graphs \cite{hyperx}.

Another approach is to directly extend existing subgraph matching algorithms to the case of hypergraphs. Among the existing algorithms, \cite{sm05,sm08} extends Ullmann's backtracking algorithm \cite{ullmann}.
\cite{hyper-iso,subhymatch} also follow the same framework with more filtering rules derived from hypergraph features to improve efficiency.
As most state-of-the-art subgraph matching algorithms (e.g., \cite{cfl,daf,ceci}) follows Ullmann's backtracking framework, such an extension can be orthogonally applied to them as well. 
Specifically, they recursively expand partial embedding vertex-by-vertex by mapping a query vertex to a data vertex at each step to enumerate all results following a given matching order and backtrack when necessary.
We denote this framework as the \emph{match-by-vertex} approach.
Hyperedges are used as a verification condition in the match-by-vertex framework, just like the edges in subgraph matching.
%
However, treating hyperedges simply as verification conditions delays the hyperedge verification and underutilises the high-order information in hypergraphs, which can lead to a huge search space and large enumeration cost.

In addition, with the rapid growth of hypergraph data these days, it is becoming difficult to compute \sm on massive hypergraphs using sequential algorithms. 
However, none of the existing \sm algorithms supports parallel execution.
For example, in the real-world hypergraph of Amazon Reviews with more than $4$ million hyperedges (\AR in \reftable{datasets}), none of the existing sequential algorithms in our experiments are able to compute queries of three hyperedges within a one-hour time limit. 
%
%
Traditional backtracking framework in subgraph matching adopts the depth-first search (DFS), which is generally hard to parallelise.
%
Distributed solutions of subgraph matching \cite{twin-twig,seed,wco-join,crystaljoin,multiway-join,star-join}, on the other hand, adopt the breath-first search (BFS) in a cluster of machines for high CPU utilisation.
But this can often lead to high memory consumption, network communications, and economic cost \cite{cost}. 

Motivated by the above reasons, we study the problem of \sm to develop an efficient and parallel solution on a single machine in this paper.



\stitle{Challenges.} 
We summarise key challenges as follows.

\begin{enumerate} [leftmargin=*]
    \item \textit{\ul{How to effectively utilise high-order information in hypergraphs?}}
    Hypergraphs contain n-ary relationships in hyperedges that are typically not presented in conventional graphs. 
    Hence, it is crucial to fully utilise the high-order information in hyperedges during matching to reduce search space and speed up enumeration.
    
    \item \textit{\ul{How to efficiently enumerate all embeddings in parallel?}}
    To improve the performance of \sm on a single machine, it is important to fully utilise the ever-developing hardware (i.e., multi-core) while managing memory consumption well.
    Furthermore, the challenge arises from the power law nature of the real-world graphs \cite{power-law-1,power-law-2} to handle workload disparity among different workers while parallelising.

\end{enumerate}

\stitle{Our Solution and Contributions.} 
To address these challenges, we develop \hm, an efficient and parallel sub\underline{H}yper\underline{G}raph \underline{Match}ing engine on a single machine. Instead of matching the query hypergraph vertex-by-vertex as in the match-by-vertex framework used by existing subgraph matching and \sm algorithms, we propose a \emph{match-by-hyperedge} framework to match the query by hyperedges to fully utilise the n-ary relationships. 
%
%
%
Specifically, we made the following contributions.

\begin{enumerate} [leftmargin=*]
    \item \textit{\underline{A match-by-hyperedge framework.}} 
    We propose to match the query hypergraph by hyperedges instead of vertices. 
    \hm expands a partial embedding by one new hyperedge at a time.
    %
    In this way, \hm is able to fully utilise the high-order information in hypergraphs to reduce search space and avoid redundant computation of enumerating matchings of vertices. 
    We store the data hypergraph in multiple tables with different \emph{hyperedge signatures} (i.e., a multiset\footnote{A multiset (i.e., bag) is a set that allows for multiple instances for each of its elements.} of vertex labels contained in a hyperedge).
    A lightweight inverted hyperedge index is then built for each table to speed up the retrieve process of incident hyperedges of a given vertex. 
    By doing so, \hm is able to generate candidate hyperedges directly using set operations (i.e., difference, union and intersections), which can be implemented very efficiently on modern hardware \cite{simd_gallop,qfilter,simd-intersection,bitmap2}.
    %
    Apart from that, we use set comparison to remove false positives during enumeration, which completely avoids expensive recursive calls in traditional backtracking-based enumeration methods.

    \item \textit{\underline{A highly optimised parallel execution engine.}} 
    Thanks to the above-mentioned design, \hm does not incur any recursive calls or build any auxiliary structures during runtime, which makes it easy to be parallelised.
    %
    %
    We adopt the dataflow model \cite{dataflow,dataflow-def} for parallel execution in \hm, which has been employed in many recent subgraph matching solutions \cite{huge,patmat-exp,wco-join,graphflow-demo}.
    %
    %
    %
    To bound memory consumption while keeping a high degree of parallelism, we design a task-based scheduler in \hm.
    With the scheduler, we prove that \hm achieves a tight memory bound of $O(\overbar{a_q}\times|E(q)|^2\times|E(H)|)$ for \sm, where $\overbar{a_q}$ is the average arity (i.e., hyperedge size) of query, $|E(q)|$ and $|E(H)|$ are the number of query and data hyperedges, respectively.
    Furthermore, the dynamic work-stealing mechanism \cite{work_stealing1,work_stealing2,cilk} is employed for fine-grained load balancing.
    
    \item \textit{\underline{In-depth experiments using real-world datasets.}}
    We conducted extensive experiments on $10$ real-world datasets. Results show the efficiency and scalability of \hm.
    Comparing with the extended version of three state-of-the-art subgraph matching algorithms \cfl \cite{cfl}, \daf \cite{daf}, \ceci \cite{ceci}, and \rapid \cite{rapidmatch}, \hm achieves an average speedup of more than $5$ orders of magnitude.
    When using multi threads, \hm achieves almost linear scalability when increasing the number of threads with near-perfect load balancing.
    Besides, \hm is the only algorithm that is able to complete all queries within the time limit.
    
\end{enumerate}

\stitle{Paper Organization.} The rest of this paper is organized as follows. \refsec{related_work} discusses related work. \refsec{background} introduces problem definition and background. In \refsec{overview}, we present the workflow and hypergraph storage of \hm. We introduce our match-by-hypergraph framework in \refsec{matching} and the design of our parallel execution engine in \refsec{execution}, respectively. Experimental evaluation and case study are presented in \refsec{experiment}, followed by conclusion in \refsec{conclusion}.

%% file: chapter/2.related_work.tex
\section{Related Work}\label{sec:related_work}

\stitle{Sequential Subgraph Matching.} 
In sequential subgraph matching, the study was initiated by Ullmann’s backtracking algorithm \cite{ullmann}, which recursively matches query vertices to data vertices following a given matching order and backtracks when necessary. The state-of-the-art algorithms \cite{turbo-iso,quicksi,cfl,daf,ceci,graph-ql,vf2} mostly follow the backtracking framework with different matching orders, pruning rules, auxiliary data structures, and enumeration methods. 
In particular, Turbo$_{ISO}$ \cite{turbo-iso} proposes to compress the query graph by merging equivalent query vertices with a \textit{CR} auxiliary data structure to speed up the computation. CFL \cite{cfl} proposes the core-forest-leaf decomposition to reduce redundant Cartesian products and introduces a more compact auxiliary structure \textit{CPI} to solve the exponential size of \textit{CR}. DAF \cite{daf} designs a new auxiliary structure \textit{CS} based on the directed acyclic graph (DAG) of the query graph and uses pruning by failing sets to increase the pruning power.
Note that \hm does not build any auxiliary structure during runtime, leading to low memory cost and easy parallelism.
A more recent solution, RapidMatch \cite{rapidmatch} integrates the backtracking-base exploration method with worst-case optimal join.
Surveys and experiments of representative algorithms have been conducted in \cite{comparison,match-survey}.

\stitle{Parallel Subgraph Matching.} 
PGX.ISO \cite{PGX.ISO} runs subgraph matching in breath-first search (BFS) to carry out parallel execution. As a result, it requires materialising all intermediate results at each step, which leads to exponential memory consumption. Moreover, load balancing is achieved by copying intermediate results into single global storage, which leads to further cost of copying and synchronization.
PSM \cite{PSM} proposes a generic framework for parallelising recursive backtracking-based subgraph matching algorithms. It represents the problem as a tree search in the state space and different matching algorithms as different orders in the search.
%
CECI \cite{ceci} uses embedding clusters as the auxiliary structure to better accommodate parallel processing with improved load-balancing.
EmptyHeaded \cite{empty-headed} and Graphflow \cite{graphflow} employ the multi-way join model to compute subgraph matching using join operations in parallel.
%
%
%
Due to the increasing size of graphs, distributed algorithms have also been widely studied \cite{psgl,twin-twig,seed,wco-join,patmat-exp,huge,crystaljoin,multiway-join,star-join,BENU,RADS} in recent years, which typically focus on unlabeled graphs.

\stitle{Subhypergraph Matching} 
Compared with subgraph matching, the research on subhypergraph matching algorithms is rather limited.
\cite{sm05,sm08} extend the framework of Ullmann’s backtracking algorithm to the case of hypergraphs. 
It utilizes hyperedges only for verification.
The rule for adding a new vertex is to verify whether there are corresponding hyperedges between it and the matched vertices. 
\cite{Yu2016} utilises the index-filter-verification framework in unlabeled hypergraphs to support similarity \sm.
\cite{subhymatch} works on \sm in unlabeled hypergraphs.
Since most hypergraphs in real-world applications are labelled, we only focus on exact \sm in labelled hypergraphs in this paper.
\cite{hyper-iso} proposes an incident hyperedge structure (IHS) filter for candidate vertex filtering to reduce the number of candidates as we will introduce in the next section.
In addition, none of the existing \sm algorithms supports parallel execution, which significantly limits their performance in large hypergraphs.

\stitle{Hypergraph Databases.} 
A hypergraph database is essential for building AGI \cite{AGI} systems and highly complex, large-scale knowledge representation (KR) systems.
AtomSpace \cite{AtomSpace} is an in-memory KR database for the Open Cognition (OpenCog) project \cite{opencog}. It adopts a generalised hypergraph (i.e., metagraph) model \cite{metagraph} and features a query engine as well as a rule-driven inferencing engine to perform reasoning. 
HypergraphDB \cite{HyperGraphDB} is an embedded and transactional graph database based on hypergraphs. 
However, the support of hypergraph pattern queries (i.e., \sm) is still on-progress.
TypeDB \cite{TypeDB} (formerly known as GRAKN \cite{grakn}) is a strongly-typed database with a rich and logical type system. It also embraces the hypergraph data model, which allows the users to model their domain based on logical and object-oriented principles. 



%% file: chapter/3.background.tex
\section{Background}\label{sec:background}

In this section, we present the formal problem definition and how to extend generic subgraph matching framework to \sm.

\subsection{Problem Definition}

We focus on undirected, connected and vertex-labelled\footnote{Our techniques can be easily applied to edge-labelled hypergraphs as well by adding additional constraints of hyperedge labels.} simple hypergraphs\ in this paper.

\begin{definition}
\textbf{(Hypergraph).} A hypergraph $H$ is defined as a tuple $H=(V,E,l,\Sigma)$ where $V$ is a finite set of vertices and $E \subseteq \mathcal{P}(V)\setminus \{\emptyset\}$ is a set of non-empty subsets of $V$ called hyperedges, where $\mathcal{P}(V)$ is the power set of $V$. $\Sigma$ is the set of labels, and $l$ is a label function that assigns each vertex $v$ a label in $\Sigma$, denoted as $l_G(v)$ or $l(v)$ when the context is clear. 
\end{definition}

The number of vertices and hyperedges in $H$ is denoted as $|V(H)|$ and $|E(H)|$, respectively. 
If a vertex $v\in V(H)$ belongs to a hyperedge $e\in E(H)$ then we say that $v$ and $e$ are \textit{incident}. The collection of hyperedges incident to a vertex $v$ is denoted as $he_H(v)$ or $he(v)$ when the context is clear.
The \textit{degree} of a vertex $v$ in $H$, denoted as $d_H(v)$ or $d(v)$ when the context is clear, is the number of hyperedges that are incident to $v$, i.e., $|he_H(v)|$. 
The \textit{arity} of a hyperedge $e$, denoted by $a(e)$, is the number of vertices in $e$, i.e., the number of vertices that are incident to $e$. 
We denote $he^a(v)$ as the set of incident hyperedges with arity $a$.
The average arity of $H$ is denoted as $\overbar{a_H}= \frac{\Sigma_{e\in E(H)}{a(e)}}{|E(H)|}$, and the maximum arity is denoted as $a_{max}$.
%
Two vertices are called \textit{adjacent} if some edge contains both of them. 
Two hyperedges $e_1$ and $e_2$ are called \textit{adjacent} if $e_1\cap e_2 \neq \emptyset$. 
We use $adj_H(u)$ or simply $adj(u)$ to represent all adjacent vertices a vertex $u$. Similarly, $adj_H(e)$ or $adj(e)$ represents all adjacent edges of an edge $e$.

\begin{definition}
\textbf{(Subhypergraph).} A subhypergraph $H'$ of $H=(V,E,l,\Sigma)$ is a hypergraph $H'=(V',E',l,\Sigma)$ where $V'\subseteq V$ and $E'\subseteq E$.
\end{definition}

\begin{definition}
\textbf{(Subhypergraph Isomorphism).} Given a query hypergraph $q$ and a data hypergraph $H$, $q$ is subhypergraph isomorphism to $H$ if and only if there is an \textit{injective} mapping $f:V(q)\rightarrow V(H)$ such that, $\forall u \in V(q), l_q(u)=l_H(f(u))$, and $\forall e_q = \{u_1,u_2,\dots,u_k\}\in E(q), \exists e_H=\{f(u_1),f(u_2),\dots,f(u_k)\} \in E(H)$.
\end{definition}

We refer to each isomorphic subhypergraph as a \textit{subhypergraph isomorphism embedding} of $q$ in $H$. 
In our match-by-hyperedge framework, we simply represent the query hyperedges as $(e_{q_1},e_{q_2},\dots ,e_{q_n} )$.
The matched embedding is therefore denoted as $m=(e_{H_1},e_{H_2},\dots ,e_{H_n})$, where $e_{H_i} = \{f(u):u\in e_{q_i}\}$ for $1\leq i \leq n$.
We use the notion of $f(e_{q_i})=e_{H_i}$ in our presentation to denote matched hyperedges, and $H_m$ to denote the subgraph in $H$ constructed by all hyperedges in $m$.
We call a subhypergraph $q'$ of the query hypergraph $q$ a \textit{partial query}, and an embedding of $q'$ a \textit{partial embedding}.

\begin{example}
For example, consider the query hypergraph $q$ and data hypergraph $H$ in \reffig{hmatch-example}. By representing $q$ as $(\{u_2,u_4\},\allowbreak \{u_0,u_1,u_2\}, \allowbreak \{u_0,u_1,u_3,u_4\})$, there are two \sm embeddings of $q$ in $H$, which can be denoted as $(e_1, e_3, e_5)$ and $(e_2, e_4, e_6)$, respectively. 
Given a partial query with one single hyperedge $q'=(\{u_2,u_4\})$, its matched partial embeddings are $(e_1)$ and $(e_2)$.
\end{example}



\stitle{Problem Statement.} Given a query hypergraph $q$ and a data hypergraph $H$, we study the task of \textit{\sm} to efficiently find all subhypergraph isomorphism embeddings of $q$ in $H$ in parallel on a single machine.

\subsection{Extending Subgraph Matching to Hypergraphs} \label{sec:baseline}

\begin{algorithm}[t]
    \footnotesize
    \setstretch{0.8}
    \SetKwFunction{FnEnumerate}{Enumerate}
    \SetKwFunction{FnGenerateCandidates}{GenerateCandidates}
    
    \SetKwProg{Fn}{Procedure}{:}{}
    
    \KwIn{A query graph $q$ and a data graph $G$}
    \KwOut{All subgraph embeddings from $q$ to $G$}
    
    $\mathcal{C},\mathcal{A}\leftarrow$ generate candidate vertex sets for each query vertex and any other auxiliary data structure \;
    
    $\varphi \leftarrow$ generate a matching order\;
    
    \tcc{Recursively match query vertices to data vertices using backtracking.}
    \FnEnumerate{$q$, $G$, $\mathcal{C}$, $\mathcal{A}$, $\varphi$, $\{\}$, $1$} \;

    \Fn{\FnEnumerate{$q$, $G$, $\mathcal{C}$, $\mathcal{A}$, $\varphi$, $M$, $i$}}{
        \If{$i=|\varphi|+1$}{
            output the embedding $M$ \;
            \Return{} \;
        }
        
        $u\leftarrow$ select a query vertex to match given $\varphi$ and $M$ \;
        
        \ForEach{$v\in$ \FnGenerateCandidates{$\mathcal{C}$, $\mathcal{A}$, $M$} }{
            \If{$v$ is a valid mapping of $u$}{
                Add the mapping $(u,v)$ to $M$ \;
                \FnEnumerate{$q$, $G$, $\mathcal{C}$, $\mathcal{A}$, $\varphi$, $M$, $1+1$} \;
                Remove the mapping $(u,v)$ to $M$ \;
            }
        
        }
        
    }

  \caption{Generic Subgraph Matching Framework}
  \label{alg:subgraph_matching}
\end{algorithm}

As discussed, converting hypergraphs into bipartite graphs can dramatically increase the graph size and incur large overhead \cite{hyperx}. Therefore, a better approach for \sm is to directly extend one of the existing subgraph matching algorithms as in existing works \cite{sm05,sm08,hyper-iso}. 
In this subsection, we briefly present a \emph{generic} framework of extending a subgraph matching algorithm to \sm as our baseline.
%
As most subgraph matching algorithms follow the backtracking framework, we select to extend the backtracking-based subgraph matching framework \cite{match-survey,comparison} to \sm. Thus, our extension can be orthogonally applied to most existing subgraph matching algorithms, including \cite{quicksi,graph-ql,cfl,ceci,daf,vf2}, etc. 

\refalg{subgraph_matching} illustrates the generic subgraph matching framework \cite{comparison,match-survey}. 
In this generic framework, the first step is to generate a candidate vertex set for each query vertex and builds any auxiliary data structures (Line~1). Then, it computes a matching order of the query vertices (Line~2). Finally, it recursively enumerates all results by sequentially mapping each query vertex to data vertices and backtracking when necessary (Line~3). To validate that two vertices $u \in V(q)$ and $v \in V(G)$ can be successfully matching (Line~10), the following constraint is applied in subgraph matching.

\begin{theorem}{(Subgraph Matching Constraint).}
Given $u \in V(q)$ and $v \in V(G)$ are two vertices in the query and data graph, the assignment $f(u)=v$ is valid if and only if for all previous vertices $f(u')=v'$ in the partial embedding such that $(u,u')\in E(q)$, there is an edge $(v,v') \in E(G)$. 
\end{theorem}

To extend the framework to the case of hypergraphs, the above constraint in Line~10 needs to be generalised in the sense that not only compatibility with respect to (w.r.t.) all edges but w.r.t. all \textit{hyperedges} is checked. Therefore, the following new constraint is applied. 
Clearly, such an extension is straightforward to implement. 

\begin{theorem}{(Subhypergraph Matching Constraint).}
Given $u \in V(q)$ and $v \in V(G)$ are two vertices in the query and data hypergraph, the assignment $f(u)=v$ is valid if and only if for all previous vertices $f(u')=v'$ in the partial embedding such that $\{ u' : u'\in M \}+\{u\}\in E(q)$, there is an edge $\{ v' : v'\in M \}+\{v\} \in E(H)$. 
\end{theorem}

\stitle{IHS Filter.}
To further improve the candidate generation (i.e., \kw{GenerateCandidates} function in Line~9) in the generic matching framework for hypergraphs, we add and implement the incident hyperedge structure (\textit{IHS}) filter proposed in existing work \cite{hyper-iso} to our extension when generating candidate vertex set in our implementation.
In \textit{IHS} filter, specifically, a data vertex $v$ is included in the candidate set of the query vertex $u$ if the following four conditions are met:

\begin{itemize}
    \item \textit{Degree and label}: $l(u)=l(v)$ and $d(u)\leq d(v)$
    \item \textit{Number of adjacent nodes}: $|adj(u)|\leq |adj(v)|$
    \item \textit{Arity containment}: $\forall a, |he^a(u)|\leq |he^a(v)|$
    \item \textit{Hyperedge labels}: $\exists e_1,e_2, \forall \sigma, |e_1(\sigma)| = |e_2(\sigma)|$, where $e_1\in he^a(u)$, $e_2\in he^a(v)$ and $e(\sigma)$ is a set of vertices having label $\sigma$ in the hyperedge $e$
\end{itemize}

\begin{remark}
It is worth noting that \cite{hyper-iso} introduces the \textit{IHS} filter only on the subgraph matching algorithm Turbo$_{ISO}$ \cite{turbo-iso}, whereas in our generic extension the filter can be applied for all backtracking-based subgraph matching algorithm, including more recent algorithms such as CFL \cite{cfl}, DAF \cite{daf} and CECI \cite{ceci}. These algorithms with the \textit{IHS} filter can yield better performance than the original method proposed in \cite{hyper-iso}.
So in our experiments, we use the extended versions of CFL, DAF and CECI as our baseline methods.
\end{remark}

%% file: chapter/4.overview.tex
\section{\hm Overview}\label{sec:overview}

In this section, we introduce the basic workflow of \hm followed by the data hypergraph storage mechanism in \hm.

\subsection{Overall Workflow}

\begin{figure}[t]
    \centering
    \centerline{\includegraphics[width=.9\columnwidth]{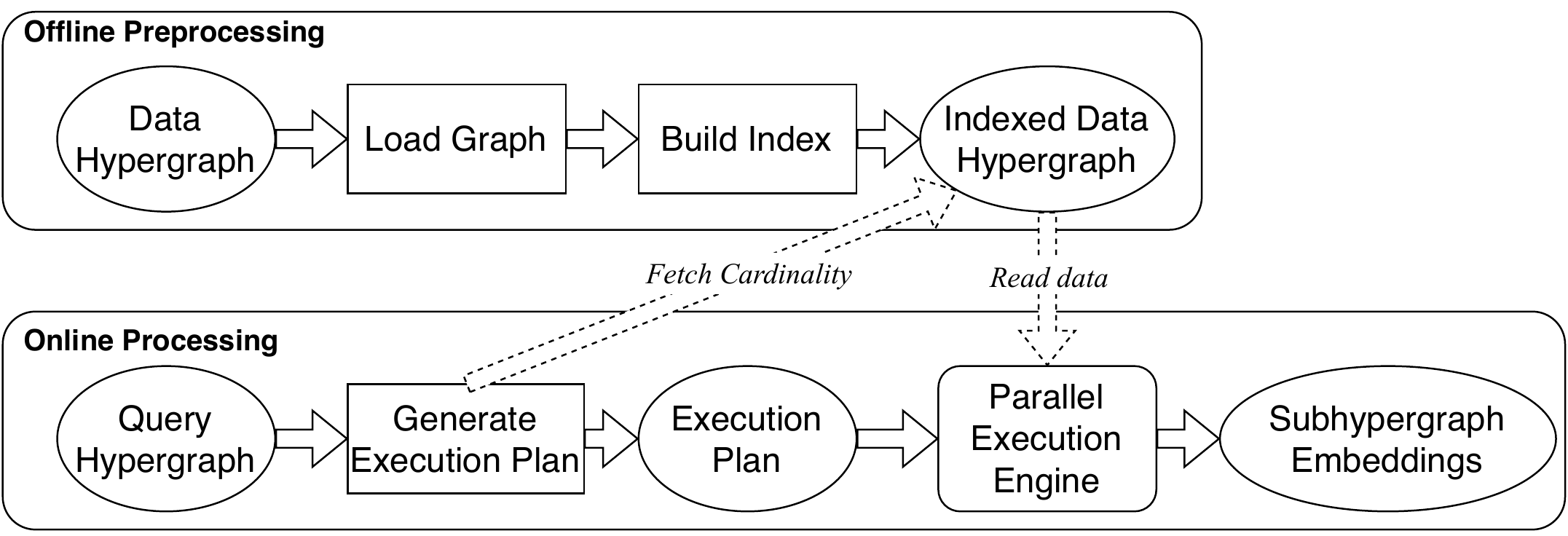}}
    \caption{ \hm Framework Overview. The solid arrow represents the workflow and the dotted arrow shows the interaction between data and processing steps.}
    \label{fig:overview}
\end{figure}

The workflow overview of \hm is illustrated in \reffig{overview}. Specifically, two main stages are \textit{offline data hypergraph preprocessing} and \textit{online query processing}. 

In the offline data hypergraph preprocessing stage, the first step is to load the data hypergraph from the source (e.g., text files) and construct the data hypergraph structure (\refsec{data_structure}). Once the data hypergraph data structure is constructed, \hm builds a lightweight inverted hyperedge index to boost the retrieve process of all incident hyperedges of a given vertex (\refsec{inverted_index}). At the end of prepossessing, an indexed data hypergraph is created. Note that \hm does not build any auxiliary data at runtime, the indexed data hypergraph is created only once offline and it is lightweight, as will be discussed later. 

In the online query processing stage, \hm receives a query hypergraph as its input. Then, the query hypergraph is sent to the plan generator to generate an execution plan. The plan generator fetches cardinality information for the indexed data hypergraph to select a better matching order. The generated execution plan is then input into \hm's parallel execution engine. The execution engine accesses the indexed data hypergraph and executes the given plan to compute all subhypergraph embeddings in parallel.


\subsection{Data Hypergraph Storage} \label{sec:data_structure}

In \hm, we store the data hypergraph in multiple \emph{hyperedge tables}, where each hyperedge table has a unique hyperedge signature. We define the concept of hyperedge signature as follows.

\begin{definition}
\textbf{(Hyperedge Signature).} The signature of a hyperedge $e$, denoted as $\mathcal{S}(e)$, is a \textit{multiset} of all vertex labels contained in $e$, i.e., $\mathcal{S}(e) = multiset\{l(v): v\in e\}$.
\end{definition}

We denote $he(v,s)$ as the set of incident hyperedges with signature $s$. 
\hm stores data hyperedges with different hyperedge signatures in separated hyperedge tables denoted as \textit{partitions}. As a result, to search the candidate hyperedges of a query hyperedge $e_q$, \hm only needs to scan the partition with the signature $\mathcal{S}(e_q)$, rather than scanning the whole hypergraph.

\begin{table}[t]
    \footnotesize
    \centering
    \begin{subtable}{.38\columnwidth}
    \centering
    \caption{Partition $1$}
    \scalebox{0.9}{
    \begin{tabular}{|cc|}
    \hline
    \multicolumn{2}{|c|}{$\mathcal{S}(e) = \{A,B\}$}                         \\ \hline
    \multicolumn{1}{|c|}{\multirow{2}{*}{$E$}} & $e_1 = \{v_2,v_4\}$         \\ \cline{2-2} 
    \multicolumn{1}{|c|}{}                     & $e_2 = \{v_4,v_6\}$         \\ \hline
    \multicolumn{1}{|c|}{\multirow{3}{*}{$I$}} & $v_2 \rightarrow [e_1]$     \\ \cline{2-2} 
    \multicolumn{1}{|c|}{}                     & $v_4 \rightarrow [e_1,e_2]$ \\ \cline{2-2} 
    \multicolumn{1}{|c|}{}                     & $v_6 \rightarrow [e_2]$     \\ \hline
    \end{tabular}
    }
    \end{subtable}
    \begin{subtable}{.53\columnwidth}
    \centering
    \caption{Partition $2$}
    \scalebox{0.9}{
    \begin{tabular}{|cc|}
    \hline
    \multicolumn{2}{|c|}{$\mathcal{S}(e) = \{A,A,C\}$}                           \\ \hline
    \multicolumn{1}{|c|}{\multirow{2}{*}{$E$}} & $e_3 = \{v_0,v_1,v_2\}$         \\ \cline{2-2} 
    \multicolumn{1}{|c|}{}                     & $e_4 = \{v_3,v_5,v_6\}$         \\ \hline
    \multicolumn{1}{|c|}{\multirow{2}{*}{$I$}} & $v_0,v_1,v_2 \rightarrow [e_3]$ \\ \cline{2-2} 
    \multicolumn{1}{|c|}{}                     & $v_3,v_5,v_6 \rightarrow [e_4]$ \\ \hline
    \end{tabular}
    }
    \end{subtable}
    \\
    \begin{subtable}{\columnwidth}
    \centering
    \caption{Partition $3$}
    \scalebox{0.9}{
    \begin{tabular}{|cc|}
    \hline
    \multicolumn{2}{|c|}{$\mathcal{S}(e) = \{A,A,B,C\}$}                         \\ \hline
    \multicolumn{1}{|c|}{\multirow{2}{*}{$E$}} & $e_5 = \{v_0,v_1,v_4,v_6\}$     \\ \cline{2-2} 
    \multicolumn{1}{|c|}{}                     & $e_6= \{v_2,v_3,v_4,v_5\}$      \\ \hline
    \multicolumn{1}{|c|}{\multirow{3}{*}{$I$}} & $v_0,v_1,v_6 \rightarrow [e_5]$ \\ \cline{2-2} 
    \multicolumn{1}{|c|}{}                     & $v_4 \rightarrow [e_5,e_6]$     \\ \cline{2-2} 
    \multicolumn{1}{|c|}{}                     & $v_2,v_3,v_5 \rightarrow [e_6]$ \\ \hline
    \end{tabular}
    }
    \end{subtable}
    \caption{ Example Data Layout of the Hypergraph in \reffig{data}. The header of each table represents the signature of all its hyperedges. $E$ represents hyperedges and $I$ represents inverted hyperedge index.}
    \label{tab:edge_lists}
\end{table}

\begin{example}
The partitioned hyperedge tables of the data graph $H$ in \reffig{data} are shown in \reftable{edge_lists}.
For the given data hypergraph, \hm constructs three partitions having signatures $\{A,B\}$, $\{A,A,C\}$ and $\{A,A,B,C\}$, respectively. 
\end{example}

\stitle{Size Analysis.} 
The proposed hypergraph data structure in \hm brings only a very small overhead of an additional signature header for each partition, which is no larger than the size of all hyperedges (i.e., all hyperedges have unique signatures in the worst case). 
Thus, the total size of storing all hyperedges in \hm is $O(\overbar{a_H}\times |E(H)|)$.

\subsection{Inverted Hyperedge Index} \label{sec:inverted_index}

In subgraph matching in conventional graphs, it is essential to access all connected edges (i.e., neighbours) of a given vertex. In hypergraphs, similarly, it is often demanded to get all incident hyperedges of a given vertex. %
Given a hyperedge table, it requires a linear scan to complete such an operation which can be time-consuming for large hypergraphs. To further speed up this process of finding all incident hyperedges (with a certain signature) of a given vertex, we adopt the common technique of inverted index \cite{bitmap1,bitmap2} to build a lightweight \textit{inverted hyperedge index} for each hyperedge table in \hm. 

\begin{example}
The inverted hyperedge index of the data graph $H$ in \reffig{data} is also shown in \reftable{edge_lists}. It maps each vertex in a hyperedge list to a \textit{posting list} of hyperedge IDs of all its incident hyperedges in the hyperedge list (in ascending order). 
\end{example}

\stitle{Size Analysis.} 
The inverted hyperedge index in \hm is also lightweight. 
For each hyperedge, its hyperedge ID will appear in the posting list of all the vertices it contains.
Therefore, each hyperedge $e$ takes additionally $O(a(e))$ space.
The total size of the inverted edge index is also $O(\overbar{a_H}\times |E(H)|)$.

%% file: chapter/5.matching_by_edges.tex
\section{Match-by-Hyperedge Framework}\label{sec:matching}

In \hm, we propose a \textit{match-by-hyperedge} framework for efficient \sm enumeration. By using hyperedges as the minimal matching sub-structure, \hm is able to fully utilise the higher order information in hypergraphs to generate candidates and avoid redundant computation. 

The framework is illustrated in \refalg{hgmatch}. Given the query hypergraph $q$ and the data hypergraph $G$, we first compute a \textit{matching order} of the query hyperedges (Line~1), which is defined in \refdef{matching_order}.

\begin{definition} \label{def:matching_order}
\textbf{(Matching Order).} 
A matching order in \hm, denoted as $\varphi$, is a permutation of $E(q)$. $\varphi[i]$ is the $i$th hyperedge in $\varphi$ and $\varphi[i:j]$ is the set of hyperedges from index $i$ to $j$ ($1\leq i\leq j \leq |\varphi|$).
\end{definition}

After that, we initialize the set of embeddings $M$ as the hyperedges that match the first query hyperedge in $\varphi$, namely $\varphi[0]$ (Line~2-3). Then, for the remaining query hyperedges in $\varphi$, \hm iteratively expands each partial embedding by one matched hyperedges in parallel until all hyperedges are matched (Line~4-12). When expanding a partial embedding $m$, we first generate a candidate set $C(e_q)$ for the query hyperedge $e_q$ that is currently being matched (Line~7). Then, a validation is performed to filter out false positives (Line~9).

Note that we present the algorithmic framework of \hm (i.e., \refalg{hgmatch}) in the form of a bread-first search (BFS) for the ease of presentation, whereas the actual parallel enumeration is scheduled using \hm's task-based scheduler (\refsec{scheduler}).

\begin{algorithm}[t]
    \footnotesize
    \setstretch{0.8}
    \SetKw{parallel}{parallel}
    
    \SetKwFunction{FnCompOrder}{ComputeMatchingOrder}
    \SetKwFunction{FnGenCand}{GenerateHyperedgeCandidates}
    \SetKwFunction{FnVerify}{IsValidEmbedding}
    
    \KwIn{A query hypergraph $q$ and a data hypergraph $H$}
    \KwOut{All embeddings from $q$ to $H$}
    
    \tcc{Compute a matching order of query hyperedges.}
    $\varphi \leftarrow$ \FnCompOrder{$q$, $H$} \;
    
    \tcc{Enumerate embeddings of the 1st query hyperedge.}
    $e_1\leftarrow $ the 1st query hyperedge in $\varphi$ \;
    $M\leftarrow$ all matched hyperedges of $e_1$ in $H$ \;
    
    \tcc{Iteratively match the remaining query hyperedges.}
    \ForEach{$e_q \in (\varphi - e_1)$}{
    
        $M'\leftarrow \{\}$ \;
        
        \parallel \ForEach{$m\in M$}{
            $C(e_q)\leftarrow$ \FnGenCand{$q$, $H$, $\varphi$, $e_q$, $m$} \;
            
            \ForEach{$c\in C(e_q)$}{
                $m'\leftarrow m+c$ \;
                \If{\FnVerify{$q'$, $m'$} where $q'$ is the partial query}{
                    $M'\leftarrow M'+{m'}$ \;
                }
            }
        }
        $M\leftarrow M'$ \;
    }
    
    \Return{$M$} \;
    
  \caption{\hm Framework}
  \label{alg:hgmatch}
\end{algorithm}

\subsection{Matching Order} \label{sec:matching_order}

\hm can work with any connected matching order.
We compute the matching order of query hyperedges using the cardinality information fetched from the metadata of hyperedge tables to match \emph{infrequent} and \emph{highly connected} hyperedge as early as possible. 

\begin{definition}
\textbf{(Hyperedge Cardinality).}
The cardinality of a query hyperedge $e_q$ in a data hypergraph $H$, denoted as $Card(e_q,H)$, is $|\{e: e\in E(H), \mathcal{S}(e) = \mathcal{S}(e_q) \}|$.
\end{definition}

\begin{algorithm}[t]
    \footnotesize
    \setstretch{0.8}
    
    \KwIn{A query hypergraph $q$ and a data hypergraph $H$}
    \KwOut{A matching order $\varphi$}
    
    \tcc{Find the hyper with smallest cardinality as a starting hyperedge.}
    $e_1\leftarrow \argmin_{e} Cardy(e, H)$, where $e\in E(q)$ \;
    
    $\varphi\leftarrow (e_1)$ \;
    
    \tcc{Iteratively add connected hyperedges}
    \While{$|\varphi| \neq |E(q)|$}{
        $e_q\leftarrow \argmin_e \frac{Card(e, H)}{V_{\varphi}\cap e}$, where $e \in E(q)-\varphi$ and $V_{\varphi}\cap e \neq \emptyset$ \;

        $\varphi\leftarrow \varphi + {e_q}$ \;
    }
    
    \Return{$\varphi$} \;
    
  \caption{\kw{ComputeMatchingOrder}}
  \label{alg:matching_ord}
\end{algorithm}

The algorithm for computing a matching order is given in \refalg{matching_ord}.
It starts with the query hyperedge $e_1$ that has the minimal cardinality in the data graph (Line~1-2).
Then, we iterate over all remaining edges in the query graph and add the query hyperedge with small cardinality and a highly connected number of vertices with the existing partial query (Line~4).

\sstitle{\red{Complexity.}} 
\red{
Since $Card(e,H)$ can be directly accessed from the hyperedge tables (i.e. the number of rows in a table) in $O(1)$, the time complexity of \refalg{matching_ord} is $O(|E(q)|^2)$.
}

\subsection{Candidates Generation} \label{sec:cand_gen}

In this subsection, we introduce how \hm generates candidates of a given query hyperedge. We first present some observations used in \hm to prune candidates. Then we discuss how \hm generates these candidates efficiently using set operations.

Suppose that we are computing the candidates data hyperedge of a query hyperedge $e_q = \varphi[i]$, for a partial embedding $m$, where the current partial query is $q'$ and the partial query after matching $e_q$ is $q$. We generate and prune the candidates based on the following observations.

\begin{observation} \label{observation_1}
\textit{(Hyperedge Signature).} 
The matched hyperedge must have the same hyperedge signature as the query signature, i.e., 
$\mathcal{S}(f(e_q)) = \mathcal{S}(e_q)$.
\end{observation}

\begin{observation} \label{observation_2}
\textit{(Hyperedge Adjacency).} 
For each adjacent hyperedge $e$ of $e_q$ in $E(q)$, $f(e_q)$ must be adjacent to $f(e)$, i.e.,
$\forall e\in E(q), e\in adj(e_q), \exists v\in f(e), v\in f(e_q)$.
\end{observation}

\begin{observation} \label{observation_3}
\textit{(Hyperedge Non-Adjacency).} 
For each non-adjacent hyperedge $e$ of $e_q$ in $E(q)$, $f(e_q)$ must not be adjacent to $f(e)$, i.e.,
$\forall e\in E(q), e\not\in adj(e_q), \not\exists v\in f(e), v\in f(e_q)$.
\end{observation}

\begin{observation} \label{observation_4}
\textit{(Labels and Degree of Incident Vertices).} 
Suppose $e\in adj(e_q)$, if $u\in e$ and $u\in e_q$, then $\exists v\in f(e), v\in f(e_q),l_H(v)=l_q(u), d_{H_m}(v)=d_{q'}(u)$.
\end{observation}

Based on the four observations above, the procedure of candidate generation in \hm is given in \refalg{cand_gen}.
The algorithm firstly obtains a set of vertices in the partial embedding $m$ that must not be incident to the current hyperedge being matched (Line~1). At Line~2, we initialise an empty set of sets of potential hyperedge candidates.
Then, for each hyperedge $e$ in the current partial query that is adjacent to $e_q$, we iterate each query vertex $u$ in $e$ that is also incident to $e_q$ (Line~3-4).
For each vertex $u$, we found a set $V_{incdt}$ of all vertices in $f(e)$ that can be possibly matched to $u$ using Observation \ref{observation_2}, \ref{observation_3}. and \ref{observation_4} (Line~5).
Given $V_{incdt}$, the candidate hyperedges must be incident to at least one vertex in $V_{incdt}$, having the same hyperedge signature as $e_q$ (Observation \ref{observation_1}).
Therefore, we add the union of hyperedges incident to $v\in V_{incdt}$ with the hyperedge signature of $\mathcal{S}(e_q)$ as an element to $C'$ (Line~6).
Lastly, the final candidate set is computed by taking the intersection of all elements (i.e., each element is a set of hyperedges) in $C'$ (Line~7). This is because $f(e_q)$ must be incident to each vertex $u$ in Line~4.



Note that \hm employs set difference, union and intersection in candidate generation, which can be implemented very efficiently on modern hardware \cite{simd_gallop,qfilter,simd-intersection,bitmap2}.
With the help of \hm's inverted hyperedge index, $he(v,\allowbreak \mathcal{S}(e_q))$ can be quickly fetched in a constant time from the index (Line~6),  followed by a direct set intersection (Line~7).
%
%
%
%



\begin{algorithm}[t]
    \footnotesize
    \setstretch{0.8}
    
    \KwIn{A query hypergraph $q$, a data hypergraph $H$, a matching order $\varphi$, a query hyperedge $e_q$ , a partial embedding $m$}
    \KwOut{All candidates hyperedges of $e_q$ in $H$}
    
    \tcc{All vertices in $m$ that are definitely non-incident to $f(e_q)$}
    $V_{n\_incdt} = \{v: v\in f(e), f(e)\in m, e\in E(q), e\not\in adj(q)\}$ \;
    
    \tcc{Init an empty set of sets of hyperedges.}
    $C' = \{\}$ 
    
    \tcc{Find a set of candidates for each vertex that the matched data hyperedge may be incident to}
    \ForEach{$e\in adj(e_q), f(e)\in m$}{
        \ForEach{$u\in e,u\in e_q$}{
            $V_{incdt}\leftarrow \{v: v\in f(e) - V_{n\_incdt},l_H(v)=l_q(u),d_{H_m}(v)=d_{q'}(u)\}$, where $q'$ is the partial query\;
            \tcc{append a set of candidate hyperedges to $C'$}
            $C'\leftarrow C' + \{\bigcup_{v\in V_{incdt}}he(v,\mathcal{S}(e_q))$\} \;
        }
    }
    
    \tcc{Intersecting all sets of candidates to obtain the final candidates}
    $C(e_q)\leftarrow \bigcap_{E_{i}\in C'} E_{i}$ \;
    
    \Return{$C(e_q)$} \;
    
  \caption{\kw{GenerateHyperedgeCandidates}}
  \label{alg:cand_gen}
\end{algorithm}



\sstitle{\red{Complexity.}} 
\red{
In the worst case, assuming that $f(e_q)$ is incident to all previously matched vertices, where each vertex can be contained in $O(E(H))$ hyperedges in the data hypergraph. The time complexity of \refalg{cand_gen} is $O(|V(q)|\times |E(H)|)$.
}

\begin{example}
We give an example for the query and data hypergraph given in \reffig{hmatch-example}.
Suppose that the matching order is $\varphi = (\{u_2,u_4\},\allowbreak \{u_0,u_1,u_2\}, \allowbreak \{u_0,u_1,u_3,u_4\})$ and $m=(e_1=\{v_2,v_4\},e_3=\{v_0,v_1,v_2\})$. To compute $C(\{u_0,u_1,u_3,u_4\})$, the candidate hyperedge must be incident to $v_0,v_1\in e_3$ and $v_4\in e_1$. Thus, by accessing the inverted hyperedge index in the partition with signature $s=\mathcal{S}(\{u_0,u_1,u_3,u_4\})$, namely partition $3$ in \reftable{edge_lists}, we compute the candidates as $he(v_0,s)\cap he(v_1,s)\cap he(v_4,s) = \{e_5\}$.

\end{example}

\begin{remark}
More filtering rules can be added to \hm. However, we found in our experiments that \refalg{cand_gen} yields an extremely low false positive rate, as demonstrated in our experiments. 
Thus, we do not include more filtering rules to trade off between filtering time and false positive rate.
\end{remark}

\subsection{Embedding Validation} \label{sec:emb_vali}

The candidates generation in \hm may produce false positives. Hence we perform an embedding validation to remove false positives. Suppose the embedding to be verified is $m'$ and the corresponding partial query is $q'$, and we denote the last hyperedge in $m'$ and the corresponding match hyperedge in $q'$ as $e_{m'}$ and $e_{q'}$, respectively. We first use the following observations to validate.

\begin{observation} \label{val_rule_1}
\textit{(Number of Vertices). $H_{m'}$ and $q'$ must have the same number of vertices, i.e., $|V(H_{m'})|=|V(q')|$.}
\end{observation}


Obviously, Observation \ref{val_rule_1} is trivial to validate. 
Moreover, to ensure the candidate hyperedge is valid, we can explicitly compute a mapping of vertices between the query and data hyperedge stated in the following lemma.

\begin{lemma}{(Mapping of Vertices).} \label{lem:vertex_mapping}
There exist a \textit{bijective} mapping $f':e_{q'}\leftarrow e_{m'}$ such that, $\forall u\in e_{q'}, l_q(u)=l_H(f'(u))$, and $\forall u\in e, \forall e\in E(q'), f'(u)\in f(e)$.
\end{lemma}

A naive way to check \reflem{vertex_mapping} is to compute the vertex mapping using backtracking. However, it introduces many overheads and can be time-consuming similar to the match-by-vertex approach. 
Instead, we introduce the new concept of \textit{vertex profile} for each vertex in the newly added hyperedge and compare the multisets of vertex profiles for all vertices in the partial query and embedding to verify they are valid.

\begin{definition}
\textit{(Vertex Profile.)} A vertex profile of a vertex $v$ in $H_{m'}$ if a tuple $\mathcal{P}(v)=(l(v),he_{H_{m'}}(v))$.
\end{definition}

A vertex profile represents the label and all incident hyperedges of a vertex. For a vertex $u$ in a partial query $q'$ and a partial embedding $m'$, we use the notion of $\mathcal{P}_{q'\rightarrow m'}(u) = (l_{q'}(u), \{f(e):e\in he_{q'}(u)\})$ to denote the profile of its label and the set of all matched hyperedges in the partial embedding of its incident hyperedges. We overload the term $\mathcal{P}(u)$ to denote $\mathcal{P}_{q'\rightarrow m'}(u)$ when context is clear.
Using the definition of vertex profile, we have the following theorem.

\begin{theorem} \label{thm:validation}
\textit{(Equivalence of Vertex Profile.)} 
Given a partial query $q'$, a corresponding partial embedding $m'$, with the last hyperedge being match denoted as $e_{m'}$ and $e_{q'}$.
The partial embedding is valid if and only if the two \textit{multisets} $\{\mathcal{P}_{q'\rightarrow m'}(u) : u\in e_{q'} \}$ and $\{\mathcal{P}(v) : u\in e_{m'} \}$ are equal.
\end{theorem}

\begin{proof}
When $\{\mathcal{P}_{q'\rightarrow m'}(u) : u\in e_{q'} \}$ and $\{\mathcal{P}(v) : u\in e_{m'} \}$ are equal, there must exist a bijective mapping that satisfies \reflem{vertex_mapping}. This can be done by mapping each query vertex to a data vertex with the same vertex profile. Since the two multisets are equal, we can definitely construct at least one such bijective mapping.
\end{proof}

Employing the concept of vertex profile, \hm avoids backtracking searches. 
The equivalence of two multisets can be easily compared using hash-based implementation. 
%
%
Combining Observation \ref{val_rule_1} and \refthm{validation}, \hm derives its validation process as illustrated in \refalg{emb_vali}.

\begin{algorithm}[t]
    \footnotesize
    \setstretch{0.8}
    
    \KwIn{A partial query $q'$, a partial embedding $m'$ }
    \KwOut{If $m'$ is valid w.r.t. $q'$}
    
    \tcc{Verify the numbers of vertices are equal.}
    \If{$|V(q')|!=|V(m')|$}{
        \Return False \; 
    }
    
    \tcc{Construct the multiset for $q'$.}
    $S_{q'} = multiset()$ \;
    \ForEach{$u\in e_{q'}$}{
        $S_{q'}\leftarrow S_{q'} + \mathcal{P}_{q'\rightarrow m'}(u)$ \;
    }
    
    \tcc{Construct the multiset for $m'$.}
    $S_{m'} = multiset()$ \;
    \ForEach{$v\in e_{m'}$}{
        $S_{m'}\leftarrow S_{m'} + \mathcal{P}(v)$ \;
    }
    
    \tcc{Return whether two multisets are equal or not.}
    \Return{$S_{q'}=S_{m'}$} \;
    
  \caption{\kw{IsValidEmbedding}}
  \label{alg:emb_vali}
\end{algorithm}


\begin{figure}[h]
    \centering
    \begin{subfigure}[b]{.38\columnwidth}
        \centering
        \includegraphics[width=0.95\textwidth]{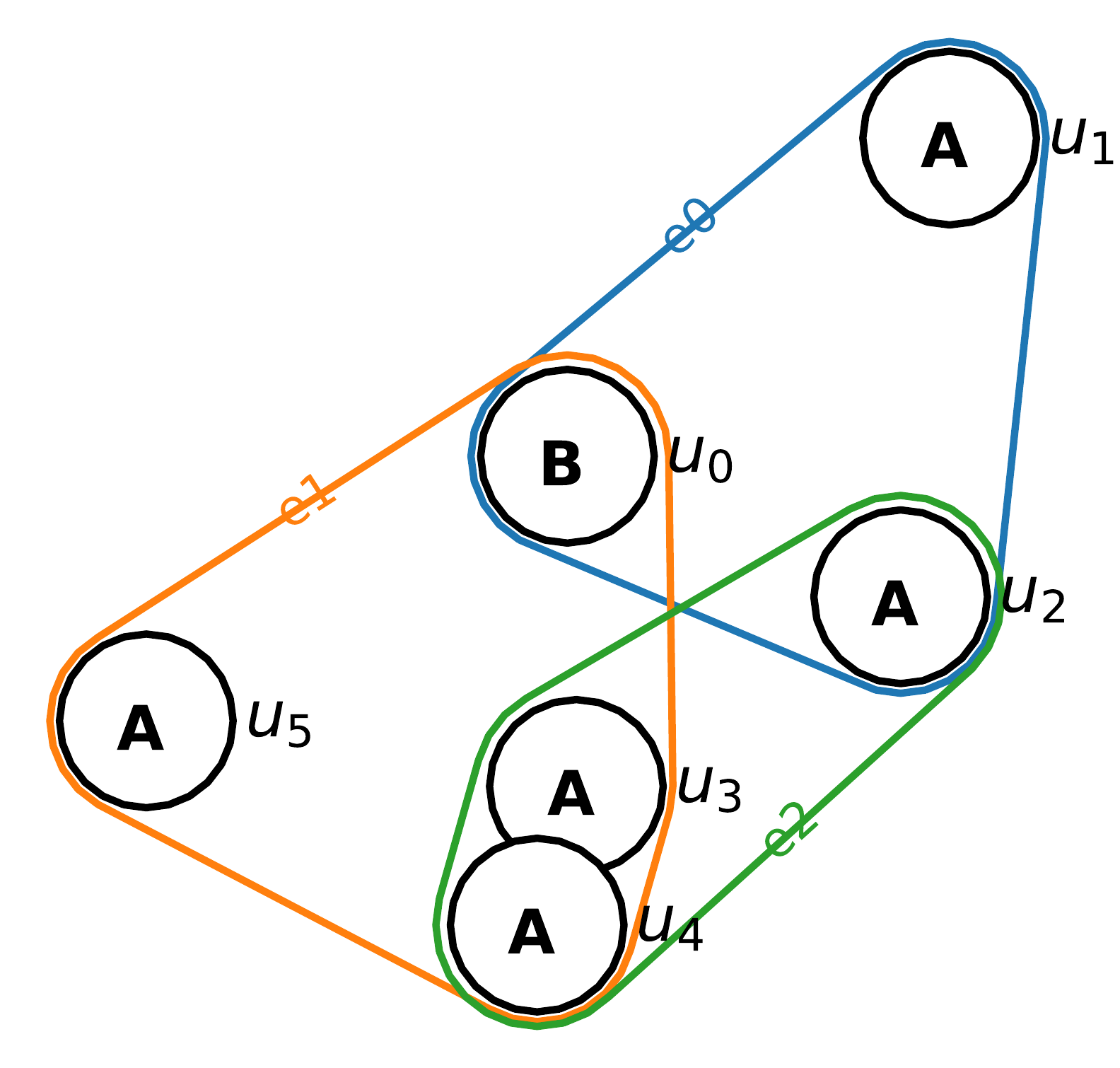}
        \caption{Partial Query ($q'$)}
        \label{fig:partial_query}
    \end{subfigure} %
    ~
    \begin{subfigure}[b]{.4\columnwidth}
        \centering
        \includegraphics[width=\textwidth]{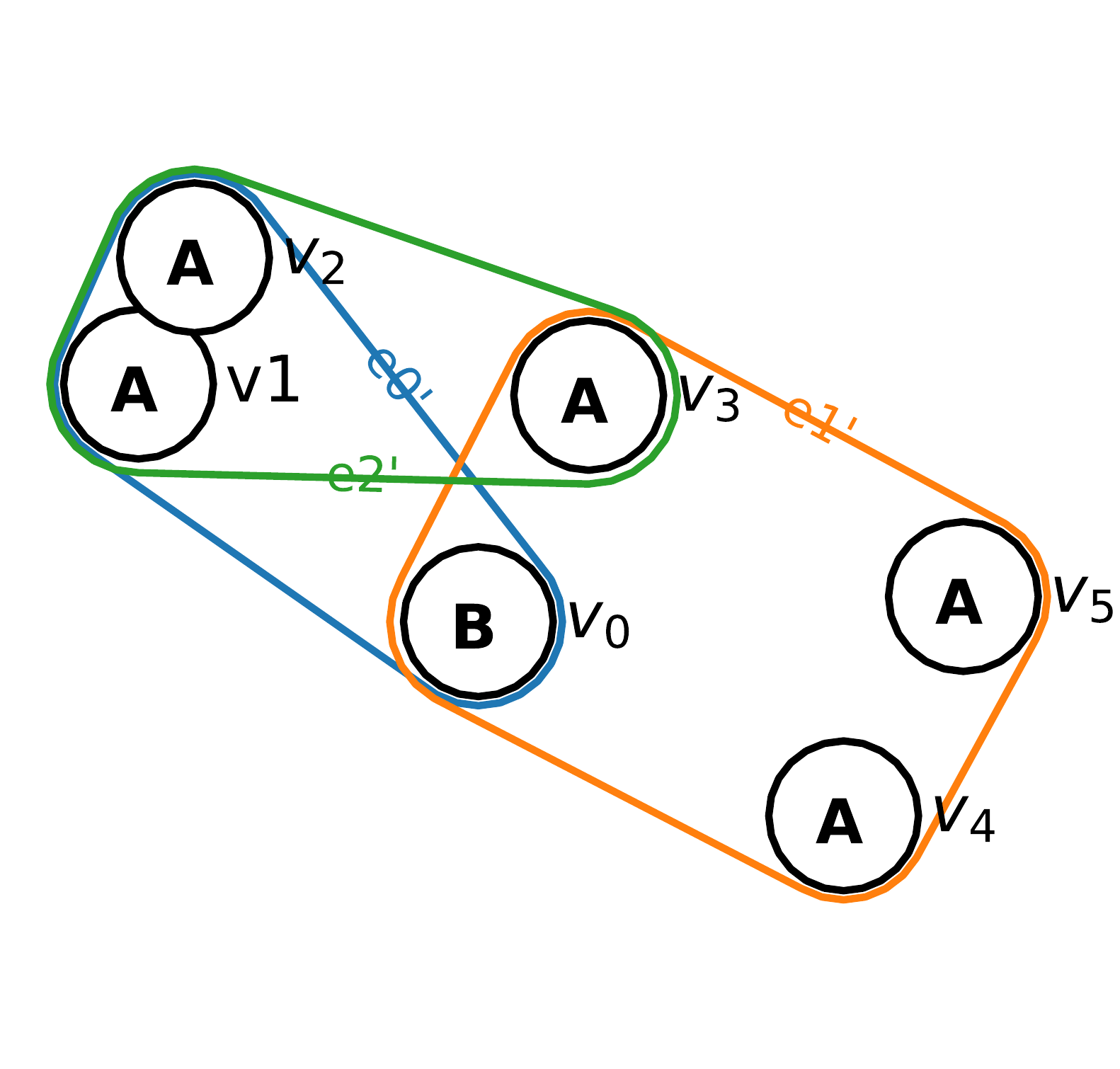}
        \caption{Partial embedding ($m'$)}
        \label{fig:partial_embedding}
    \end{subfigure}
    \caption{ Example of Embedding Validation.}
    \label{fig:validate-example}
\end{figure}

\begin{example}
An example of embedding validation is illustrated in \reffig{validate-example}, where we have a partial query $q'$ and a candidate partial embedding $m'$.
Suppose the matching order of $q'$ is $(e_0,e_1,e_2)$.
Currently we have a partial embedding $(e_0',e_1')$ of $(e_0,e_1)$ and want to verify the mapping $f(e_2)=e_2'$ is valid.
For $e_3\in q'$, we compute the vertex profile of all its vertices, where $\mathcal{P}(u_2) = (A,\{e_0',e_2'\})$,
$\mathcal{P}(u_3) = (A,\{e_1',e_2'\})$, and $\mathcal{P}(u_4) = (A,\{e_1',e_2'\})$.
Similarly, for $e_3'\in m'$, we compute $\mathcal{P}(v_1) = (A,\{e_0',e_2'\})$,
$\mathcal{P}(v_2) = (A,\{e_0',e_2'\})$, and $\mathcal{P}(v_3) = (A,\{e_1',e_2'\})$.
We notice that $\{\mathcal{P}(u_2), \mathcal{P}(u_3), \mathcal{P}(u_4)\}\neq \{\mathcal{P}(v_1), \mathcal{P}(v_2), \mathcal{P}(v_3)\}$.
Hence, $m'$ is not a valid embedding of $q'$.
\end{example}

\sstitle{\red{Complexity.}} 
\red{
The time complexity of \refalg{emb_vali} is $O(\overbar{a_q}\times |E(q)|)$, as it iterates over every vertex in every query hyperedge.
}

%

%% file: chapter/6.parallel_execution.tex
\section{Parallel Execution}\label{sec:execution}

In this section, we describe the details of \hm's parallel execution engine. On receiving the execution plan of a query hypergraph, \hm initialises a \textit{thread pool} of $p$ threads to run the plan in parallel.

\subsection{Dataflow Model}

\begin{figure}[t]
    \centering
    \begin{subfigure}[b]{.315\columnwidth}
        \centering
        \includegraphics[width=0.9\linewidth]{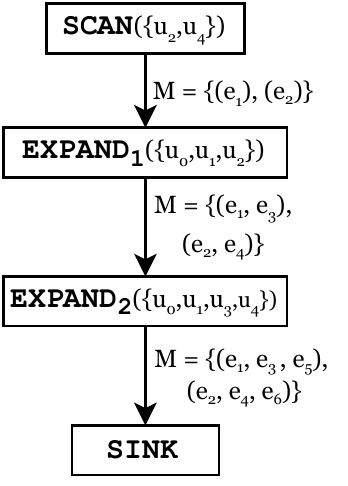}
        \caption{Dataflow Graph}
        \label{fig:dataflow}
    \end{subfigure}
    \begin{subfigure}[b]{.54\columnwidth}
        \centering
        \includegraphics[width=0.9\linewidth]{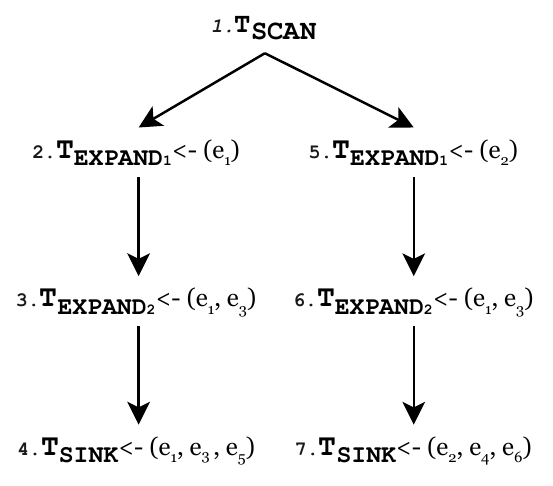}
        \caption{Task Tree with Schedule Order}
        \label{fig:task_tree}
    \end{subfigure}
    \caption{ Example Dataflow Graph and Task Tree with Schedule Order of \reffig{hmatch-example}.}
\end{figure}

We adopt the popular dataflow model \cite{dataflow,graphflow,patmat-exp,wco-join,huge,flink} for \hm, where computation is abstracted as a dataflow graph \cite{dataflow-def}. 
A dataflow graph is a directed acyclic graph (DAG), in which each vertex is an operator, and the directed edges represent data flows. An operator consists of a predefined computing instruction, a certain input and output channels. We introduce three primitive operators \scan, \expand and \sink in \hm as follows.

\sstitle{\scan.} \scan($e_q$) is always the first operator in \hm.
It accepts a query hyperedge $e_q$ as its parameter, iterates over a hyperedge partition of the data hypergraph $H$, and outputs all data hyperedges with signature $s=\mathcal{S}(e_q)$.

\sstitle{\expand.} \expand($e_q$) accepts a query hyperedge $e_q$ as its parameter and partial embeddings as its input. It expands each of its input partial embedding by one more hyperedge $e_q$ (\refsec{cand_gen} and \refsec{emb_vali}), and output new embeddings.

\sstitle{\sink.} \sink is always the last operator in \hm.
It accepts partial embeddings as its input and consumes the results of \sm, via either counting or printing all embeddings.

A dataflow graph in \hm is always a \textit{direct path} from a \scan operator to a \sink operator with a number of \expand operators in between.

\begin{example}
An example dataflow graph for the query and data hypergraph in \reffig{hmatch-example} is shown in \reffig{dataflow}, where operators are drawn in rectangles and data flows (i.e., intermediate results) are illustrated next to arrows.
\end{example}

\begin{remark}
The dataflow design of \hm makes it highly customizable and allows it to be easily extended with other functionalities of hypergraph databases. 
These can be done by introducing new dataflow operators.
Examples include adding extra aggregation and property filtering to the dataflow graph.
We leave this as an interesting future work.
\end{remark}

\subsection{Task-based Scheduler} \label{sec:scheduler}

A straightforward and commonly used method for parallel execution is to use breath-first search (BFS), which is easy to implement and can achieve high CPU utilisation. 
However, due to the exponential number of intermediate results of \sm, BFS-based scheduling usually consumes a large size of memory and easily results in an out-of-memory error.
Depth-first search (DFS), on the other hand, is often adopted for sequential algorithms.
DFS is generally hard to parallelise and can lead to low CPU utilisation and load imbalance in the parallel environment.
To fulfil the needs of high parallelism degree and low memory consumption at the same time, we design a task-based scheduler in \hm to execute the dataflow graph in parallel.

\begin{definition}
\textbf{(Task).}
A task is a minimal unit to schedule in \hm. It takes either a partition of the data hyperedges or a (partial) embedding as input, executes predefined computation logic, and spawns zero or more new tasks.
\end{definition}

There are three types of tasks in \hm with respect to the three dataflow operators (i.e., \scan, \expand and \sink). We denote them as \tscan, \texpand and \tsink, respectively.
A \tscan task takes a partition of the data hyperedges, and spawns a \texpand task for each of the hyperedge in the partition. A \texpand task takes a partial embedding and spawns a \texpand or \tsink task for each expanded partial embedding. A \tsink task takes a partial embedding and counts/outputs the embedding. Once a dataflow graph is constructed, \hm executes the dataflow graph using tasks.

\begin{remark}
Tasks in \hm are \textit{lightweight}. A task contains only a partial embedding and a pointer to the function defining its execution logic (as defined in its corresponding dataflow operator). Hence, the overhead of spawning tasks is very small.
\end{remark}

We then discuss how tasks are scheduled in \hm. We first assume the thread pool has only one thread to ease our presentation. To control memory consumption, \hm employs a \textit{last-in first-out} (LIFO) strategy when scheduling the tasks, in which each thread has a task queue.
Specifically, newly spawned tasks will be put to the head of the task queue, and the most recent task is scheduled first (i.e., from the head). The use of LIFO scheduling order avoids materialising all the intermediate results in the memory at one time to efficiently reduce memory consumption. 

\begin{example}
A tree of all tasks spawned by the operators in \reffig{dataflow} is given in \reffig{task_tree}. The scheduling order of each task is individuated by the number before it. As shown in the figure, the scheduling employs a LIFO order.
\end{example}

With the scheduler, \hm achieves a bounded memory.

\begin{theorem}
\hm schedules a \sm job with the memory bound of $O(\overbar{a_q}\times|E(q)|^2\times|E(H)|)$.
\end{theorem}

\begin{proof}
A dataflow graph in \hm has $|E(q)|$ operators. As \hm schedules in a LIFO order, in each step, each operator can spawn $|C(e_q)|$ tasks (i.e., $O(|E(H)|)$) to the task queue, where each task takes a partial embedding with at most $|E(q)|$ hyperedges.
Representing the partial embedding requires a space of $O(\overbar{a_q}\times|E(q)|)$ to maintain all vertex IDs. 
Therefore, a task queue takes $O(\overbar{a_q}\times|E(q)|^2\times|E(H))$ space. Since there are $p$ task queues in a \hm thread pool where $p$ is a pre-configured constant. The overall memory bound is $O(\overbar{a_q}\times|E(q)|^2\times|E(H)|)$.
\end{proof}


\subsection{Load Balancing} \label{sec:load}

Graph computation is usually irregular due to the power-law characteristics in real-world graphs \cite{power-law-1,power-law-2}. A simple solution to distribute the load among multiple threads is to assign each thread with an equal share of the firstly matched hyperedges. 
This \textit{static} and \textit{coarse-grained} technique can still suffer from load skew due to the unpredictable number of embeddings generated by one hyperedge \cite{huge}.  
In \hm. we adopt the \textit{fine-grained} work-stealing technique \cite{work_stealing1,work_stealing2,cilk,RI} to \textit{dynamically} balance the load.

In \hm's dynamic  work stealing, each thread in the thread pool maintains a \textit{deque} as its task queue. Since tasks are scheduled in LIFO order new tasks will be spawned into the head of the task queue. Each task execution will pop the most recent task from the head of the task queue.
Once a thread has completed its own job by emptying its task queue, it will randomly pick one of the other threads with a non-empty task queue and steal half of the tasks from the tail to balance the load. 
We implemented a non-blocking deque \cite{Chase-Lev-deque} to reduce the overhead of lock operations.
The stealing operates on the minimal scheduling unit in \hm, namely task, to achieve near-perfect load balancing in fine grain.




%% file: chapter/7.experiment.tex
\section{Experiments}\label{sec:experiment}

\subsection{Experimental Setup}

\hm is implemented in Rust.
%
All experiments are conducted on a server with two $20$-core Xeon E5-2698 v4 CPUs ($40$ threads each) and $512$GB memory.

\stitle{Baselines.}
We compare \hm with the state-of-the-art subgraph matching algorithms \cfl \cite{cfl}, \daf \cite{daf} and \ceci \cite{ceci}.
We adopt the C++ implementations in a recent experimental study of subgraph matching\footnote{\url{https://github.com/RapidsAtHKUST/SubgraphMatching}} \cite{match-survey}, and extend them to the case of hypergraphs as described in \refsec{baseline} with additional \textit{IHS} filter.
Note that this implementation utilises single instruction multiple data (SIMD) instructions \cite{simd_gallop} to speed up set intersections.
We did not implement SIMD set intersections in \hm.
%
The modified algorithms are denoted as \cflh, \dafh and \cecih, respectively. 
We do not include \cite{hyper-iso} for the reason discussed at the end of \refsec{baseline}.
Also, we do not compare \cite{sm05,sm08} since the subgraph matching algorithm they extend, namely the Ullmann's algorithm \cite{ullmann}, has been largely outperformed by \cfl, \ceci and \daf in the literature.
We also compare \rapid \footnote{\url{https://github.com/RapidsAtHKUST/RapidMatch}}\cite{rapidmatch}. But since \rapid uses join-based techniques which cannot be fitted in our generic backtracking framework, we directly convert the query and data hypergraph to bipartite graphs in \rapid.

\stitle{Datasets.}
We use $10$ real-world data hypergraphs with labelled vertices in our experiment downloaded from \cite{datasets}. 
They are house committees (\HC), MathOverflow answers (\MA), contact high school (\CH), contact primary school (\CP), senate bills (\SB), house bills (\HB), Walmart trips (\WT), Trivago clicks (\TC), StackOverflow answers (\SA), and Amazon reviews (\AR).
We preprocess the datasets to remove all repeated hyperedges and all repeated vertices in one hyperedge. 
The statistics of the datasets are shown in \reftable{datasets}.

\begin{table}[t]
\footnotesize
\centering
\caption{Table of Datasets }
\label{tab:datasets}
\scalebox{0.9}{
\begin{tabular}{|c|r|r|r|r|r|r|}
\hline
Dataset &
  \multicolumn{1}{c|}{$|V|$} &
  \multicolumn{1}{c|}{$|E|$} &
  \multicolumn{1}{c|}{$|\Sigma|$} &
  \multicolumn{1}{c|}{$a_{max}$} &
  \multicolumn{1}{c|}{$\overbar{a}$} 
  & \multicolumn{1}{c|}{$|Index|$} 
  \\ \hline \hline
\HC & 1,290      & 331       & 2      & 81     & 34.8   & 178KB \\ \hline
\MA & 73,851     & 5,444     & 1,456  & 1,784  & 24.2   & 2.1MB \\ \hline
\CH & 327        & 7,818     & 9      & 5      & 2.3   & 109KB  \\ \hline
\CP & 242        & 12,704    & 11     & 5      & 2.4   & 190KB  \\ \hline
\SB & 294        & 20,584    & 2      & 99     & 8.0   & 2.1MB  \\ \hline
\HB & 1,494      & 52,960    & 2      & 399    & 20.5   & 15.5MB \\ \hline
\WT & 88,860     & 65,507    & 11     & 25     & 6.6   & 6.8MB  \\ \hline
\TC & 172,738    & 212,483   & 160    & 85     & 4.1   & 7.8MB  \\ \hline
\SA & 15,211,989 & 1,103,193 & 56,502 & 61,315 & 23.7   & 419.7MB\\ \hline
\AR & 2,268,264  & 4,239,108 & 29     & 9,350  & 17.1   & 998.6MB \\ \hline
\end{tabular}
}
\end{table}

\stitle{Queries.}
We use randomly sampled subhypergraphs from the data hypergraphs as our queries. Therefore, for each query hypergraph, there must exist at least one embedding in the corresponding data hypergraph.
Specifically, we perform a random walk in the data hypergraph to generate subhypergraphs with the given number of hyperedges whose number of vertices is in the range of $[|V|_{min},|V|_{max}]$. The settings of our queries are presented in \reftable{queries}. We generate $20$ random queries for each setting.
\iffullpaper
\red{The random queries vary from low to high selectivity, we draw the distributions of the number of embeddings for each query setting in box plots in \reffig{num_embeddings}.}

\begin{figure*}[ht]
    \centering
    \begin{subfigure}[b]{0.25\textwidth}
        \centering
        \includegraphics[width=0.95\linewidth]{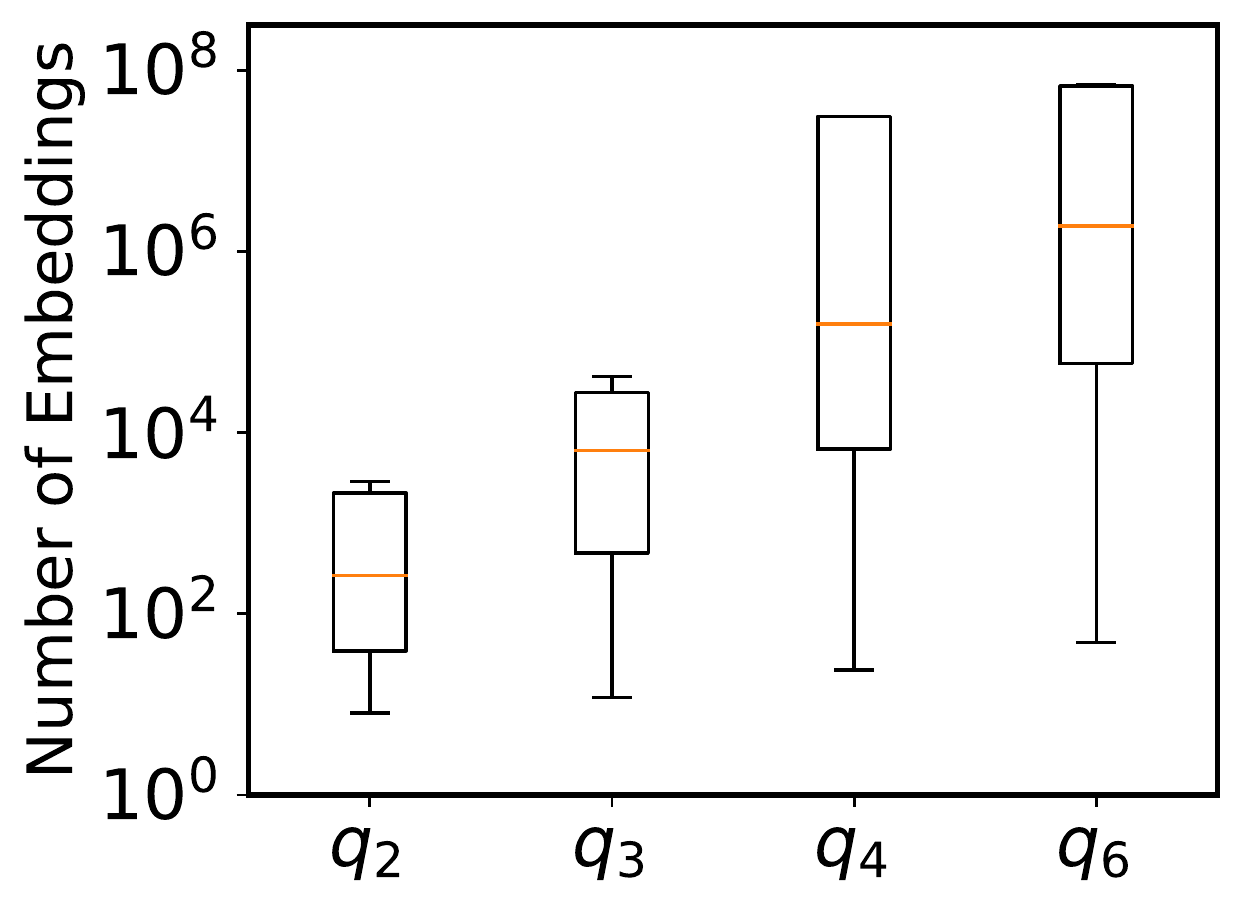}
        \caption{\HC}
    \end{subfigure}%
    ~
    \begin{subfigure}[b]{0.25\textwidth}
        \centering
        \includegraphics[width=0.95\linewidth]{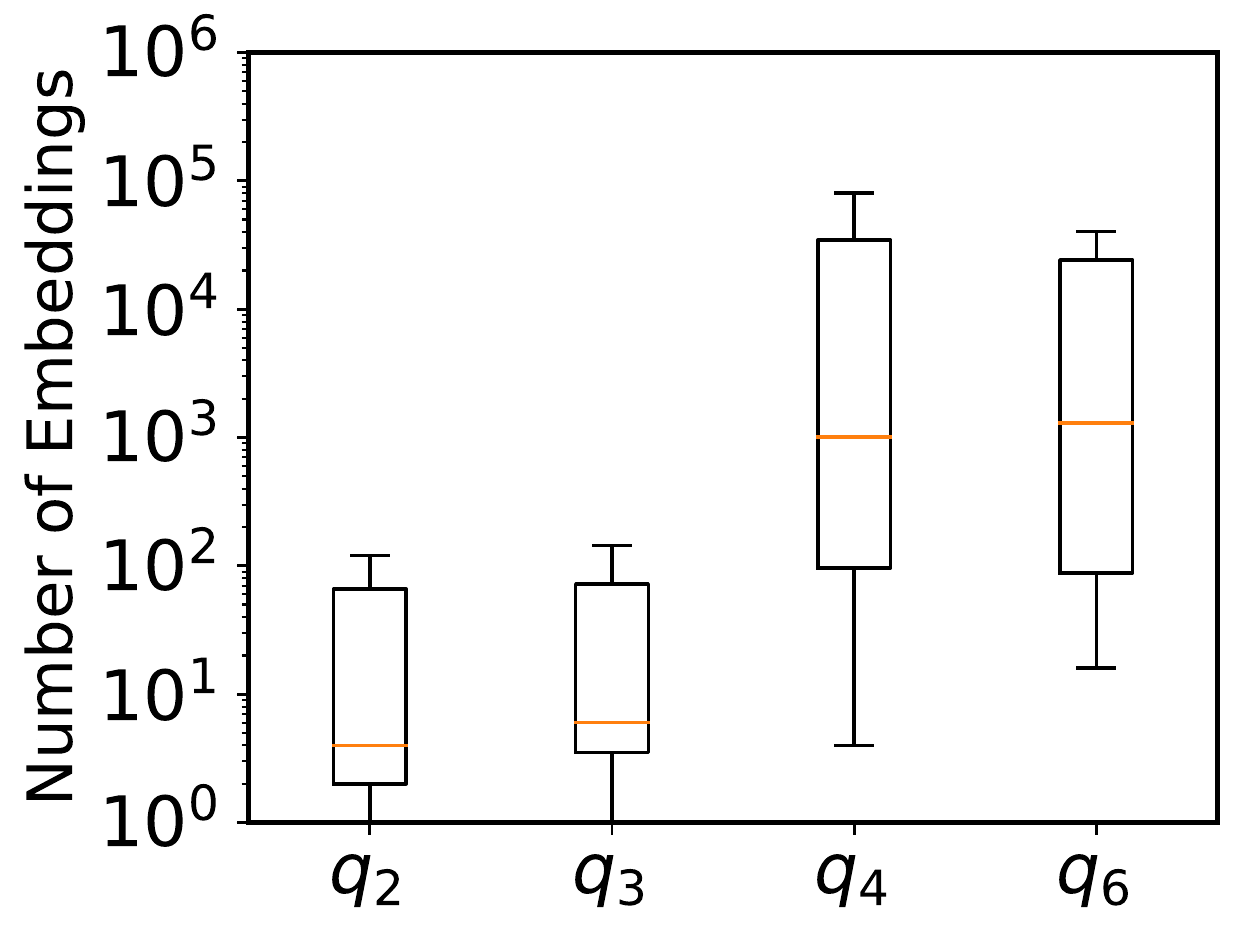}
        \caption{\MA}
    \end{subfigure}%
    ~
    \begin{subfigure}[b]{0.25\textwidth}
        \centering
        \includegraphics[width=0.95\linewidth]{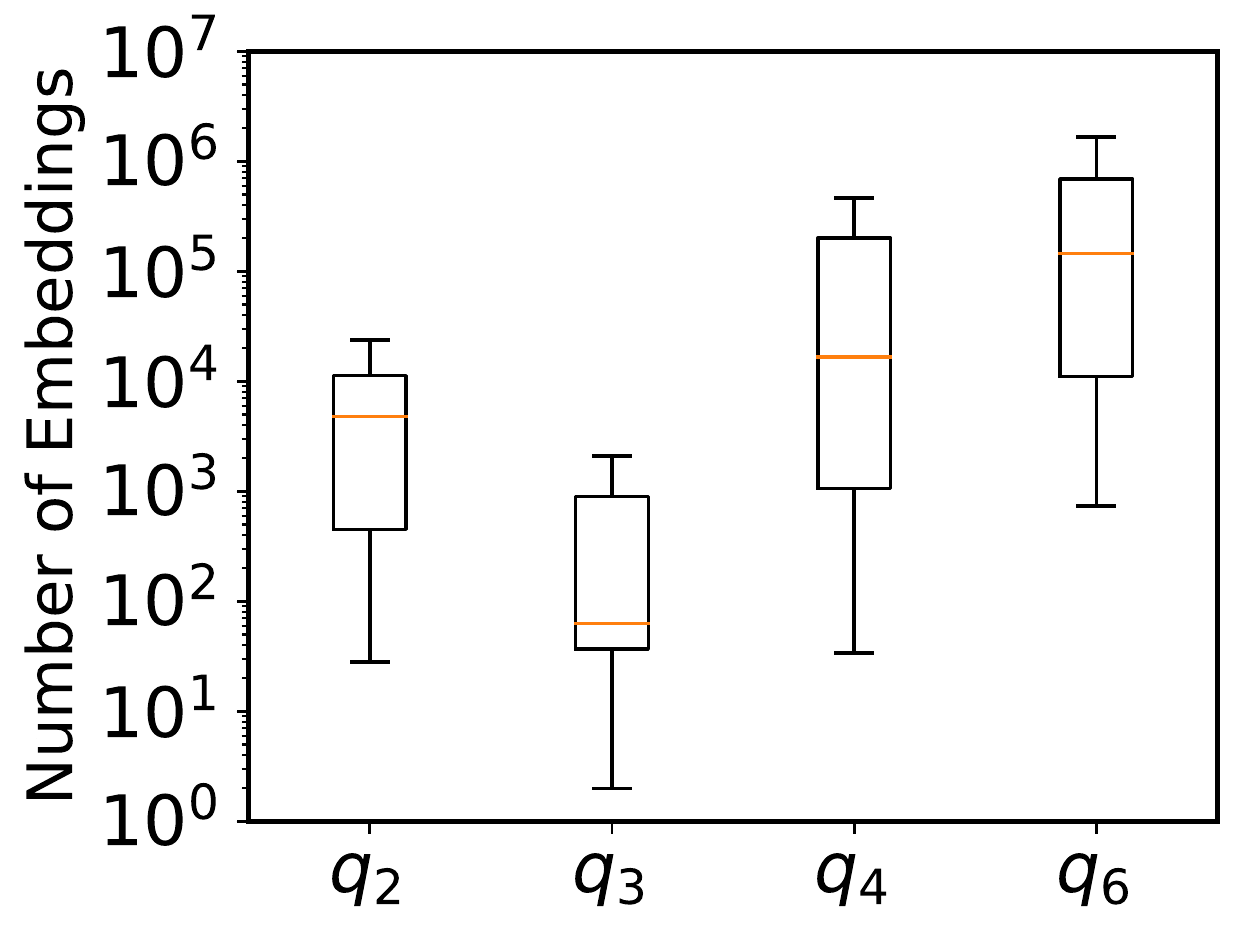}
        \caption{\CH}
    \end{subfigure}%
    \\
    \begin{subfigure}[b]{0.25\textwidth}
        \centering
        \includegraphics[width=0.95\linewidth]{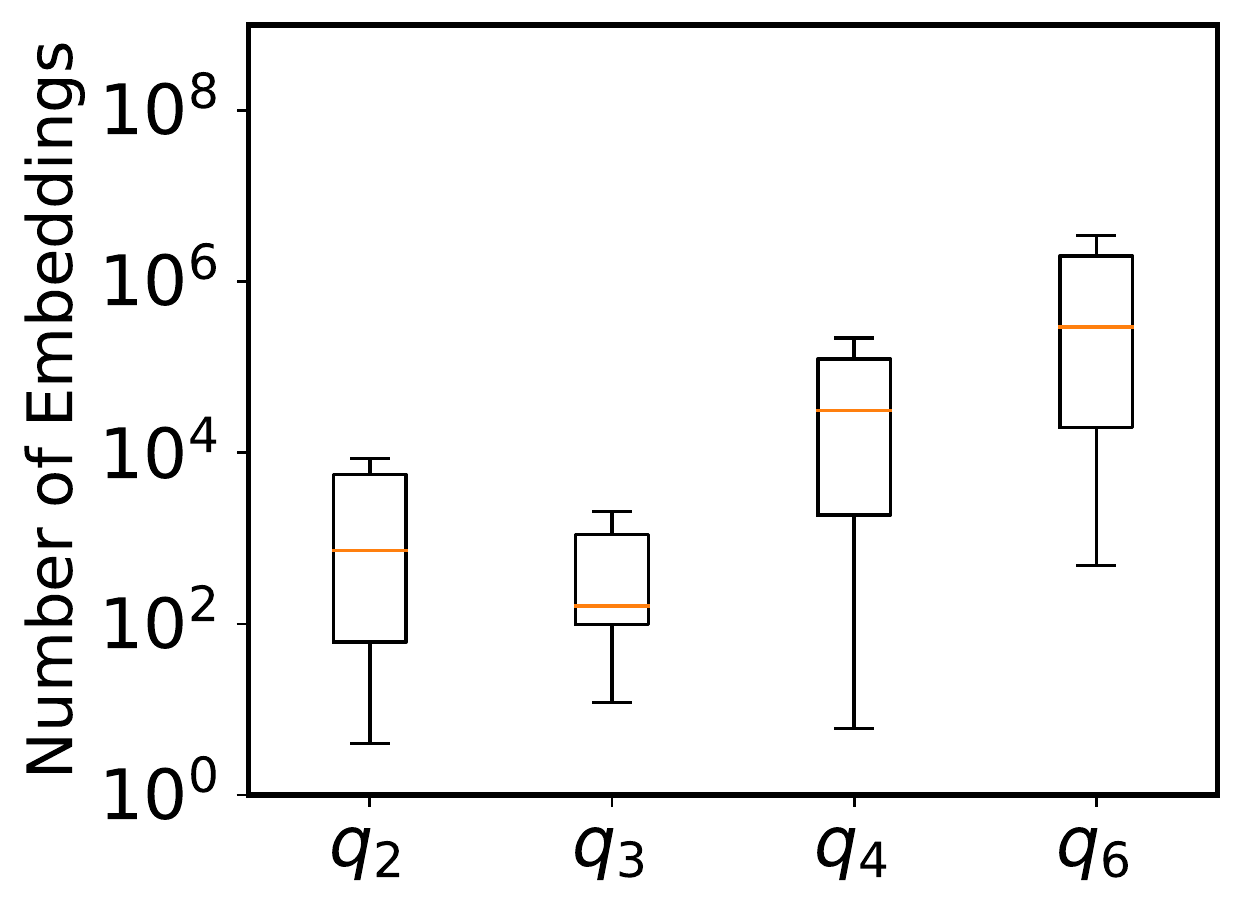}
        \caption{\CP}
    \end{subfigure}%
    ~
    \begin{subfigure}[b]{0.25\textwidth}
        \centering
        \includegraphics[width=0.95\linewidth]{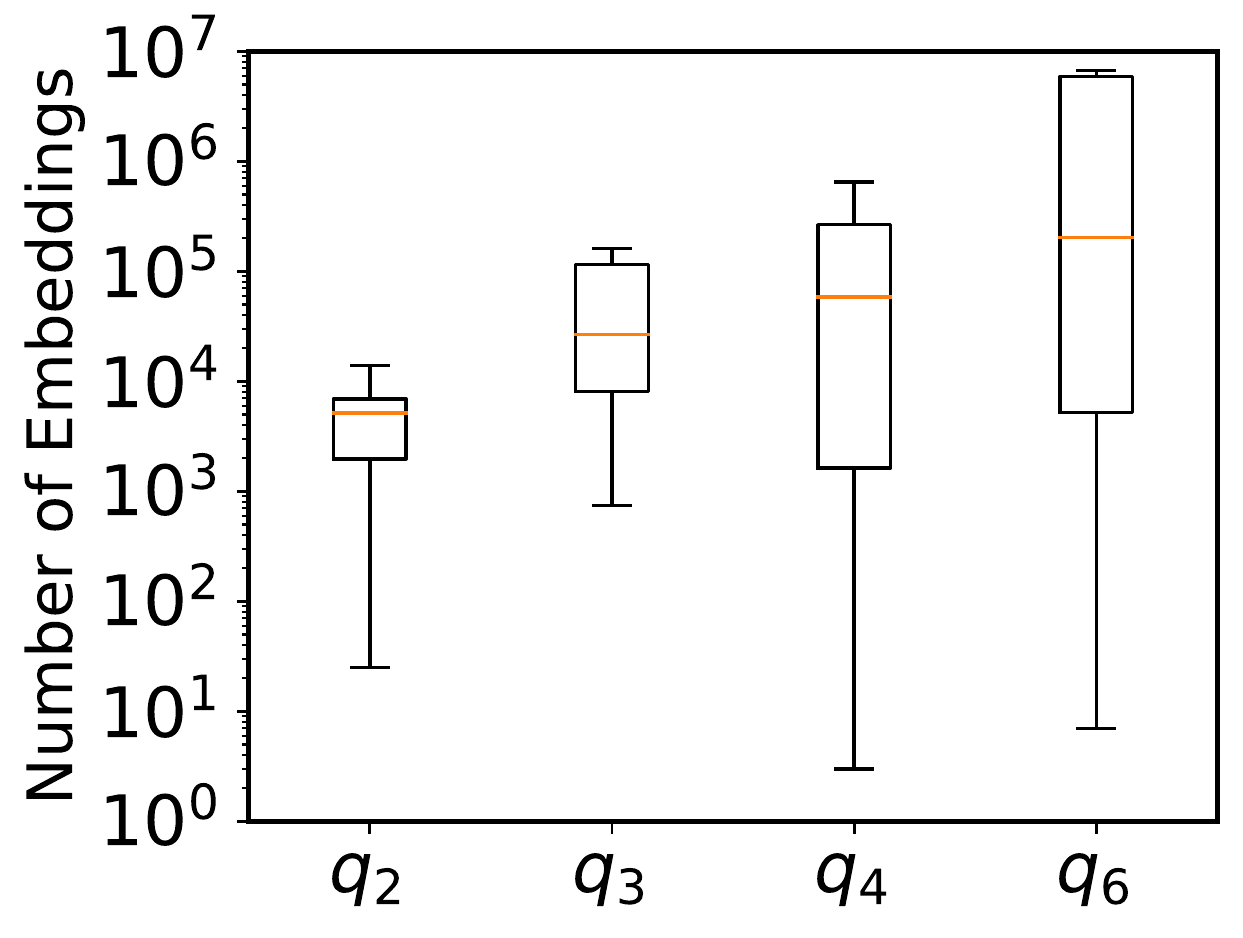}
        \caption{\SB}
    \end{subfigure}%
    ~
    \begin{subfigure}[b]{0.25\textwidth}
        \centering
        \includegraphics[width=0.95\linewidth]{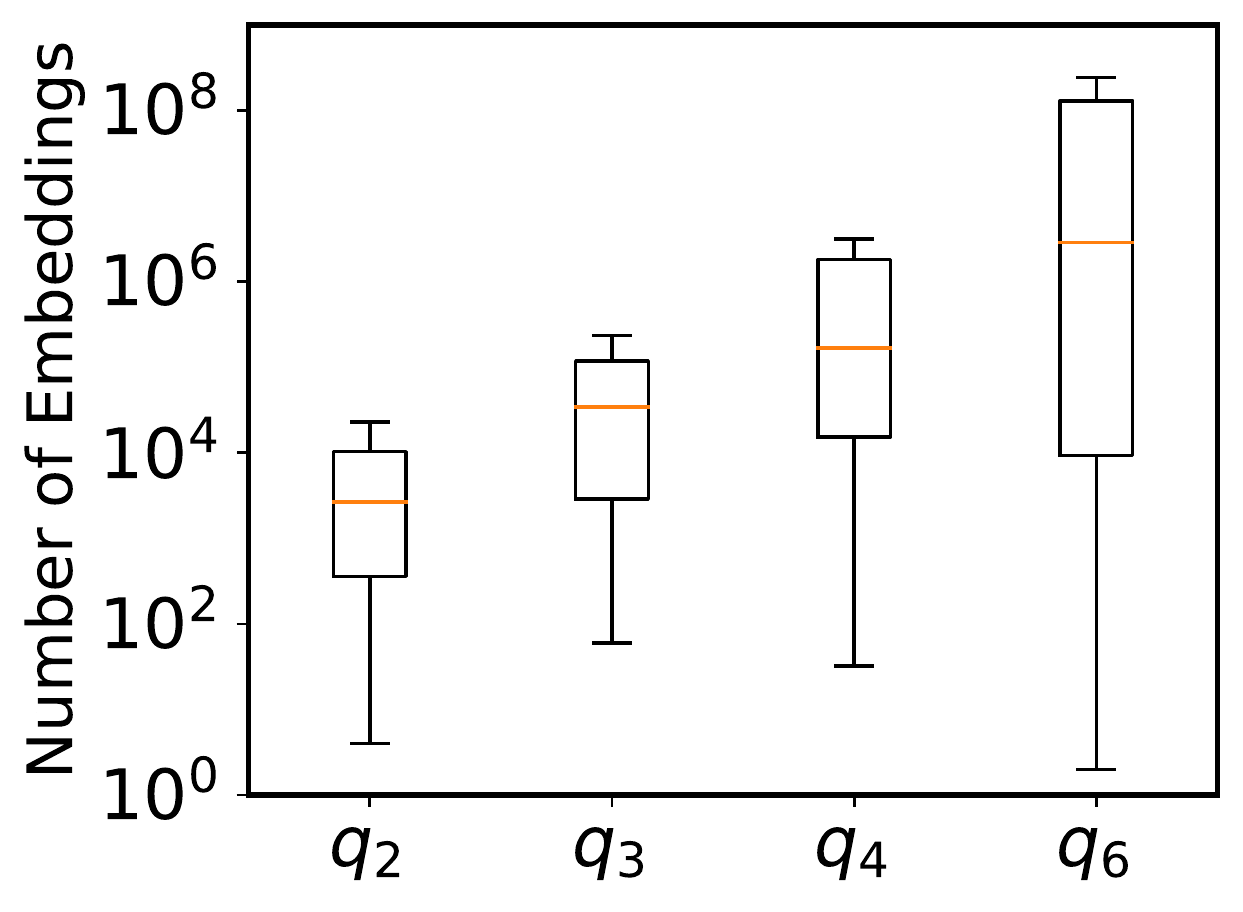}
        \caption{\HB}
    \end{subfigure}%
    \\
    \begin{subfigure}[b]{0.25\textwidth}
        \centering
        \includegraphics[width=0.95\linewidth]{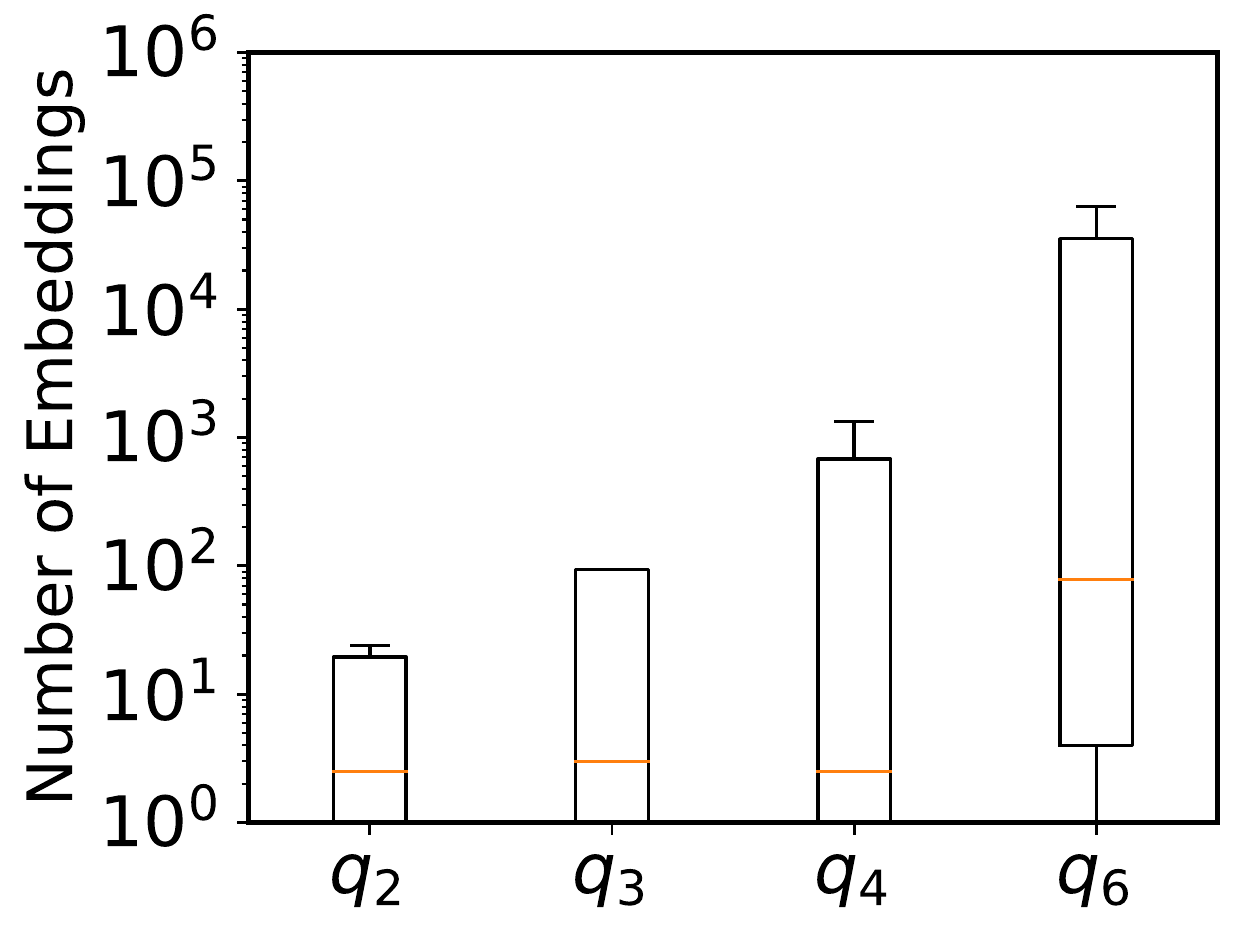}
        \caption{\WT}
    \end{subfigure}%
    ~
    \begin{subfigure}[b]{0.25\textwidth}
        \centering
        \includegraphics[width=0.95\linewidth]{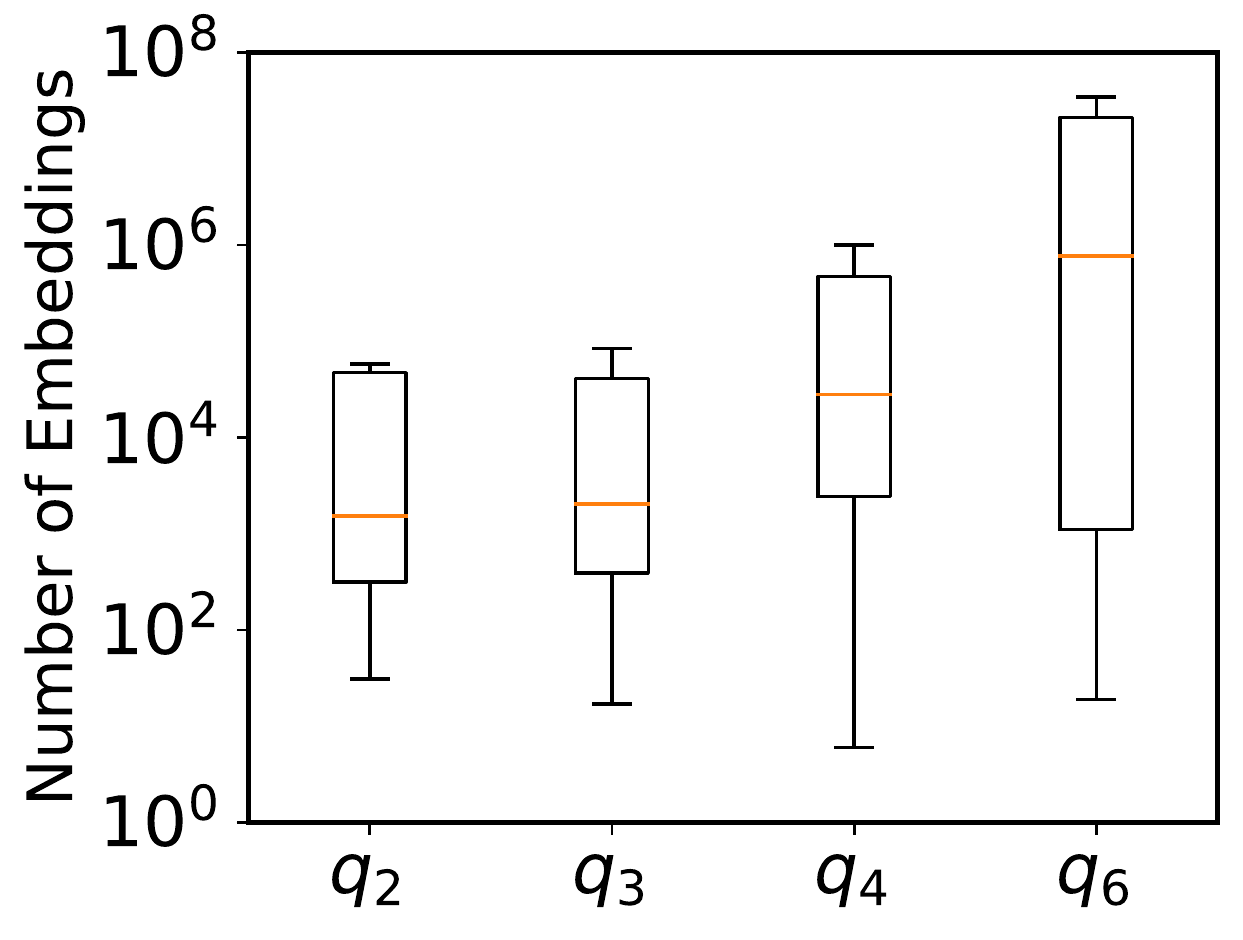}
        \caption{\TC}
    \end{subfigure}%
    ~
    \begin{subfigure}[b]{0.25\textwidth}
        \centering
        \includegraphics[width=0.95\linewidth]{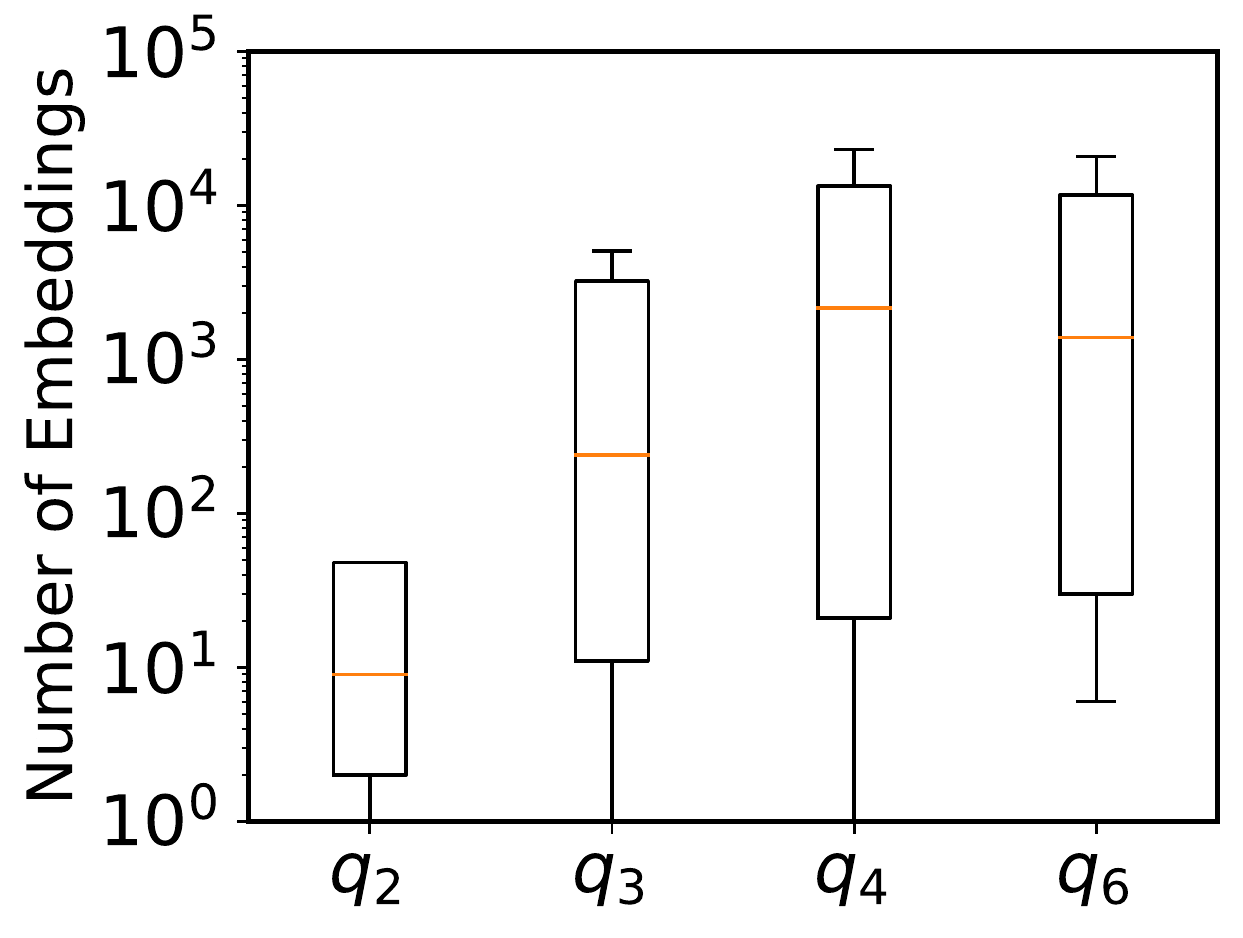}
        \caption{\SA}
    \end{subfigure}%
    \caption{\red{Number of Embeddings Distributions}}
    \label{fig:num_embeddings}
\end{figure*}
\else
\red{The random queries vary from low to high selectivity, detailed statistics of all queries are presented in the full version of this paper \cite{fullpaper}.}
\fi

\stitle{Metrics.} We measure the average elapsed time of each query type. Each query is executed three times for more precise measurement. 
\red{We count the number of embeddings for all compared methods instead of outputting them to eliminate I/O costs.}
Since \cflh, \dafh and \cecih fail on almost all queries on our largest dataset, \AR, we only use it for parallel evaluation of \hm (\refsec{parallel_eval}). For other datasets, we use them in single-thread evaluation (\refsec{single_eval}).
For single-thread comparisons (\refsec{single_eval}), we set a timeout of $1$ hour for all queries. The running time of out-of-time queries will be counted as $3600$ seconds when computing the average.

\begin{table*}[ht]
    \footnotesize
    \centering
    \begin{minipage}{.3\textwidth}
        \centering
        \caption{Table of Query Settings}
        \label{tab:queries}
        \scalebox{0.9}{
        \begin{tabular}{|c|c|c|c|c}
        \hline
        Query & $|E|$ & $|V|_{min}$ & $|V|_{max}$ \\ \hline \hline
        $q_2$ & 2     & 5           & 15          \\ \hline
        $q_3$ & 3     & 10          & 20          \\ \hline
        $q_4$ & 4     & 10          & 30          \\ \hline
        $q_6$ & 6     & 15          & 35          \\ \hline
        \end{tabular}
        }
    \end{minipage} %
    ~
    \begin{minipage}{.6\textwidth}
        \centering
        \caption{Query Completion Ratio (Single-thread)}
        \label{tab:complete}
        \scalebox{0.9}{
        \begin{tabular}{|c|cccccccccc|}
        \hline
        Algorithm &
          \multicolumn{1}{c|}{\HC} &
          \multicolumn{1}{c|}{\MA} &
          \multicolumn{1}{c|}{\CH} &
          \multicolumn{1}{c|}{\CP} &
          \multicolumn{1}{c|}{\SB} &
          \multicolumn{1}{c|}{\HB} &
          \multicolumn{1}{c|}{\WT} &
          \multicolumn{1}{c|}{\TC} &
          \multicolumn{1}{c|}{\SA} &
          Total \\ \hline
        \cflh &
          \multicolumn{4}{c|}{\multirow{4}{*}{\textbf{100\%}}} &
          \multicolumn{1}{c|}{56\%} &
          \multicolumn{1}{c|}{44\%} &
          \multicolumn{1}{c|}{76\%} &
          \multicolumn{1}{c|}{90\%} &
          \multicolumn{1}{c|}{99\%} &
          85\% \\ \cline{1-1} \cline{6-11} 
        \dafh &
          \multicolumn{4}{c|}{} &
          \multicolumn{1}{c|}{49\%} &
          \multicolumn{1}{c|}{43\%} &
          \multicolumn{1}{c|}{75\%} &
          \multicolumn{1}{c|}{90\%} &
          \multicolumn{1}{c|}{99\%} &
          84\% \\ \cline{1-1} \cline{6-11} 
        \cecih &
          \multicolumn{4}{c|}{} &
          \multicolumn{1}{c|}{50\%} &
          \multicolumn{1}{c|}{43\%} &
          \multicolumn{1}{c|}{75\%} &
          \multicolumn{1}{c|}{90\%} &
          \multicolumn{1}{c|}{99\%} &
          84\% \\ \cline{1-1} \cline{6-11} 
        \rapid &
          \multicolumn{4}{c|}{} &
          \multicolumn{1}{c|}{45\%} &
          \multicolumn{1}{c|}{44\%} &
          \multicolumn{1}{c|}{75\%} &
          \multicolumn{1}{c|}{86\%} &
          \multicolumn{1}{c|}{99\%} &
          83\% \\ \hline
        \textbf{\hm} &
          \multicolumn{10}{c|}{\textbf{100\%}} \\ \hline
        \end{tabular}
        }
    \end{minipage}
\end{table*}

\subsection{Single-thread Comparisons} \label{sec:single_eval}

In this subsection, we evaluate \hm in a single-thread environments
We use all the datasets except \AR as data hypergraphs as discussed before. 

\begin{figure}[t]
    \centering
    \includegraphics[width=0.8\linewidth]{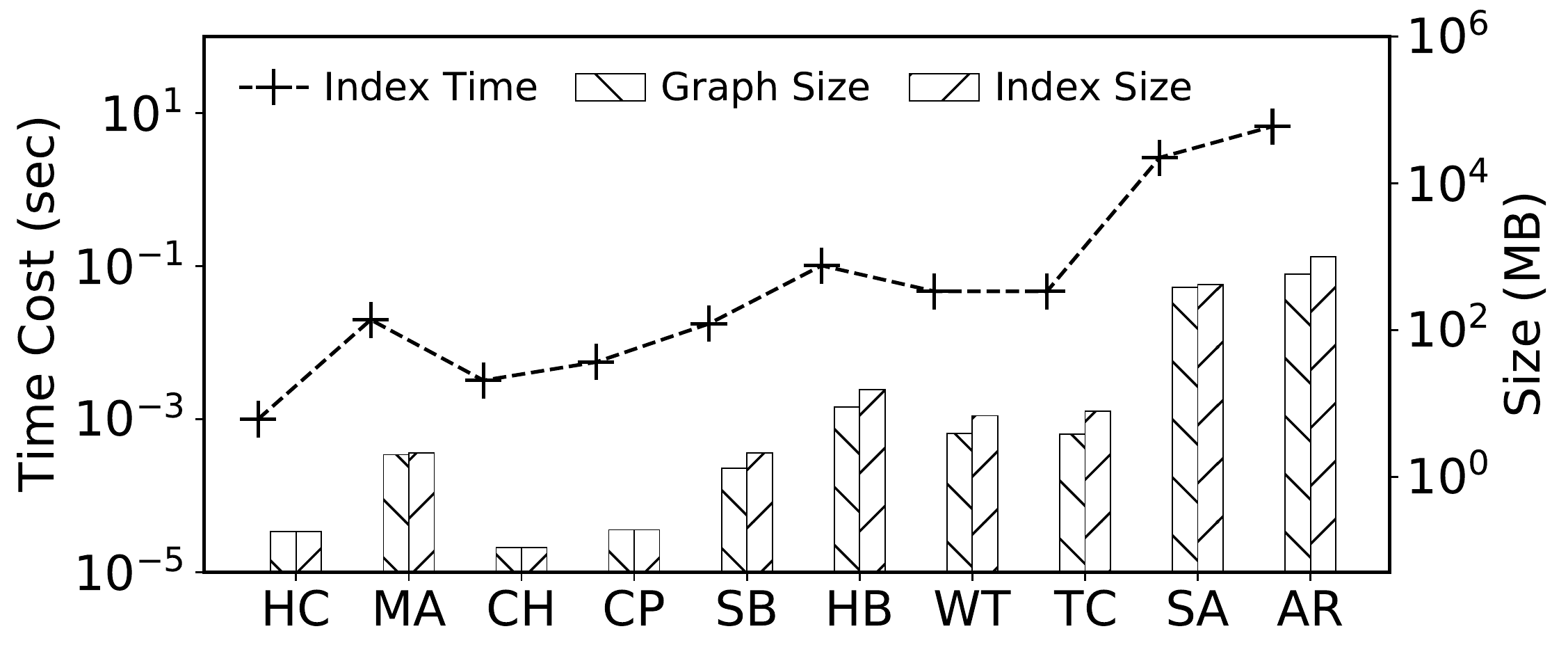}  
    \caption{\red{Building Time and Size of Index.}}
    \label{fig:index}
\end{figure}

\stitle{\red{Exp-1: Index Building.}}
\red{
We first evaluate the proposed inverted hyperedge index. \reffig{index} shows the time of building the index, the size of the graph and the size of the index (in MB). 
The time of index building is extremely fast. It takes only around $6.7$ seconds to build the index even for the largest dataset \AR. As for index size, we observe that the index size is similar to the original graph size, which confirms our size analysis in \refsec{inverted_index}.
}

\stitle{Exp-2: Overall Comparisons.} 
We then compare \hm with \cflh, \cecih, \dafh, and \rapid to verify the efficiency of our matching by hyperedges framework (\refsec{matching}).
The results are shown in \reffig{all_cases}.
As demonstrated, \hm significantly outperforms \cflh, \dafh, \cecih, and \rapid in all cases, with an average speedup of $5\times10^4$, $1\times10^5$ and $7\times10^5$, and $1\times10^6$ respectively.
Especially for data hypergraphs with high average arity $\overbar{a}$ including \HC, \MA, \HB, and \SA, \hm outperforms \cflh, \dafh and \cecih by up to $6$ orders of magnitude and \rapid by up to $7$ orders of magnitude.
This is because \hm can fully use the high-order information in hypergraphs to filter out unpromising hyperedges and reduce redundant computation.
%

In addition, \hm is the only algorithm that completes all queries within the time limit. The query completion rate 
is shown in \reftable{complete}. All algorithms run successfully for smaller datasets (i.e., \HC, \MA, \CH, and \CP). However, as the size of data hypergraphs grows, \rapid, \cflh, \dafh, and \cecih start to fail on some queries. 
This is because of the huge search space they have to explore.
%
%

\begin{figure*}[ht]
    \centering
    \begin{subfigure}[b]{\textwidth}
    \centering
        \includegraphics[height = 0.15in]{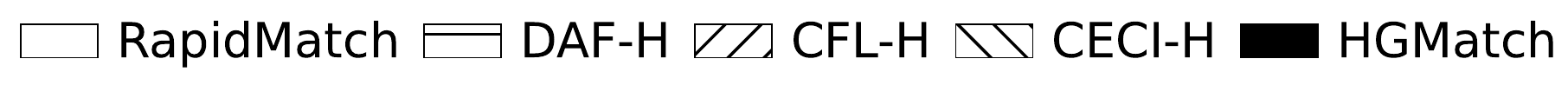}
    \end{subfigure}%
    \\
    \begin{subfigure}[b]{0.25\textwidth}
        \centering
        \includegraphics[width=0.95\linewidth]{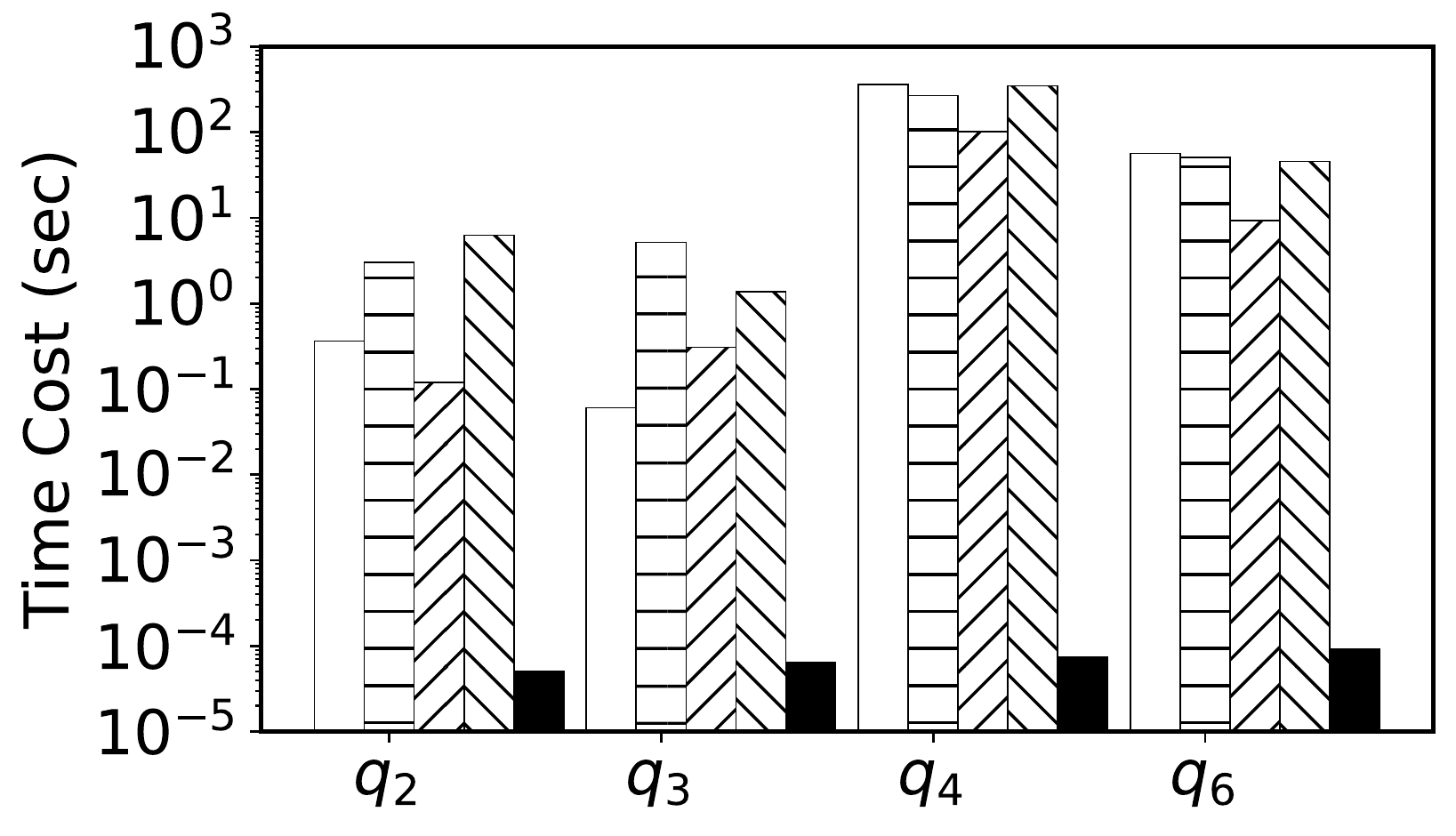}
        \caption{\HC}
    \end{subfigure}%
    ~
    \begin{subfigure}[b]{0.25\textwidth}
        \centering
        \includegraphics[width=0.95\linewidth]{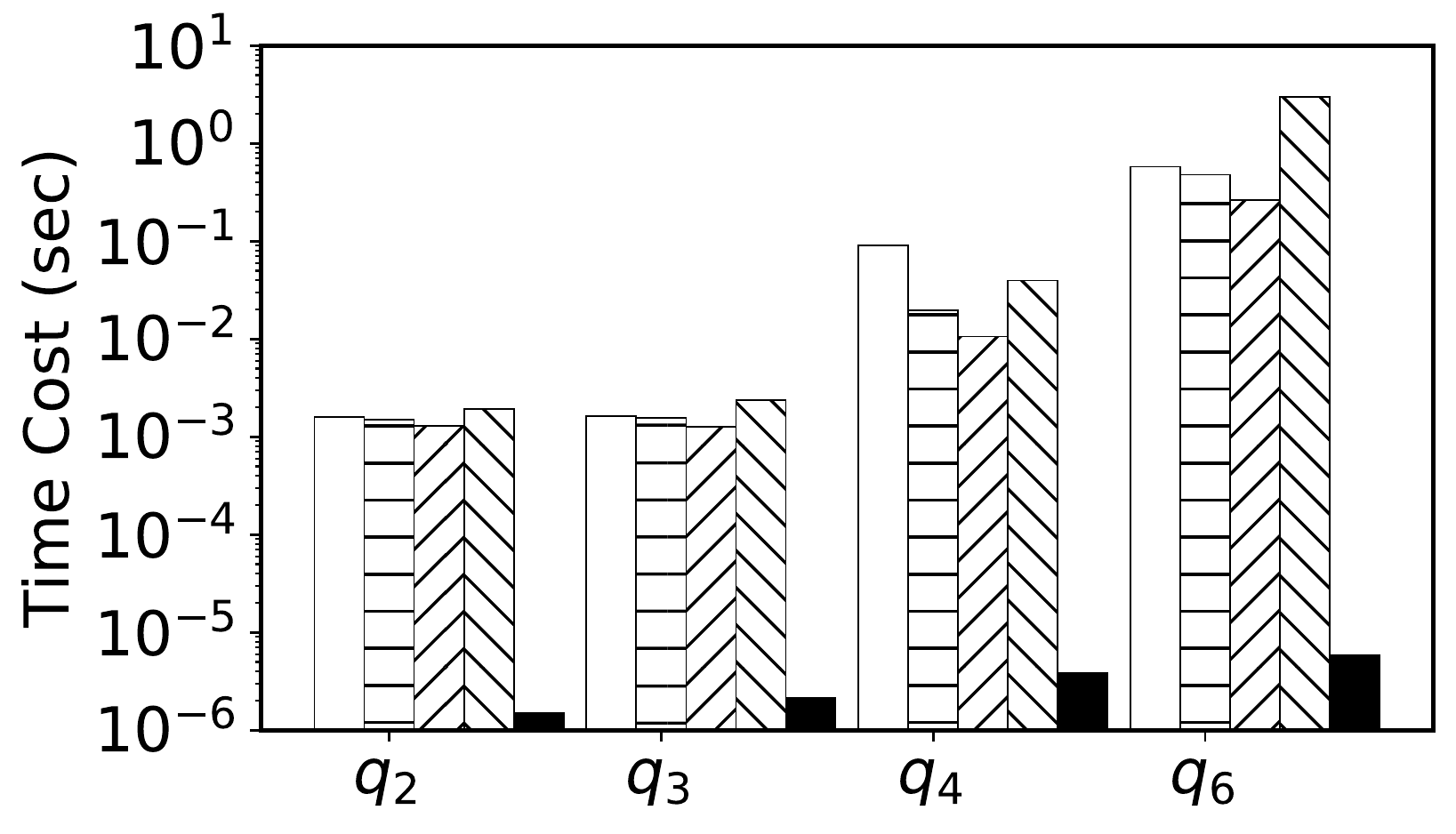}
        \caption{\MA}
    \end{subfigure}%
    ~
    \begin{subfigure}[b]{0.25\textwidth}
        \centering
        \includegraphics[width=0.95\linewidth]{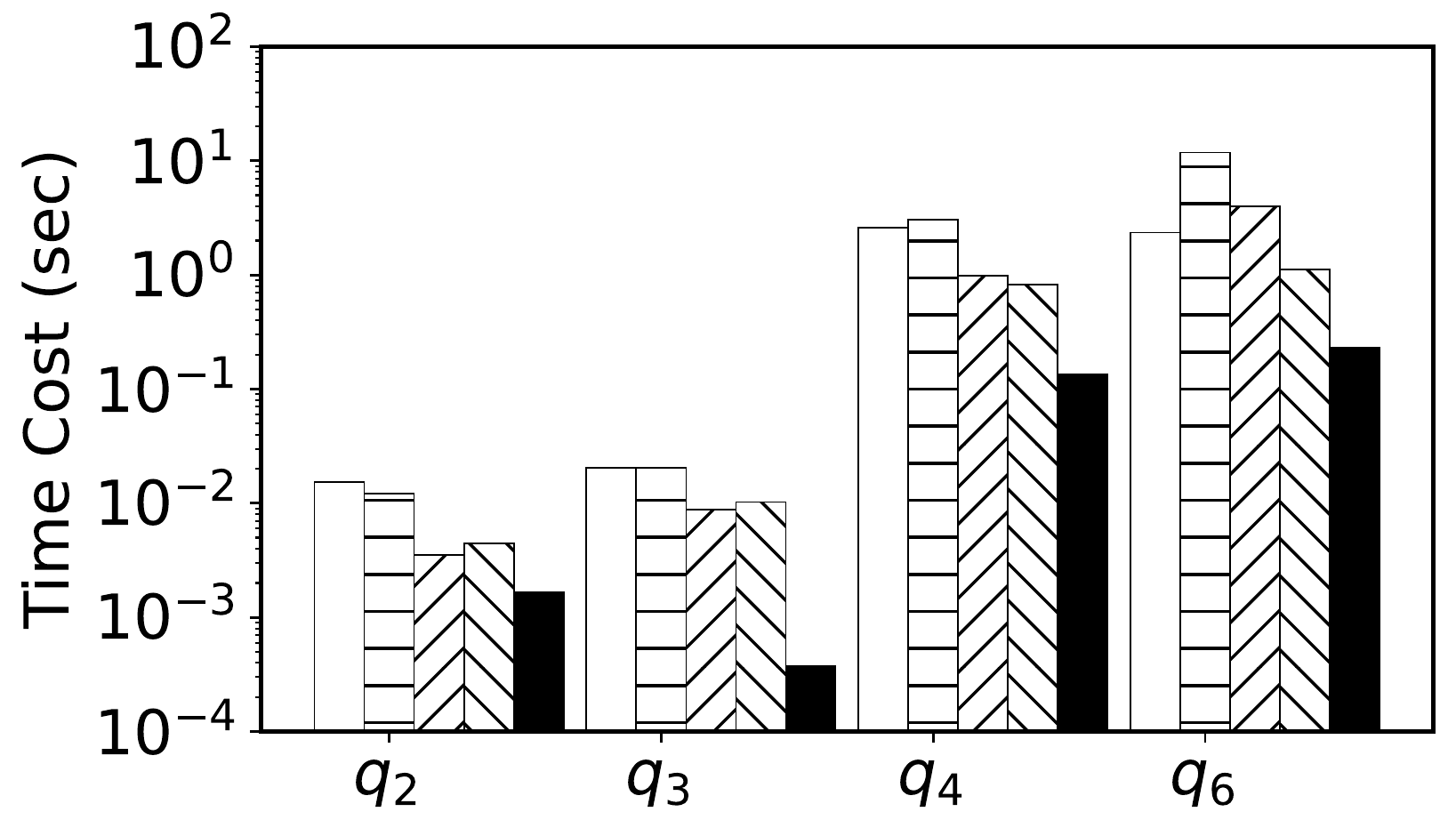}
        \caption{\CH}
    \end{subfigure}%
    \\
    \begin{subfigure}[b]{0.25\textwidth}
        \centering
        \includegraphics[width=0.95\linewidth]{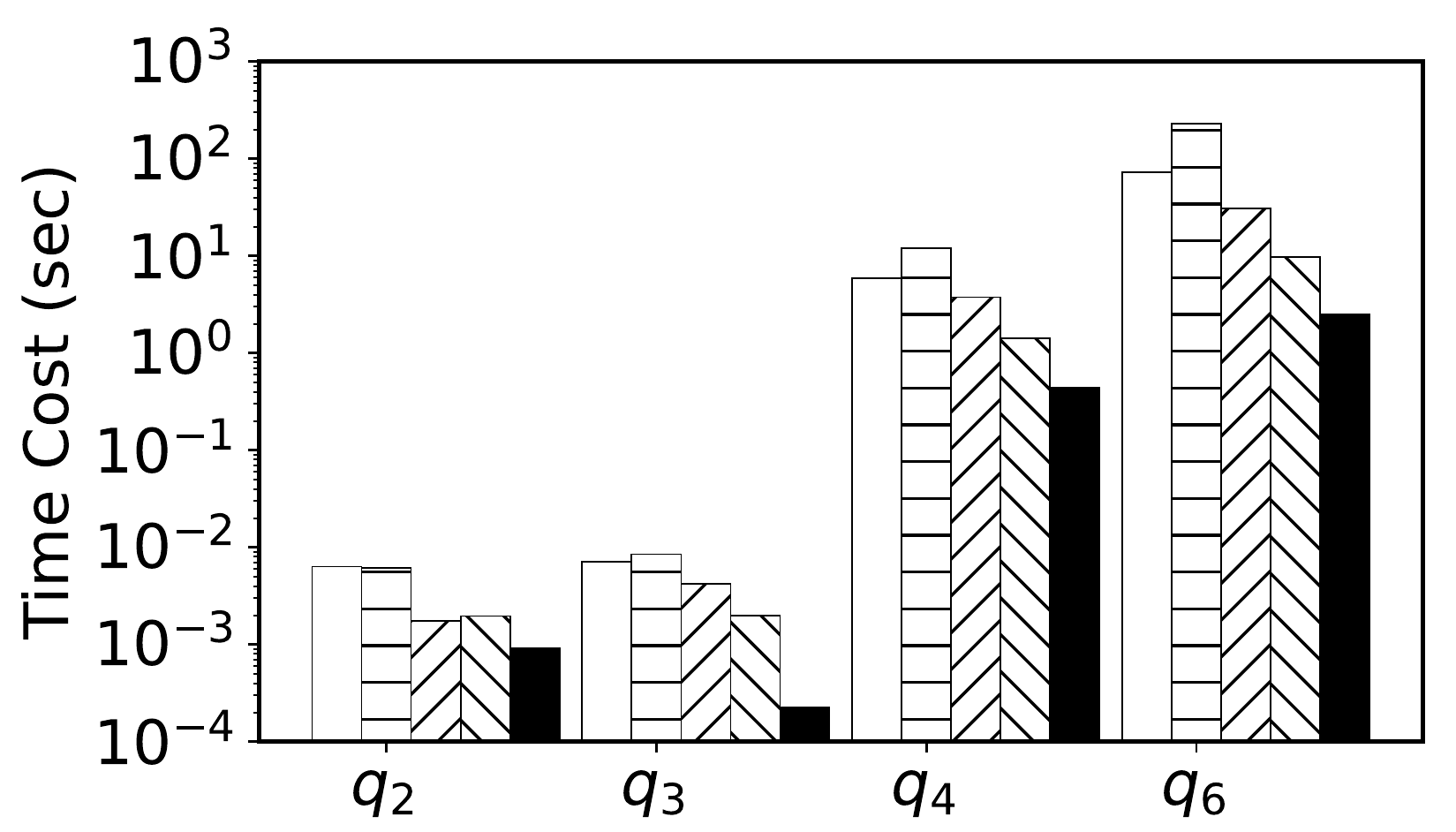}
        \caption{\CP}
    \end{subfigure}%
    ~
    \begin{subfigure}[b]{0.25\textwidth}
        \centering
        \includegraphics[width=0.95\linewidth]{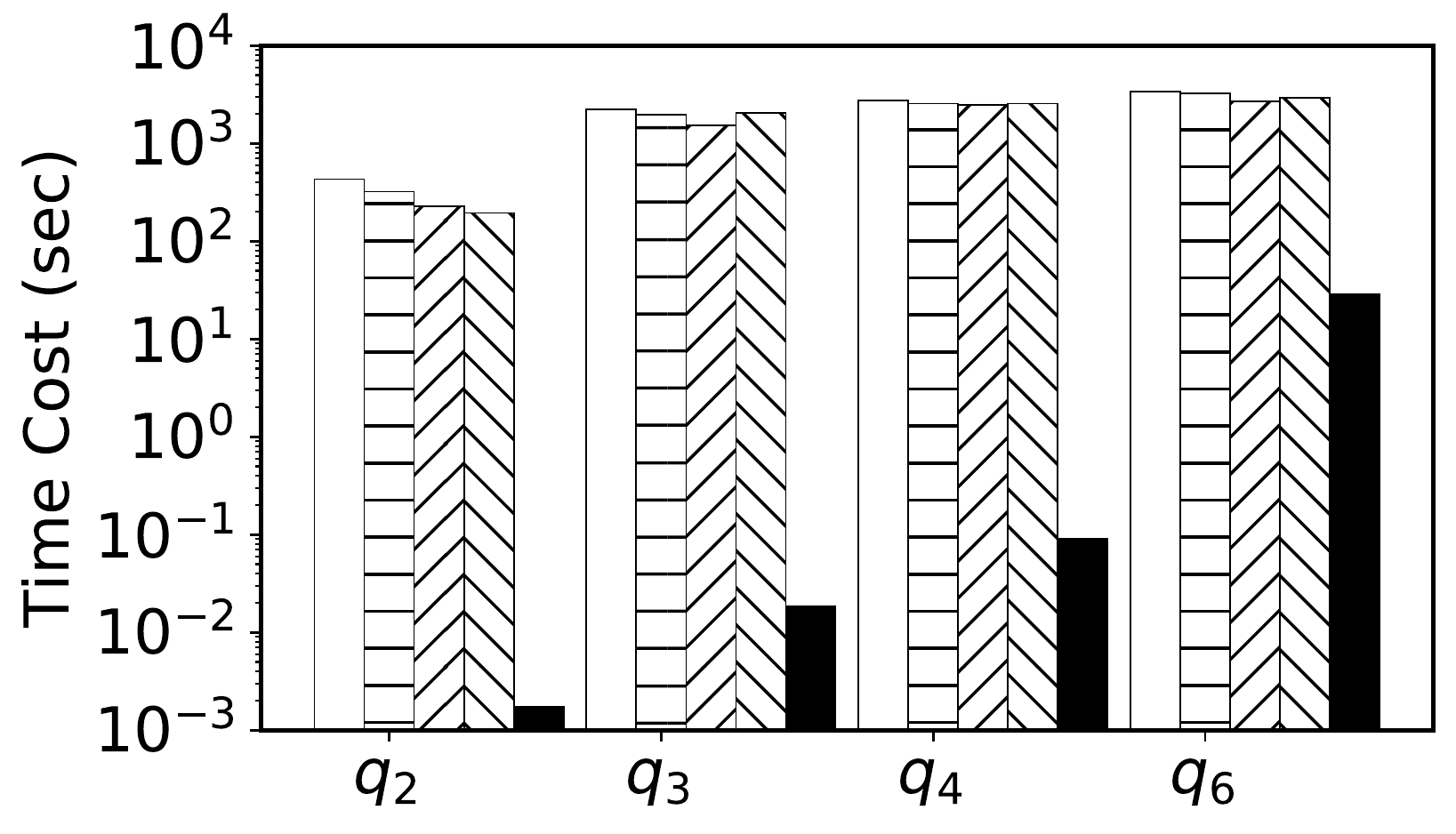}
        \caption{\SB}
    \end{subfigure}%
    ~
    \begin{subfigure}[b]{0.25\textwidth}
        \centering
        \includegraphics[width=0.95\linewidth]{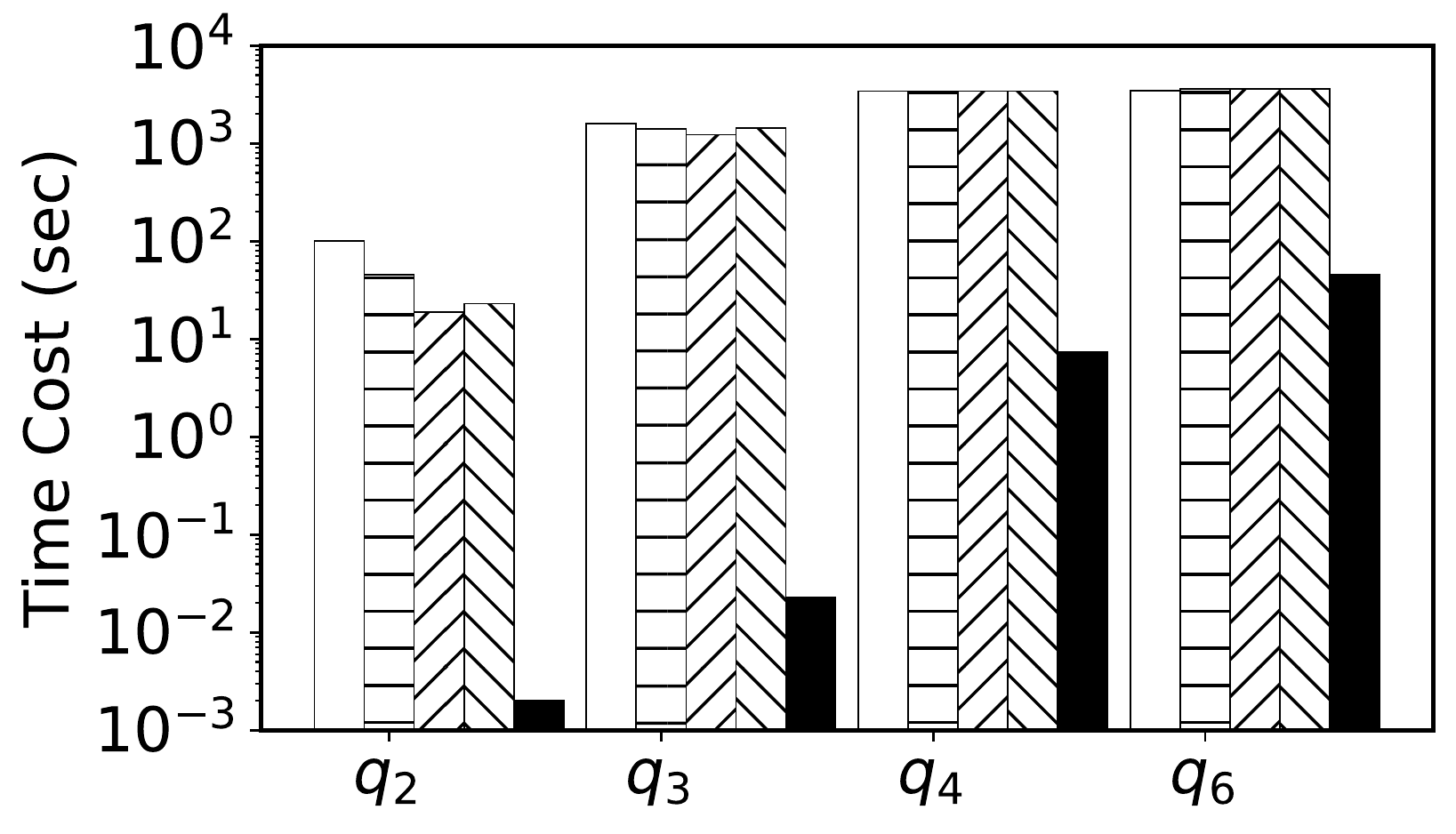}
        \caption{\HB}
    \end{subfigure}%
    \\
    \begin{subfigure}[b]{0.25\textwidth}
        \centering
        \includegraphics[width=0.95\linewidth]{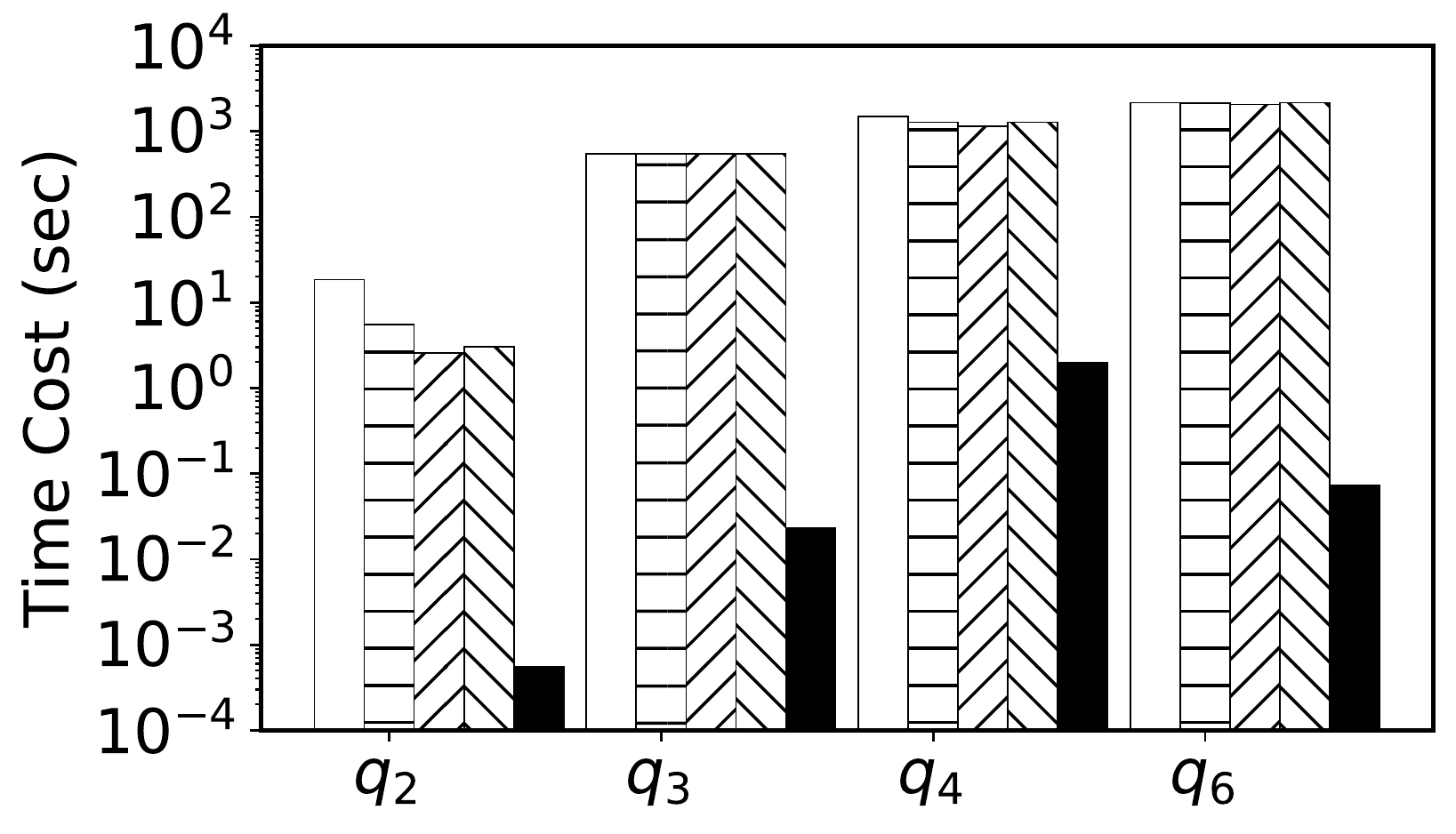}
        \caption{\WT}
    \end{subfigure}%
    ~
    \begin{subfigure}[b]{0.25\textwidth}
        \centering
        \includegraphics[width=0.95\linewidth]{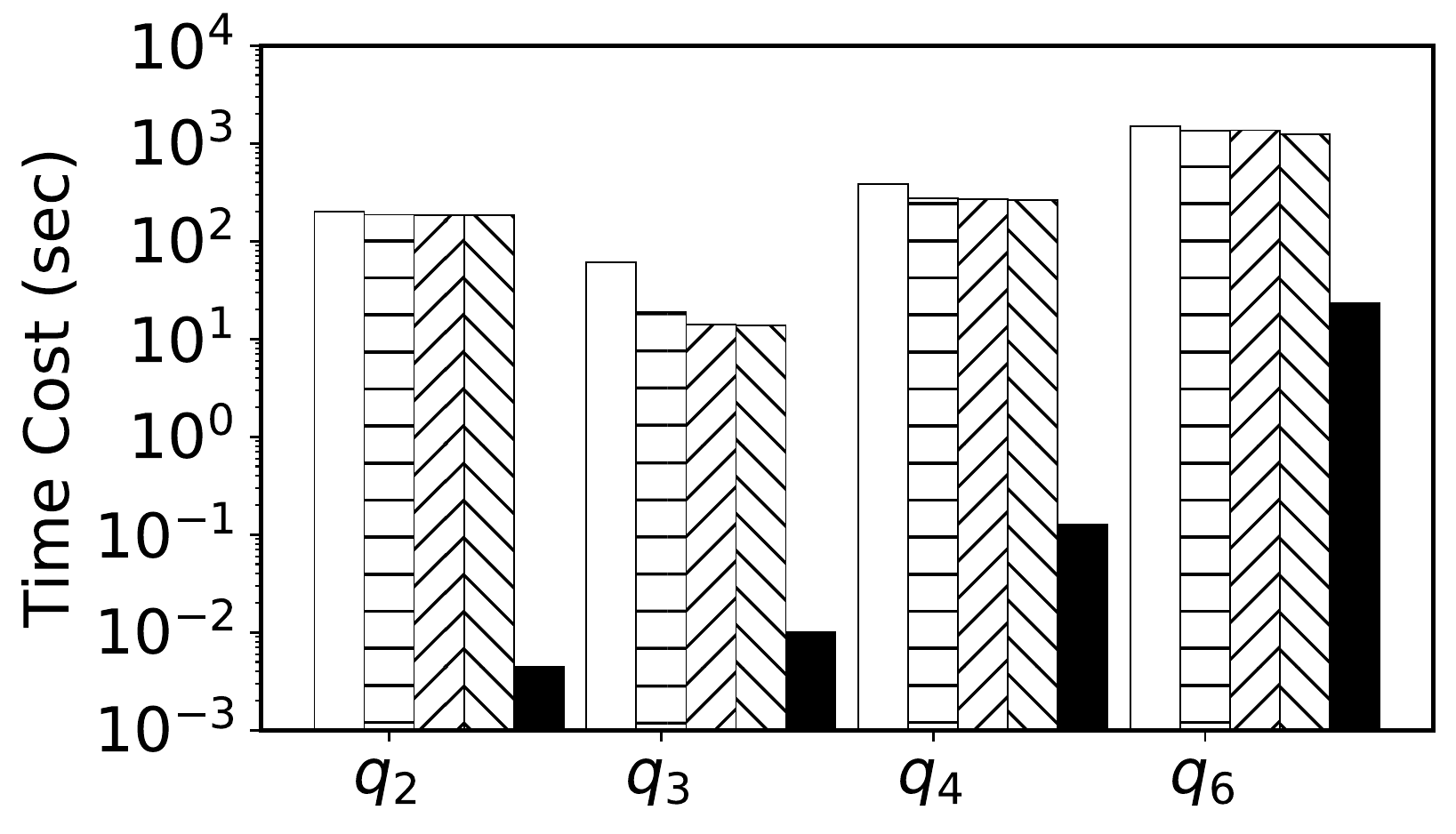}
        \caption{\TC}
    \end{subfigure}%
    ~
    \begin{subfigure}[b]{0.25\textwidth}
        \centering
        \includegraphics[width=0.95\linewidth]{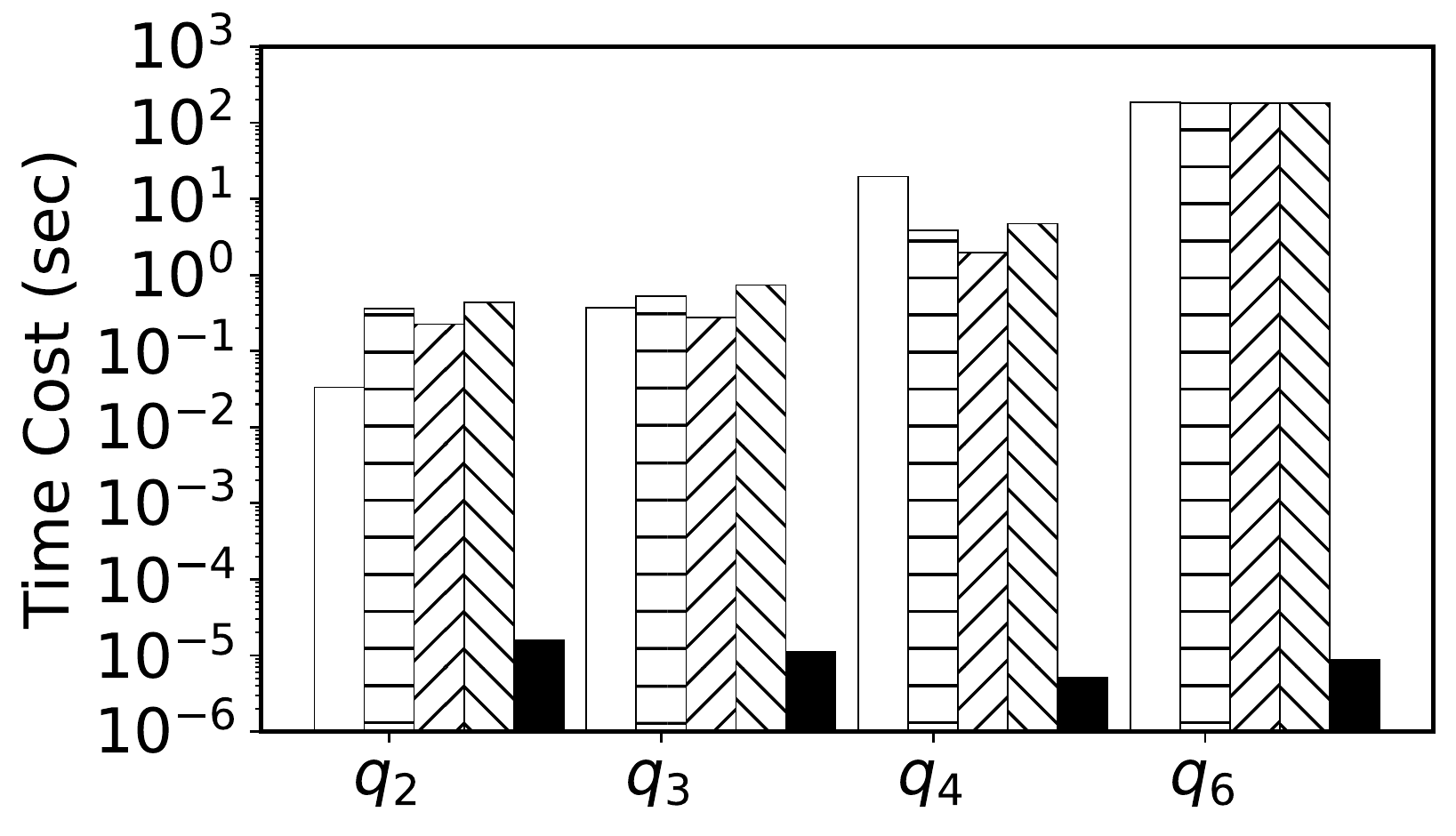}
        \caption{\SA}
    \end{subfigure}%
    \caption{Single-thread Comparisons.}
    \label{fig:all_cases}
\end{figure*}

\begin{figure*}[ht]
    \centering
    \begin{minipage}{0.35\linewidth}
        \centering
        \includegraphics[width=0.95\linewidth]{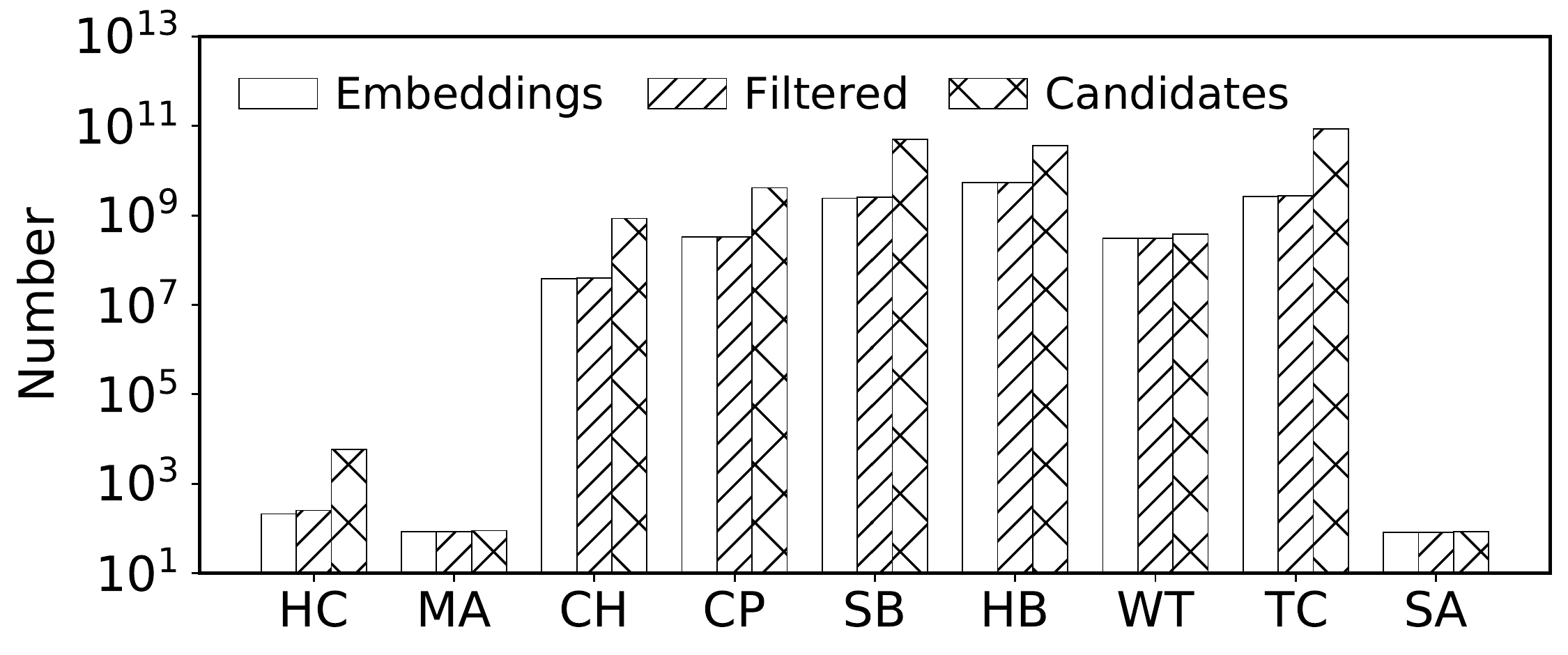}  
        \caption{Candidates Filtering.}
        \label{fig:filtering}
    \end{minipage}
    \begin{minipage}{0.35\linewidth}
         \centering
        \begin{subfigure}[b]{0.8\linewidth}
              \centering
              \includegraphics[height = 0.12in]{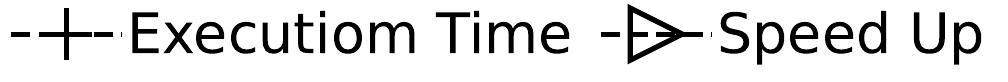}
        \end{subfigure}%
        \\
        \begin{subfigure}{.45\linewidth}
              \centering
              \includegraphics[width=\linewidth]{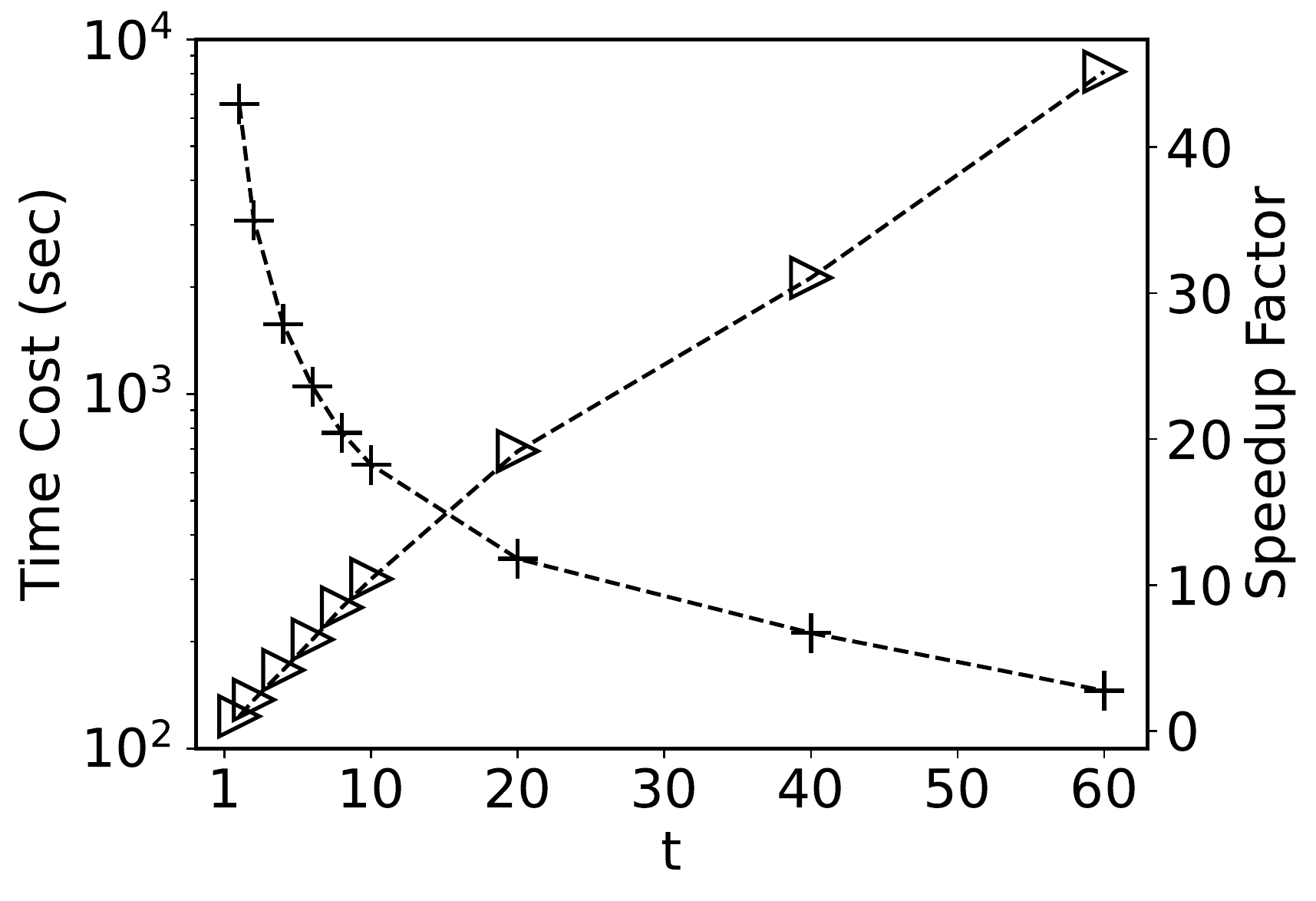}
              \caption{$q_{3}^{1}$}
        \end{subfigure} 
        ~
        \begin{subfigure}{.45\linewidth}
              \centering
              \includegraphics[width=\linewidth]{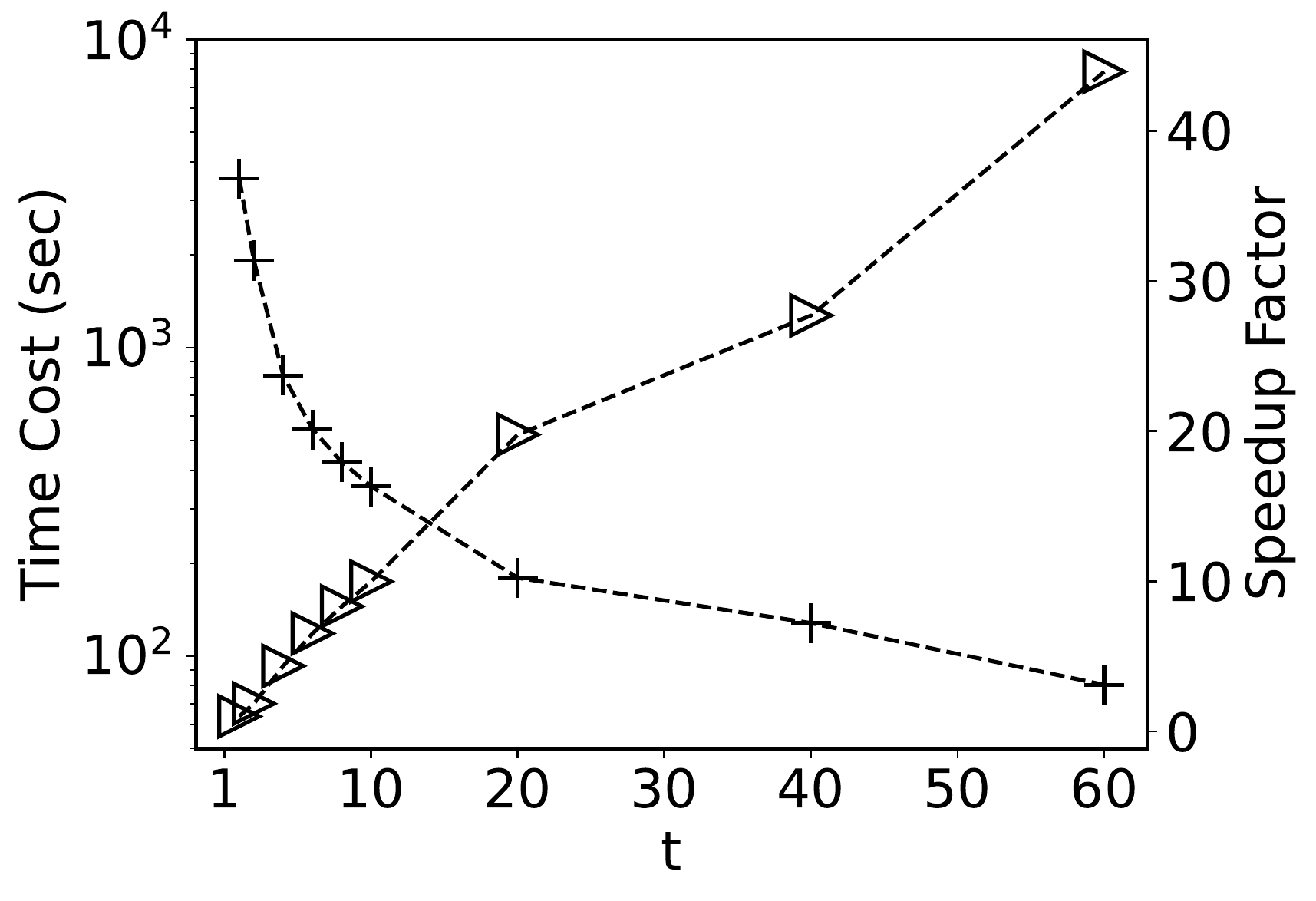}
              \caption{$q_{3}^{2}$}
        \end{subfigure}
        \caption{Vary Number of Threads.}
        \label{fig:scalability}
    \end{minipage}
\end{figure*}

\begin{figure*}[ht]
    \centering
    \begin{minipage}[b]{0.4\linewidth}
        \centering
        \includegraphics[width=0.8\linewidth]{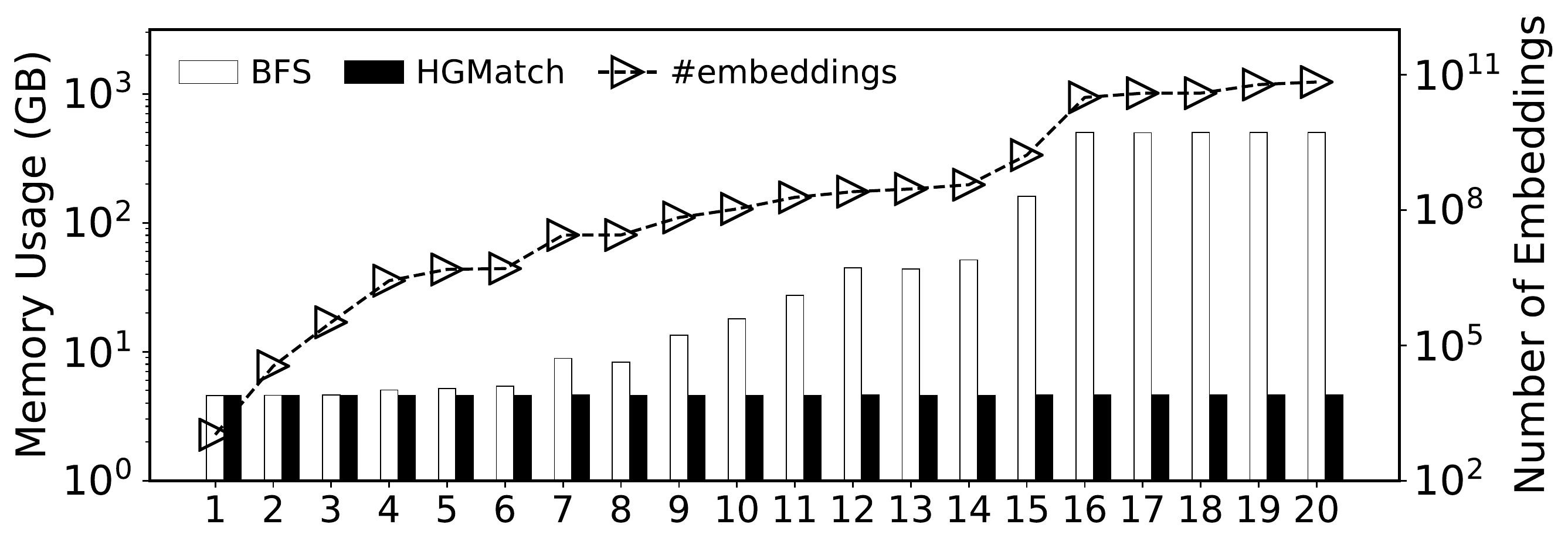}  
        \caption{Task-based Scheduling.}
        \label{fig:mem}
    \end{minipage}
    \begin{minipage}[b]{0.4\linewidth}
        \centering
        \includegraphics[width=0.8\linewidth]{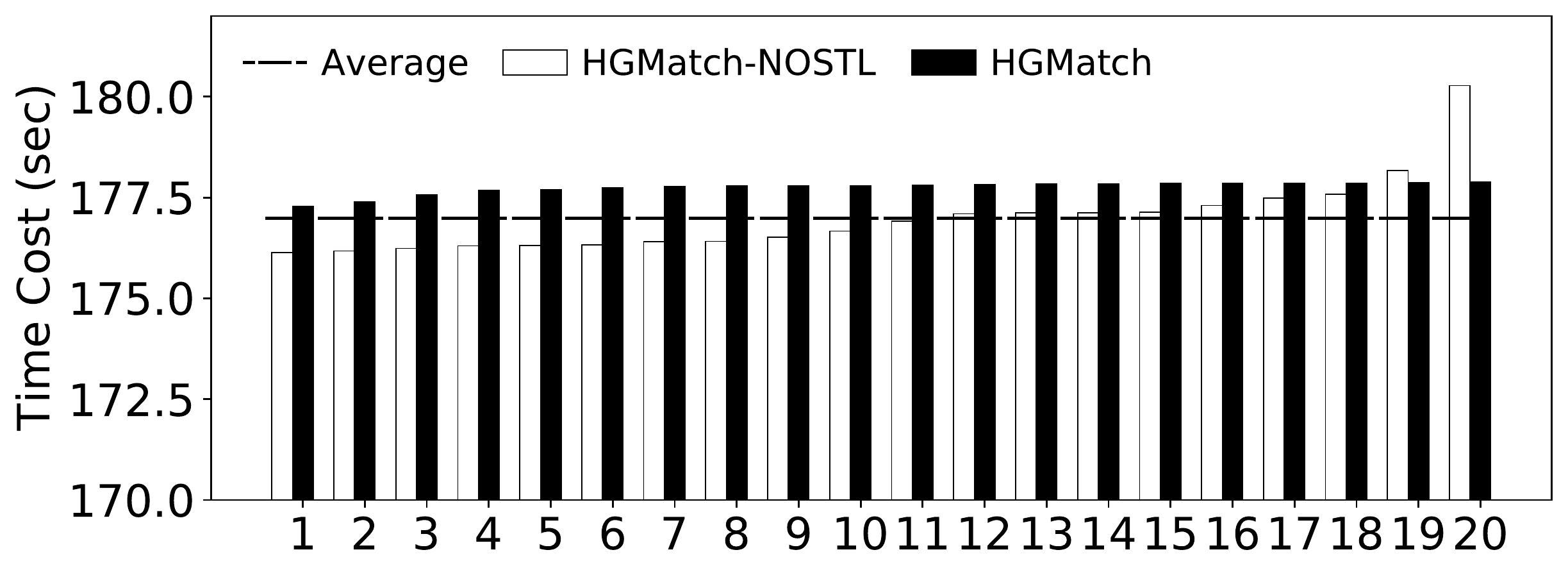}  
        \caption{Work Stealing.}
    \label{fig:work-stealing}
    \end{minipage}
    
\end{figure*}

\stitle{Exp-3: Candidates Filtering.}
In this experiment, we evaluate the pruning power of \hm's candidate generation (\refsec{cand_gen}) and embedding validation process (\refsec{emb_vali}). The results are given in \reffig{filtering}, where we count the number of true embeddings (denoted as `Embeddings'), the number of candidates after applying the vertex number check in Observation \ref{val_rule_1} (denoted as `Filtered'), and the number of candidates generation using \refalg{cand_gen} (denoted as `Candidates'). We draw the total number of all candidates of all queries for each data hypergraph. For data hypergraphs \MA and \SA, there are almost no false positive candidates in the candidate set due to a large number of labels. 
%
For other datasets with fewer labels, the candidate generation method can generate more false positive embeddings. However, only after a fast check of the number of vertices in the partial embedding, \hm are able to filter out the most majority of unpromising results. We observe that $97\%$ of the filtered results are true positive embeddings.
This is because, again, \hm can fully utilise the high-order information in hyperedges to prune candidates.
The results reveal the pruning power of \hm, which finally contributes to the significant speedup of existing algorithms.

\subsection{Parallel Comparisons} \label{sec:parallel_eval}

In this subsection, we evaluate \hm in the multi-thread environment. We use our largest dataset, \AR, as the data hypergraph as described and the queries in $q_3$ as the default queries.
%
Note that the original authors implementation of \daf\footnote{\url{https://github.com/SNUCSE-CTA/DAF}} and \ceci\footnote{\url{https://github.com/iHeartGraph/ceci-release}} support parallel execution. However, we do not compare them in our parallel experiments due to the errors in their code that causes segmentation faults and/or reports wrong numbers of embeddings. 


\stitle{Exp-4: Scalability.} 
We conduct a scalability test of \hm by varying the number of threads used for parallel execution.
We present the results of $2$ random queries from $q_3$ with a large number of embeddings.
We denote the two queries as $q_{3}^{1}$ and $q_{3}^{2}$.
Specifically, $q_{3}^{1}$ has about $3.86\times10^{10}$ results and $q_{3}^{2}$ has about $2.53\times10^{8}$ results.
We vary the number of threads from $1$ to $60$.
The results are shown in \reffig{scalability}.
\hm demonstrates almost perfect linear scalability when the number of threads is equal or below $20$ (i.e., $20\times$ speedup when using $20$ thread), thanks to the highly optimised parallel execution engine and dynamic load balancing mechanism.
When the number of threads is beyond $20$, the speedup factor slightly decreases due to non-uniform memory access (NUMA) and hyper-threading in the CPUs of our machine (i.e., $2$ physical CPUs with hyper-threading). In the future, we will investigate NUMA optimisations of \hm.

\stitle{Exp-5: Scheduling.} In this experiment, we evaluate the memory usage of \hm to test the effectiveness of its task-based scheduler (\refsec{scheduler}). 
We compare \hm's task-based scheduler with BFS-style scheduling using $20$ threads.
The memory usages and the number of embeddings of the $20$ random queries in $q_3$ are illustrated in \reffig{mem}.  
As the number of embeddings increases, memory usage grows rapidly for BFS-style scheduling.
The results indicate for queries with many results, the memory usage of BFS-style scheduling is significantly larger than \hm's task-based scheduler because of the materialisation of all intermediate results.
This can lead to out-of-memory errors in machines with smaller memory capacities or when querying complex queries. 
However, \hm's task-based scheduler keeps the memory usage bounded with stable memory consumption of around $4.8$GB for all $20$ queries while achieving almost linear scalability, as demonstrated in the previous experiment.


\stitle{Exp-6: Load Balancing.}
We further evaluate the effectiveness of \hm's dynamic work stealing mechanism. Due to the space limit, we present the results of $q_3^2$ in Exp-4 executed using $20$ threads. The running time of each worker (i.e.thread) is shown in \reffig{work-stealing}. Time is reordered to sort in ascending order for ease of illustration.
We compare \hm with the load balancing technique of assigning the load by the firstly matched hyperedges (denoted as `HGMatch-NOSTL' in the figure).
%
%
%
However, when dynamic work stealing is not applied, we observe load differences among different workers, especially for the last worker.
%
On the other hand, when dynamic work stealing is applied, \hm achieves a near-perfect load balancing (the dashed line) with little overhead.



\subsection{\red{Case Study}} \label{sec:case_study}

\begin{figure}[h]
    \centering
    \begin{subfigure}[b]{.4\linewidth}
        \centering
        \includegraphics[width=0.8\linewidth]{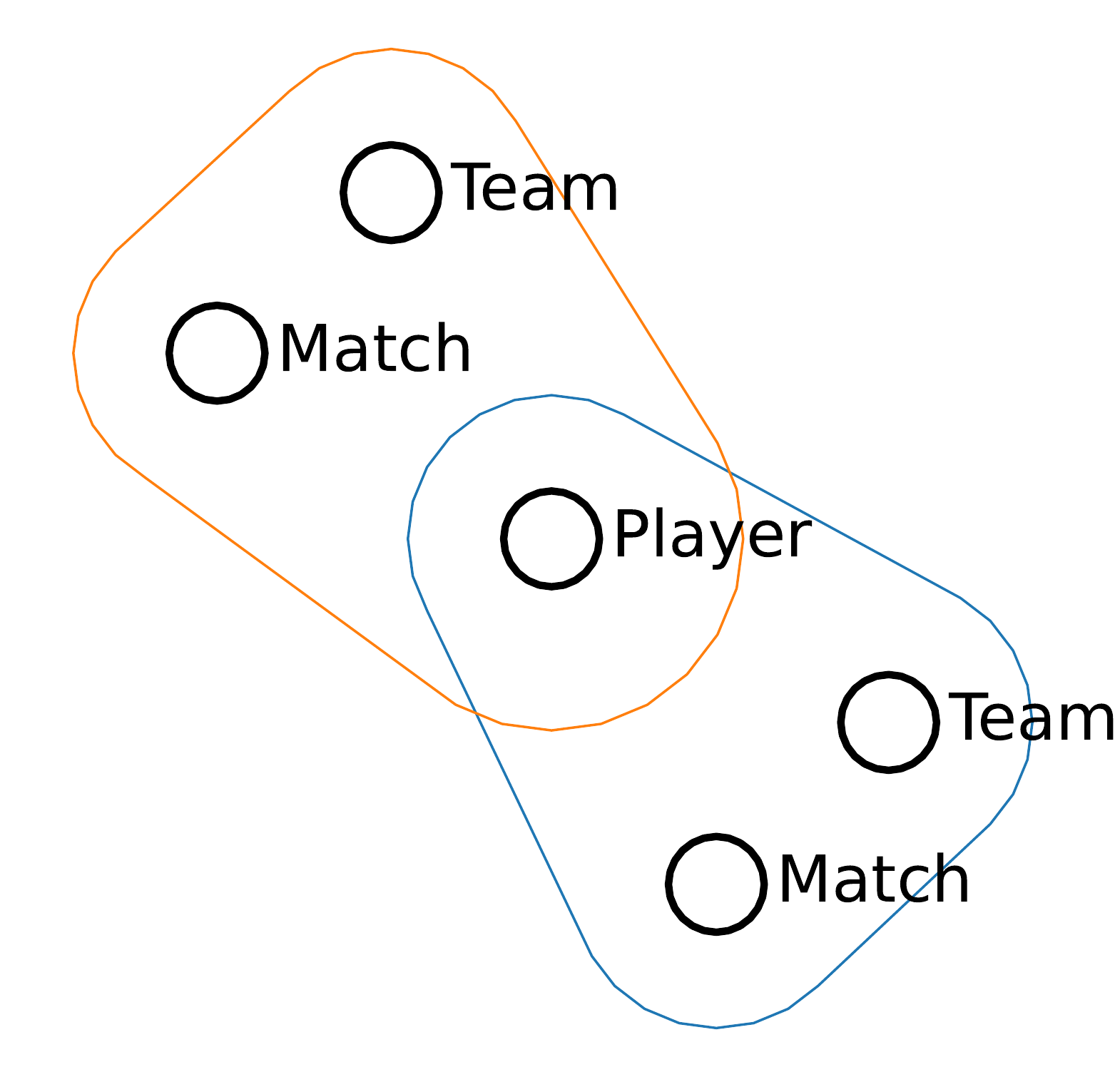}
        \caption{Example Query $1$}
        \label{fig:query1}
    \end{subfigure} %
    ~
    \begin{subfigure}[b]{.4\linewidth}
        \centering
        \includegraphics[width=0.8\linewidth]{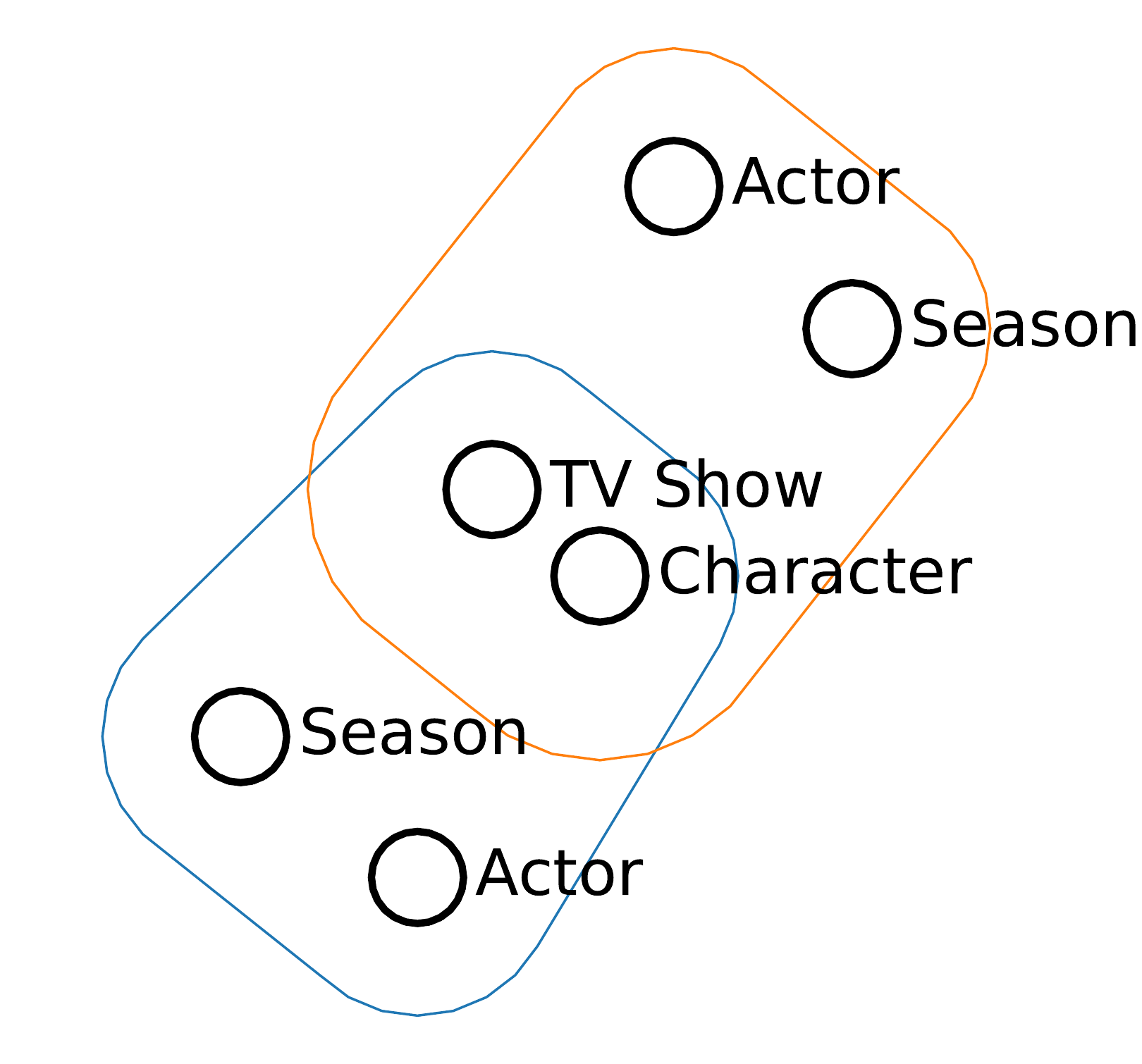}
        \caption{Example Query $2$}
        \label{fig:query2}
    \end{subfigure}
    \caption{\red{Example Queries in Hypergraph Knowledge Base.}}
    \label{fig:case_study}
\end{figure}

\red{
We demonstrate a case study of \sm on question answering over hypergraph knowledge base to help illustrate its applications.
We have conducted the case study on the \JF hypergraph knowledge base dataset \cite{Wen2016OnTR}, which is a small subset of non-binary relations extracted from the knowledge base Freebase \cite{Freebase}. The dataset is a hypergraph with the label for each vertex representing its type.
For example, the hyperedges with labels \textit{(Player, Team, Match)} indicate the fact that a football player played in a match representing a team. 
Another example of hyperedges with labels \textit{(Actor, Character, TV Show, Season)} means an actor played a character in the TV show on the season.
We present two example queries in \reffig{case_study} to answer real-world questions over the knowledge hypergraph by users with respect to the above example hyperedge types.
Query $1$ (\reffig{query1}) is to find the results of `Football players who represented different teams in different matches'. \hm finds $111$ embeddings of this query in the dataset. 
For instance, the football player Óscar Cardozo played for the Paraguay national football team in FIFA World Cup 2010, but played for the S.L. Benfica team in UEFA Europa League 2014.
Query $2$ (\reffig{query2}) is to find the results of `Actors who played the same character in a TV show on different seasons'. \hm finds $76$ embeddings of this query in the dataset. For instance, the actor Carlo Bonomi played the character Pingg in the TV show Pingu during seasons 1-4, and the same character was played by David Sant during seasons 5-6.
}

%% file: chapter/8.conclusion.tex
\section{Conclusion}\label{sec:conclusion}

We present \hm, an efficient parallel system for subhypergraph matching. 
Observed that the existing approaches with the match-by-vertex framework delay hyperedge verification and underutilise high-order information, we propose a novel match-by-hyperedge framework.
We store the data hypergraph in tables with a lightweight inverted hyperedge index and use set operations to efficiently generate hyperedge candidates.
%
%
A highly optimised parallel engine is developed, which adopts the dataflow model.
%
It features a task-based scheduler for bounded-memory execution and a dynamic work-stealing mechanism for load balancing. 
Experiments show that \hm significantly outperforms the baseline algorithms by orders of magnitude on average when using only a single thread and achieves almost linear scalability with near-perfect load balancing when the number of threads increases.

%% file: main.bbl
\begin{thebibliography}{10}

\bibitem{empty-headed}
C.~R. Aberger, S.~Tu, K.~Olukotun, and C.~R{\'e}.
\newblock Emptyheaded: A relational engine for graph processing.
\newblock In {\em Proceedings of the 2016 International Conference on
  Management of Data}, SIGMOD '16, pages 431--446, New York, NY, USA, 2016.
  ACM.

\bibitem{multiway-join}
F.~N. Afrati, D.~Fotakis, and J.~D. Ullman.
\newblock Enumerating subgraph instances using map-reduce.
\newblock In {\em Data Engineering (ICDE), 2013 IEEE 29th International
  Conference on}, pages 62--73. IEEE, 2013.

\bibitem{hyper-app-learn}
S.~Agarwal, K.~Branson, and S.~Belongie.
\newblock Higher order learning with graphs.
\newblock In {\em Proceedings of the 23rd International Conference on Machine
  Learning}, ICML '06, page 17–24, New York, NY, USA, 2006. Association for
  Computing Machinery.

\bibitem{dataflow}
T.~Akidau, R.~Bradshaw, C.~Chambers, S.~Chernyak, R.~J. Fernández-Moctezuma,
  R.~Lax, S.~McVeety, D.~Mills, F.~Perry, E.~Schmidt, and S.~Whittle.
\newblock The dataflow model: A practical approach to balancing correctness,
  latency, and cost in massive-scale, unbounded, out-of-order data processing.
\newblock {\em Proceedings of the VLDB Endowment}, 8:1792--1803, 2015.

\bibitem{wco-join}
K.~Ammar, F.~McSherry, S.~Salihoglu, and M.~Joglekar.
\newblock Distributed evaluation of subgraph queries using worst-case optimal
  low-memory dataflows.
\newblock {\em Proc. VLDB Endow.}, 11(6):691–704, Feb. 2018.

\bibitem{baader1999term}
F.~Baader and T.~Nipkow.
\newblock {\em Term rewriting and all that}.
\newblock Cambridge university press, 1999.

\bibitem{datasets}
A.~R. Benson.
\newblock Hypergraph datasets.
\newblock \url{https://www.cs.cornell.edu/~arb/data/}, 2022.

\bibitem{ceci}
B.~Bhattarai, H.~Liu, and H.~H. Huang.
\newblock Ceci: Compact embedding cluster index for scalable subgraph matching.
\newblock In {\em Proceedings of the 2019 International Conference on
  Management of Data}, SIGMOD ’19, page 1447–1462, New York, NY, USA, 2019.
  Association for Computing Machinery.

\bibitem{cfl}
F.~Bi, L.~Chang, X.~Lin, L.~Qin, and W.~Zhang.
\newblock Efficient subgraph matching by postponing cartesian products.
\newblock In {\em Proceedings of the 2016 International Conference on
  Management of Data}, SIGMOD ’16, page 1199–1214, New York, NY, USA, 2016.
  Association for Computing Machinery.

\bibitem{bitmap1}
T.~A. Bj\o{}rklund, N.~Grimsmo, J.~Gehrke, and O.~Torbj\o{}rnsen.
\newblock Inverted indexes vs. bitmap indexes in decision support systems.
\newblock In {\em Proceedings of the 18th ACM Conference on Information and
  Knowledge Management}, CIKM '09, page 1509–1512, New York, NY, USA, 2009.
  Association for Computing Machinery.

\bibitem{cilk}
R.~D. Blumofe, C.~F. Joerg, B.~C. Kuszmaul, C.~E. Leiserson, K.~H. Randall, and
  Y.~Zhou.
\newblock Cilk: An efficient multithreaded runtime system.
\newblock {\em SIGPLAN Not.}, 30(8):207–216, Aug. 1995.

\bibitem{work_stealing1}
R.~D. Blumofe and C.~E. Leiserson.
\newblock Scheduling multithreaded computations by work stealing.
\newblock {\em J. ACM}, 46(5):720–748, Sept. 1999.

\bibitem{Freebase}
K.~Bollacker, C.~Evans, P.~Paritosh, T.~Sturge, and J.~Taylor.
\newblock Freebase: A collaboratively created graph database for structuring
  human knowledge.
\newblock In {\em Proceedings of the 2008 ACM SIGMOD International Conference
  on Management of Data}, SIGMOD '08, page 1247–1250, New York, NY, USA,
  2008. Association for Computing Machinery.

\bibitem{img-hyper2}
A.~Bretto, H.~Cherifi, and D.~Aboutajdine.
\newblock Hypergraph imaging: an overview.
\newblock {\em Pattern Recognition}, 35(3):651--658, 2002.
\newblock Image/Video Communication.

\bibitem{sm05}
H.~Bunke, P.~Dickinson, and M.~Kraetzl.
\newblock Theoretical and algorithmic framework for hypergraph matching.
\newblock In {\em Proceedings of the 13th International Conference on Image
  Analysis and Processing}, ICIAP'05, page 463–470, Berlin, Heidelberg, 2005.
  Springer-Verlag.

\bibitem{sm08}
H.~Bunke, P.~Dickinson, M.~Kraetzl, M.~Neuhaus, and M.~Stettler.
\newblock {\em Matching of Hypergraphs --- Algorithms, Applications, and
  Experiments}, pages 131--154.
\newblock Springer Berlin Heidelberg, Berlin, Heidelberg, 2008.

\bibitem{Chase-Lev-deque}
D.~Chase and Y.~Lev.
\newblock Dynamic circular work-stealing deque.
\newblock In {\em Proceedings of the Seventeenth Annual ACM Symposium on
  Parallelism in Algorithms and Architectures}, SPAA '05, page 21–28, New
  York, NY, USA, 2005. Association for Computing Machinery.

\bibitem{img-re}
L.~Chen, Y.~Gao, Y.~Zhang, S.~Wang, and B.~Zheng.
\newblock Scalable hypergraph-based image retrieval and tagging system.
\newblock In {\em 2018 IEEE 34th International Conference on Data Engineering
  (ICDE)}, pages 257--268, 2018.

\bibitem{power-law-2}
F.~Chung, L.~Lu, and V.~Vu.
\newblock Spectra of random graphs with given expected degrees.
\newblock {\em Proceedings of the National Academy of Sciences},
  100(11):6313--6318, 2003.

\bibitem{power-law-1}
A.~Clauset, C.~R. Shalizi, and M.~E. Newman.
\newblock Power-law distributions in empirical data.
\newblock {\em SIAM review}, 51(4):661--703, 2009.

\bibitem{vf2}
L.~P. Cordella, P.~Foggia, C.~Sansone, and M.~Vento.
\newblock A (sub) graph isomorphism algorithm for matching large graphs.
\newblock {\em IEEE transactions on pattern analysis and machine intelligence},
  26(10):1367--1372, 2004.

\bibitem{ijcai2020p303}
B.~Fatemi, P.~Taslakian, D.~Vazquez, and D.~Poole.
\newblock Knowledge hypergraphs: Prediction beyond binary relations.
\newblock In C.~Bessiere, editor, {\em Proceedings of the Twenty-Ninth
  International Joint Conference on Artificial Intelligence, {IJCAI-20}}, pages
  2191--2197. International Joint Conferences on Artificial Intelligence
  Organization, 7 2020.
\newblock Main track.

\bibitem{bio-hyper-gene2}
S.~Feng, E.~Heath, B.~Jefferson, C.~Joslyn, H.~Kvinge, H.~D. Mitchell,
  B.~Praggastis, A.~J. Eisfeld, A.~C. Sims, L.~B. Thackray, S.~Fan, K.~B.
  Walters, P.~J. Halfmann, D.~Westhoff-Smith, Q.~Tan, V.~D. Menachery, T.~P.
  Sheahan, A.~S. Cockrell, J.~F. Kocher, K.~G. Stratton, N.~C. Heller, L.~M.
  Bramer, M.~S. Diamond, R.~S. Baric, K.~M. Waters, Y.~Kawaoka, J.~E.
  McDermott, and E.~Purvine.
\newblock Hypergraph models of biological networks to identify genes critical
  to pathogenic viral response.
\newblock {\em BMC Bioinformatics}, 22(1):287, May 2021.

\bibitem{flink}
Flink.
\newblock Apache flink.
\newblock \url{https://flink.apache.org/}, 2020.

\bibitem{hyper-app-learn4}
Y.~Gao, Z.~Zhang, H.~Lin, X.~Zhao, S.~Du, and C.~Zou.
\newblock Hypergraph learning: Methods and practices.
\newblock {\em IEEE Transactions on Pattern Analysis \& Machine Intelligence},
  44(05):2548--2566, may 2022.

\bibitem{np-complete}
M.~R. Garey and D.~S. Johnson.
\newblock {\em Computers and Intractability; A Guide to the Theory of
  NP-Completeness}.
\newblock W. H. Freeman \& Co., New York, NY, USA, 1979.

\bibitem{metagraph}
D.~Gaur, A.~Shastri, and R.~Biswas.
\newblock Metagraph: A new model of data structure.
\newblock In {\em 2008 International Conference on Computer Science and
  Information Technology}, pages 729--733, 2008.

\bibitem{AGI}
B.~Goertzel.
\newblock Patterns, hypergraphs and embodied general intelligence.
\newblock In {\em The 2006 IEEE International Joint Conference on Neural
  Network Proceedings}, pages 451--458, 2006.

\bibitem{work_stealing2}
A.~Grama and V.~Kumar.
\newblock {\em Load Balancing for Parallel Optimization Techniques}, pages
  1905--1911.
\newblock Springer US, Boston, MA, 2009.

\bibitem{hyper-iso}
T.~W. Ha, J.~H. Seo, and M.~H. Kim.
\newblock Efficient searching of subhypergraph isomorphism in hypergraph
  databases.
\newblock In {\em 2018 IEEE International Conference on Big Data and Smart
  Computing (BigComp)}, pages 739--742, 2018.

\bibitem{daf}
M.~Han, H.~Kim, G.~Gu, K.~Park, and W.-S. Han.
\newblock Efficient subgraph matching: Harmonizing dynamic programming,
  adaptive matching order, and failing set together.
\newblock In {\em Proceedings of the 2019 International Conference on
  Management of Data}, SIGMOD ’19, page 1429–1446, New York, NY, USA, 2019.
  Association for Computing Machinery.

\bibitem{qfilter}
S.~Han, L.~Zou, and J.~X. Yu.
\newblock Speeding up set intersections in graph algorithms using simd
  instructions.
\newblock In {\em Proceedings of the 2018 International Conference on
  Management of Data}, SIGMOD '18, page 1587–1602, New York, NY, USA, 2018.
  Association for Computing Machinery.

\bibitem{turbo-iso}
W.-S. Han, J.~Lee, and J.-H. Lee.
\newblock Turboiso: Towards ultrafast and robust subgraph isomorphism search in
  large graph databases.
\newblock In {\em Proceedings of the 2013 ACM SIGMOD International Conference
  on Management of Data}, SIGMOD ’13, page 337–348, New York, NY, USA,
  2013. Association for Computing Machinery.

\bibitem{coauthor-hyper}
Y.~Han, B.~Zhou, J.~Pei, and Y.~Jia.
\newblock {\em Understanding Importance of Collaborations in Co-authorship
  Networks: A Supportiveness Analysis Approach}, pages 1112--1123.
\newblock Society for Industrial and Applied Mathematics, 2009.

\bibitem{graph-ql}
H.~He and A.~K. Singh.
\newblock Graphs-at-a-time: Query language and access methods for graph
  databases.
\newblock In {\em Proceedings of the 2008 ACM SIGMOD International Conference
  on Management of Data}, SIGMOD ’08, page 405–418, New York, NY, USA,
  2008. Association for Computing Machinery.

\bibitem{hodges1997shorter}
W.~Hodges et~al.
\newblock {\em A shorter model theory}.
\newblock Cambridge university press, 1997.

\bibitem{qa-rdf}
S.~{Hu}, L.~{Zou}, J.~X. {Yu}, H.~{Wang}, and D.~{Zhao}.
\newblock Answering natural language questions by subgraph matching over
  knowledge graphs.
\newblock {\em IEEE Transactions on Knowledge and Data Engineering},
  30(5):824--837, 2018.

\bibitem{hyper-app-text}
T.~Hu, H.~Xiong, W.~Zhou, S.~Y. Sung, and H.~Luo.
\newblock Hypergraph partitioning for document clustering: A unified clique
  perspective.
\newblock In {\em Proceedings of the 31st Annual International ACM SIGIR
  Conference on Research and Development in Information Retrieval}, SIGIR '08,
  page 871–872, New York, NY, USA, 2008. Association for Computing Machinery.

\bibitem{hyperx}
J.~Huang, R.~Zhang, and J.~X. Yu.
\newblock Scalable hypergraph learning and processing.
\newblock In {\em 2015 IEEE International Conference on Data Mining}, pages
  775--780, 2015.

\bibitem{bio-hyper-gene1}
T.~Hwang, Z.~Tian, R.~Kuangy, and J.-P. Kocher.
\newblock Learning on weighted hypergraphs to integrate protein interactions
  and gene expressions for cancer outcome prediction.
\newblock In {\em 2008 Eighth IEEE International Conference on Data Mining},
  pages 293--302, 2008.

\bibitem{simd-intersection}
H.~Inoue, M.~Ohara, and K.~Taura.
\newblock Faster set intersection with simd instructions by reducing branch
  mispredictions.
\newblock {\em Proc. VLDB Endow.}, 8(3):293–304, nov 2014.

\bibitem{HyperGraphDB}
B.~Iordanov.
\newblock Hypergraphdb: A generalized graph database.
\newblock In {\em Proceedings of the 2010 International Conference on Web-Age
  Information Management}, WAIM'10, page 25–36, Berlin, Heidelberg, 2010.
  Springer-Verlag.

\bibitem{img-hyper3}
P.~Jian, K.~Chen, and C.~Zhang.
\newblock A hypergraph-based context-sensitive representation technique for vhr
  remote-sensing image change detection.
\newblock {\em International Journal of Remote Sensing}, 37(8):1814--1825,
  2016.

\bibitem{jin2021fast}
X.~Jin, Z.~Yang, X.~Lin, S.~Yang, L.~Qin, and Y.~Peng.
\newblock Fast: Fpga-based subgraph matching on massive graphs.
\newblock In {\em 2021 IEEE 37th International Conference on Data Engineering
  (ICDE)}. IEEE, 2021.

\bibitem{graphflow-demo}
C.~Kankanamge, S.~Sahu, A.~Mhedbhi, J.~Chen, and S.~Salihoglu.
\newblock Graphflow: An active graph database.
\newblock In {\em Proceedings of the 2017 ACM International Conference on
  Management of Data}, SIGMOD ’17, page 1695–1698, New York, NY, USA, 2017.
  Association for Computing Machinery.

\bibitem{dataflow-def}
Kavi, Buckles, and Bhat.
\newblock A formal definition of data flow graph models.
\newblock {\em IEEE Transactions on Computers}, C-35(11):940--948, 1986.

\bibitem{RI}
R.~{Kimmig}, H.~{Meyerhenke}, and D.~{Strash}.
\newblock Shared memory parallel subgraph enumeration.
\newblock In {\em 2017 IEEE International Parallel and Distributed Processing
  Symposium Workshops (IPDPSW)}, pages 519--529, 2017.

\bibitem{bio-hyper-ppi2}
S.~Klamt, U.-U. Haus, and F.~Theis.
\newblock Hypergraphs and cellular networks.
\newblock {\em PLOS Computational Biology}, 5(5):1--6, 05 2009.

\bibitem{parsing_hyper}
D.~Klein and C.~D. Manning.
\newblock {\em Parsing and Hypergraphs}, pages 351--372.
\newblock Springer Netherlands, Dordrecht, 2004.

\bibitem{twin-twig}
L.~Lai, L.~Qin, X.~Lin, and L.~Chang.
\newblock Scalable subgraph enumeration in mapreduce.
\newblock {\em Proc. VLDB Endow.}, 8(10):974–985, June 2015.

\bibitem{seed}
L.~Lai, L.~Qin, X.~Lin, Y.~Zhang, L.~Chang, and S.~Yang.
\newblock Scalable distributed subgraph enumeration.
\newblock {\em Proc. VLDB Endow.}, 10(3):217–228, Nov. 2016.

\bibitem{patmat-exp}
L.~Lai, Z.~Qing, Z.~Yang, X.~Jin, Z.~Lai, R.~Wang, K.~Hao, X.~Lin, L.~Qin,
  W.~Zhang, Y.~Zhang, Z.~Qian, and J.~Zhou.
\newblock Distributed subgraph matching on timely dataflow.
\newblock {\em Proc. VLDB Endow.}, 12(10):1099–1112, June 2019.

\bibitem{comparison}
J.~Lee, W.-S. Han, R.~Kasperovics, and J.-H. Lee.
\newblock An in-depth comparison of subgraph isomorphism algorithms in graph
  databases.
\newblock {\em Proc. VLDB Endow.}, 6(2):133–144, Dec. 2012.

\bibitem{simd_gallop}
D.~Lemire, L.~Boytsov, and N.~Kurz.
\newblock Simd compression and the intersection of sorted integers.
\newblock {\em Softw. Pract. Exper.}, 46(6):723–749, jun 2016.

\bibitem{img-hyper1}
X.~Li, Y.~Li, C.~Shen, A.~Dick, and A.~V.~D. Hengel.
\newblock Contextual hypergraph modeling for salient object detection.
\newblock In {\em 2013 IEEE International Conference on Computer Vision}, pages
  3328--3335, 2013.

\bibitem{cost}
F.~McSherry, M.~Isard, and D.~G. Murray.
\newblock Scalability! but at what {COST}?
\newblock In {\em 15th Workshop on Hot Topics in Operating Systems (HotOS
  {XV})}, Kartause Ittingen, Switzerland, May 2015. {USENIX} Association.

\bibitem{hyper-app-nlp}
T.~Menezes and C.~Roth.
\newblock Semantic hypergraphs, 2019.

\bibitem{grakn}
A.~Messina, H.~Pribadi, J.~Stichbury, M.~Bucci, S.~Klarman, and A.~Urso.
\newblock Biograkn: A knowledge graph-based semantic database for biomedical
  sciences.
\newblock In L.~Barolli and O.~Terzo, editors, {\em Complex, Intelligent, and
  Software Intensive Systems}, pages 299--309, Cham, 2018. Springer
  International Publishing.

\bibitem{graphflow}
A.~Mhedhbi and S.~Salihoglu.
\newblock Optimizing subgraph queries by combining binary and worst-case
  optimal joins.
\newblock {\em Proc. VLDB Endow.}, 12(11):1692–1704, July 2019.

\bibitem{AtomSpace}
{OpenCog Foundation}.
\newblock Atomspace.
\newblock \url{https://wiki.opencog.org/w/AtomSpace}, 2022.

\bibitem{opencog}
{OpenCog Foundation}.
\newblock The open cognition project.
\newblock \url{https://wiki.opencog.org/w/The_Open_Cognition_Project}, 2022.

\bibitem{crystaljoin}
M.~Qiao, H.~Zhang, and H.~Cheng.
\newblock Subgraph matching: On compression and computation.
\newblock {\em Proc. VLDB Endow.}, 11(2):176--188, Oct. 2017.

\bibitem{bio-hyper-ppi1}
E.~Ramadan, A.~Tarafdar, and A.~Pothen.
\newblock A hypergraph model for the yeast protein complex network.
\newblock In {\em 18th International Parallel and Distributed Processing
  Symposium, 2004. Proceedings.}, pages 189--, 2004.

\bibitem{PGX.ISO}
R.~Raman, O.~van Rest, S.~Hong, Z.~Wu, H.~Chafi, and J.~Banerjee.
\newblock Pgx.iso: Parallel and efficient in-memory engine for subgraph
  isomorphism.
\newblock In {\em Proceedings of Workshop on GRAph Data Management Experiences
  and Systems}, GRADES’14, page 1–6, New York, NY, USA, 2014. Association
  for Computing Machinery.

\bibitem{RADS}
X.~Ren, J.~Wang, W.-S. Han, and J.~X. Yu.
\newblock Fast and robust distributed subgraph enumeration.
\newblock {\em Proc. VLDB Endow.}, 12(11):1344--1356, July 2019.

\bibitem{quicksi}
H.~Shang, Y.~Zhang, X.~Lin, and J.~X. Yu.
\newblock Taming verification hardness: An efficient algorithm for testing
  subgraph isomorphism.
\newblock {\em Proc. VLDB Endow.}, 1(1):364–375, Aug. 2008.

\bibitem{psgl}
Y.~Shao, B.~Cui, L.~Chen, L.~Ma, J.~Yao, and N.~Xu.
\newblock Parallel subgraph listing in a large-scale graph.
\newblock In {\em Proceedings of the 2014 ACM SIGMOD International Conference
  on Management of Data}, SIGMOD ’14, page 625–636, New York, NY, USA,
  2014. Association for Computing Machinery.

\bibitem{subhymatch}
Y.~Su, Y.~Gu, Z.~Wang, Y.~Zhang, J.~Qin, and G.~Yu.
\newblock Efficient subhypergraph matching based on hyperedge features.
\newblock {\em IEEE Transactions on Knowledge and Data Engineering}, pages
  1--1, 2022.

\bibitem{PSM}
S.~Sun and Q.~Luo.
\newblock Parallelizing recursive backtracking based subgraph matching on a
  single machine.
\newblock In {\em 2018 IEEE 24th International Conference on Parallel and
  Distributed Systems (ICPADS)}, pages 1--9, 2018.

\bibitem{match-survey}
S.~Sun and Q.~Luo.
\newblock In-memory subgraph matching: An in-depth study.
\newblock In {\em Proceedings of the 2020 ACM SIGMOD International Conference
  on Management of Data}, SIGMOD ’20, page 1083–1098, New York, NY, USA,
  2020. Association for Computing Machinery.

\bibitem{rapidmatch}
S.~Sun, X.~Sun, Y.~Che, Q.~Luo, and B.~He.
\newblock Rapidmatch: A holistic approach to subgraph query processing.
\newblock {\em Proc. VLDB Endow.}, 14(2):176–188, oct 2020.

\bibitem{star-join}
Z.~Sun, H.~Wang, H.~Wang, B.~Shao, and J.~Li.
\newblock Efficient subgraph matching on billion node graphs.
\newblock {\em Proc. VLDB Endow.}, 5(9):788–799, May 2012.

\bibitem{hyper-app-social}
S.~Tan, Z.~Guan, D.~Cai, X.~Qin, J.~Bu, and C.~Chen.
\newblock Mapping users across networks by manifold alignment on hypergraph.
\newblock In {\em Proceedings of the Twenty-Eighth AAAI Conference on
  Artificial Intelligence}, AAAI'14, page 159–165. AAAI Press, 2014.

\bibitem{bio-hyper-gene3}
L.~Tran.
\newblock Hypergraph and protein function prediction with gene expression data,
  2012.

\bibitem{ullmann}
J.~R. Ullmann.
\newblock An algorithm for subgraph isomorphism.
\newblock {\em J. ACM}, 23(1):31–42, Jan. 1976.

\bibitem{TypeDB}
{Vaticle Ltd}.
\newblock Typedb.
\newblock \url{https://github.com/vaticle/typedb}, 2022.

\bibitem{opencog-graph}
L.~Vepštas.
\newblock Graphs, metagraphs, ram, cpu.
\newblock
  \url{https://github.com/opencog/atomspace/blob/master/opencog/sheaf/docs/ram-cpu.pdf},
  2022.

\bibitem{bitmap2}
J.~Wang, C.~Lin, Y.~Papakonstantinou, and S.~Swanson.
\newblock An experimental study of bitmap compression vs. inverted list
  compression.
\newblock In {\em Proceedings of the 2017 ACM International Conference on
  Management of Data}, SIGMOD '17, page 993–1008, New York, NY, USA, 2017.
  Association for Computing Machinery.

\bibitem{hyper-app-ml}
Y.~Wang, P.~Li, and C.~Yao.
\newblock Hypergraph canonical correlation analysis for multi-label
  classification.
\newblock {\em Signal Processing}, 105:258--267, 2014.

\bibitem{BENU}
Z.~Wang, R.~Gu, W.~Hu, C.~Yuan, and Y.~Huang.
\newblock Benu: Distributed subgraph enumeration with backtracking-based
  framework.
\newblock In {\em 2019 IEEE 35th International Conference on Data Engineering
  (ICDE)}, pages 136--147. IEEE, 2019.

\bibitem{watkin2017introduction}
J.~L. Watkin.
\newblock An introduction to the clips programming language.
\newblock 2017.

\bibitem{Wen2016OnTR}
J.~Wen, J.~Li, Y.~Mao, S.~Chen, and R.~Zhang.
\newblock On the representation and embedding of knowledge bases beyond binary
  relations.
\newblock In {\em International Joint Conference on Artificial Intelligence},
  2016.

\bibitem{communities-graph}
J.~Yang and J.~Leskovec.
\newblock Defining and evaluating network communities based on ground-truth.
\newblock In {\em Proceedings of the ACM SIGKDD Workshop on Mining Data
  Semantics}, MDS '12, New York, NY, USA, 2012. Association for Computing
  Machinery.

\bibitem{huge}
Z.~Yang, L.~Lai, X.~Lin, K.~Hao, and W.~Zhang.
\newblock {\em HUGE: An Efficient and Scalable Subgraph Enumeration System},
  page 2049–2062.
\newblock Association for Computing Machinery, New York, NY, USA, 2021.

\bibitem{Yu2016}
X.~Yu and T.~Korkmaz.
\newblock Hypergraph querying using structural indexing and
  layer-related-closure verification.
\newblock {\em Knowledge and Information Systems}, 46(3):537--565, Mar 2016.

\bibitem{hyper-app-learn3}
D.~Zhou, J.~Huang, and B.~Sch\"{o}lkopf.
\newblock Learning with hypergraphs: Clustering, classification, and embedding.
\newblock In B.~Sch\"{o}lkopf, J.~Platt, and T.~Hoffman, editors, {\em Advances
  in Neural Information Processing Systems}, volume~19. MIT Press, 2006.

\end{thebibliography}
